\documentclass[11pt]{article} 
\usepackage[letterpaper,margin=1in]{geometry} 

\usepackage{setspace}
\usepackage[english]{babel}

\usepackage{amsthm} 
\usepackage{mathrsfs}
\usepackage{amssymb}
\usepackage{thmtools}      
\usepackage{amsmath}
\usepackage{cancel}
\usepackage[all, pdf]{xy}
\usepackage{geometry}
\usepackage{graphicx}
\usepackage{tikz}
\usetikzlibrary{arrows.meta, positioning, fit, calc}
\usetikzlibrary{positioning,arrows.meta,decorations.pathmorphing,shapes.geometric,shapes.arrows}
\usepackage{adjustbox} 
\usepackage{caption} 
\usepackage{diagbox} 
\usepackage{pgf}
\usepackage{pgfplots}
\usepackage{comment}
\usepackage{xcolor}
\usepackage[bookmarks=true]{hyperref} 
\usepackage{bookmark}

\theoremstyle{definition} 

\sffamily
\newtheorem{thm}{Theorem}[section]

\newtheorem{cor}[thm]{Corollary}
\newtheorem{lem}[thm]{Lemma}
\newtheorem{prop}[thm]{Proposition}
\newtheorem*{prop*}{Proposition} 
\newtheorem{defn}[thm]{Definition}
\newtheorem{fact}[thm]{Fact}

\def\B#1{\mathbb #1}
\def\C#1{\mathcal #1}

\def\san#1{\textsf{#1}}
\def\F#1{\textbf{#1}}

\def\Pow{\mathscr{P}}
\def\RT{\text{RTC}}
\def\TC{\text{TC}}
\newcommand{\bisim}{\mathrel{\raisebox{0.2ex}{$\underline{\leftrightarrow}$}}}

\def\PP{\Pow^+}
\def\PPX{\Pow^+(\C X)}
\def\PPA{\Pow^+(\C A)}

\def\E[#1]{\exists^{#1}}
\def\A[#1]{\forall^{#1}}

\newcommand{\circsurd}{\tikz[baseline=(char.base)]{
  \node[draw,circle,inner sep=0.4pt,line width=0.3pt] (char) {$\surd$};
}} 

\usepackage{layout}      
\usepackage{printlen}    
\uselengthunit{in}

\begin{document}


\begin{titlepage}
  \thispagestyle{empty} 

  \vspace*{1.6cm}

  \begin{center}
    {\LARGE\bfseries Uniform Interpolation in Distributed Knowledge Modal Logics}\\[1.5em]

    {\large
      Kexu Wang\\
      \textit{Jinan University, Guangzhou, China}\\
      \texttt{wangkexuphy@gmail.com}
    }\\[1em]

    {\large
      Liangda Fang\\
      \textit{Jinan University, Guangzhou, China}\\
      \texttt{fangld@jnu.edu.cn}
    }

    \vspace{1.5em}

    \begin{minipage}{0.86\textwidth}
      \small
      \textbf{Abstract.} 
Uniform interpolation is the property that, for any formula and set of atoms, there exists the strongest consequence omitting those atoms. It plays a central role in knowledge representation and reasoning tasks such as knowledge update and information hiding. This paper studies the uniform interpolation property in epistemic modal logics with \emph{distributed knowledge}, which captures agents’ collective reasoning abilities. Building on the bisimulation-quantifier perspective, we extend the canonical-formula and literal-elimination framework of Fang, Liu, and van Ditmarsch to distributed knowledge settings and introduce the concept of collective $p$-bisimulation. We show that, for distributed knowledge modal logics $\textsf{K}_n\textbf{D}$, $\textsf{D}_n\textbf{D}$, and $\textsf{T}_n\textbf{D}$, every satisfiable canonical formula's uniform interpolant omitting an atom $p$ is exactly its remainder of eliminating $p$. Then, we provide a finer analysis for the transitive and Euclidean systems $\textsf{K45}_n\textbf{D}$, $\textsf{KD45}_n\textbf{D}$, and $\textsf{S5}_n\textbf{D}$, and prove that every formula of modal depth $k + 1$ has a uniform interpolant of modal depth $2 k + 1$. Thus, we prove the uniform interpolation property in all the six distributed knowledge modal logics. Finally, we generalize the results to some variants with propositional common knowledge and discuss the method's limitations.
%
    \end{minipage}

  \end{center}

  \vfill
\end{titlepage}

%

\newpage
\setcounter{page}{0}

\pdfbookmark{\contentsname}{}
\tableofcontents

\newpage
\setcounter{page}{1}
\section{Introduction}
A logic has the \emph{uniform interpolation} property if for every formula and every set of atoms, there exists a formula that omits those atoms while preserving exactly the same consequences over the remaining vocabulary. This property is closely tied to propositional quantifier elimination and forgetting in knowledge representation \cite{eiter2019brief, QIU2024104077}, and it underpins knowledge update \cite{Liu2011Progression, feng2023knowledge}, information hiding \cite{Halpern2004Anonymity}, and ontology reuses \cite{Koopmann2015Ontology}. Uniform interpolation is well understood in propositional logic and multi-agent epistemic modal logics. However, its status in the logics with \emph{distributed knowledge}—which allow reasoning about what a set of agents collectively knows—has remained unclear. In this paper, we investigate uniform interpolation in epistemic modal logics with distributed knowledge.

Uniform interpolation originates in Henkin’s work \cite{Henkin1963} as a strengthened version of Craig’s interpolation. While a Craig’s interpolant depends on both an antecedent and a consequent, a uniform interpolant depends only on the antecedent and the set of atoms to retain, so uniform interpolation implies Craig’s interpolation. Over the decades, research on uniform interpolation has followed various lines. Researchers have adopted proof-theoretic or algebraic techniques, seeking constructive procedures or algebraic conditions to ensure the uniform interpolation property. Another influential line of research seeks to ensure uniform interpolation through bisimulation quantifiers. It was independently initiated by Ghilardi and Zawadowski \cite{ghilardi1995undefinability} and by Visser \cite{visser1996uniform}. The intuitive idea is an extension of the classical bisimulation: Two models are $p$-bisimilar if they are bisimilar except possibly on an atom $p$. A logic is $p$-bisimulation invariant if every two $p$-bisimilar models satisfy the same formulas of this logic without $p$ occurring. For every formula $\phi$, we intend to find a formula $\psi$ without $p$ occurring such that a model $M'$ satisfies $\psi$ if and only if, there exists a model $M$ of $\phi$ such that $M$ and $M'$ are $p$--bisimilar. This $\psi$ is called a result of forgetting $p$. Observe that $\psi$ exactly serves as a uniform interpolant of $\phi$, and one may write $\psi$ in a quantified form to emphasize that this $p$-bisimulation behaves like an existential quantifier. This bisimulation-quantifier line of research is also what our paper follows.

Historically, all these research lines, proof-theoretic, algebraic, or bisimulation-quantifier, have offered complementary insights and jointly shaped today’s understanding of uniform interpolation. Pitts \cite{Pitts1992} gave the first syntactic proof for intuitionistic propositional logic, followed by Shavrukov’s result for G{\"o}del-L{\"o}b logic (\san{GL}) \cite{shavrukov1993subalgebras}. Ghilardi \cite{GHILARDI1995189} provided an algebraic proof for the basic modal logic $\mathsf K$. Visser gave the bisimulation-quantifier proofs for both $\san K$ and \san{GL} \cite{visser1996uniform}. While alongside Pitt's research line, B{\'i}lkov{\'a} \cite{Bilkova2007Uniform} provided a proof-theoretic proof for uniform interpolation in $\san{K}$ and $\san{T}$. In contrast, $\san{K4}$ and $\san{S4}$ are known to have no uniform interpolation property \cite{ghilardi1995undefinability, Bilkova2007Uniform}. It is noteworthy that $\san{K4}$ and $\san{S4}$ are classic examples that enjoy Craig's interpolation but have no uniform interpolation property. Wolter \cite{Wolter96} showed uniform interpolation in \san{S5} and proposed a syntactic method, the fusion of modal logics, which preserves the uniform interpolation property, and it lifts the property from single-agent to multi-agent settings. Pattinson \cite{pattinson2013logic} proved that any multi-agent rank-1 modal logics (namely, axiomatized by depth-one formulas, such as $\mathsf{K}_n$ and $\mathsf{D}_n$) have uniform interpolation. Studer \cite{STUDER2009611} showed the failure of uniform interpolation in the logics with common knowledge, $\san{K}_n \F{C}$ and $\san{K4}_n \F{C}$, by disproving their Beth property. D'Agostino and Hollenberg showed that $\mu$-calculus, an extension of $\san K$ with a fixed point operator, has uniform interpolation \cite{d2000logical}. Building on the disjunctive-normal-form technology of Janin and Walukiewicz \cite{Janin1995Aut}, D'Agostino and Lenzi \cite{DAGOSTINO2006256} first proposed the method of literal elimination and gave explicit constructions of uniform interpolants for disjunctive $\mu$-calculus. Moss \cite{moss2007finite} developed the notion of \emph{canonical formulas}, and showed that every formula can be represented as a disjunction of a unique set of canonical formulas. Then, Fang, Liu, and van Ditmarsh \cite{fang2019forgetting} proposed the literal elmination on canonical formulas and showed comprehensive bisimulation-quantifier conclusions for multi-agent modal logics $\san{K}_n$, $\san{D}_n$, $\san{T}_n$, $\san{K45}_n$, $\san{KD45}_n$, and $ \san{S5}_n$, as well as their propositional common knowledge versions. In parallel, the non-bisimulation-quantifier lines are also fruitful: Seifan, Schr{\"o}der, and Pattinson \cite{Seifan2017Coalgebraic} proposed the condition for uniform interpolation in coalgebraic modal logics, namely, rank-1 modal logics and provided new results more than $\san{K}$ and $\san{D}$. Iemhoff \cite{Iemhoff2019Uniform} linked the uniform interpolation property to the existence of sequent calculi. van der Giessen, Jalali, and Kuznets \cite{vanderGiessen2025Nested} presented a new proof-theoretic result for $\mathsf{K}$, $\mathsf{D}$, $\mathsf{T}$, and $\mathsf{S5}$ by generalized sequent caculi. \cite{Feree2024Coq} realized the uniform interpolation proofs for  $\san{K}, \san{GL}, \san{iSL}$ via the interactive theorem prover Coq. A stronger variant of uniform interpolation, uniform Lyndon interpolation, has been explored recently in intuitionistic logics \cite{Tabatabai2022iM, Tabatabai2024Conditional} and in the extensions of $\san{K5}$ \cite{Giessen2023K5} using proof-theoretic methods. Despite this rich development, an important notion in multi-agent modal logic has not been delved into, that is, distributed knowledge.

Distributed knowledge characterizes the agents' cooperation in multi-agent modal logic. A new modality $\F D_{\C B}$ is introduced for every nonempty set $\C B$ of agents, which characterizes ``the agents in $\C B$ collectively know that ...'', and the intersection of the relations of agents in $\C B$ acts as its accessibility relation. This operator is monotone, that is, $\F{D}_{\C B} \phi \to \F{D}_{\C C} \phi$ if $\C C \subseteq \C B$. The notion of distributed knowledge dates back to Halpern and Moses’s work on distributed systems \cite{Halpern1990}. Then, \cite{Fagin1995reason} provided the conclusions of soundness and completeness for distributed knowledge modal logics. However, the classical bisimulation no longer works in the scenario of distributed knowledge. Roelofsen \cite{roelofsen2007distributed} introduced the concept of \emph{collective bisimulation} that characterizes the comparison of models and the expressiveness of distributed knowledge. Then, researchers have also studied how distributed knowledge interacts with the other agent operators. W{\'a}ng and {\AA}gotnes \cite{Wang2011Public} extended public announcement logic with distributed knowledge. Then, {\AA}gotnes and W{\'a}ng \cite{AGOTNES20171} further defined resolution operators and axiomatized logics combining distributed knowledge and common knowledge. Then, {\AA}gotnes, Alechina, and Galimullin \cite{aagotnes2022logics} studied group-announcement operators interacting with distributed knowledge. Apart from the Kripke semantics, \cite{Ditmarsch2022Knowledge} and \cite{Goubault2023Semi} introduced simplicial semantics for epistemic logic with distributed knowledge and provided their axiomatization. The work  \cite{castaneda2023communication} has further explored distributed knowledge through dynamic and topological perspectives.
On the applied side, distributed knowledge has become a key tool in multi-agent systems: it is used in verification tools such as MCMAS \cite{lomuscio2017mcmas} in distributed decision making. Murai and Sano have established their proof-theoretic proof for Craig interpolation in epistemic logics $\san{K}_n$, $\san{T}_n$, $\san{K4}_n$, and $\san{S4}_n$ with distributed knowledge \cite{Murai2020CraigDistr}, and intuitionistic epistemic logics extended with distributed knowledge are also developed \cite{su2021artemov, murai2022intuitionistic}. However, the work of Murai and Sano leaves the cases of the distributed knowledge $\san{K45}_n$, $\san{KD45}_n$, $\san{S5}_n$ open. As a strengthened property, uniform interpolation in distributed knowledge modal logics has also remained an open problem — a gap this paper aims to address.

This paper's most closely related work is by Fang, Liu, and van Ditmarsch \cite{fang2019forgetting}, who, alongside the bisimulation-quantifier line, proved the closure under forgetting, and hence, the uniform interpolation property, for a broad family of modal systems $\san{K}_n$, $\san{D}_n$, $\san{T}_n$, $\san{K45}_n$, $\san{KD45}_n$, and $\san{S5}_n$. Since every formula can be represented as the disjunction of a unique set of canonical formulas, and forgetting is closed under disjunction, they reduce the problem to finding the canonical formulas' results of forgetting (uniform interpolants). They proposed the literal elimination of canonical formulas, and through the bismulation-based model construction, they proved that for each canonical formula $\delta$, its remainder $\delta^p$ of eliminating a literal $p$ is exactly a result of forgetting $p$ (uniform interpolant). Thus, uniform interpolation is proved. We extend these techniques to the distributed knowledge settings, and then obtain the analogous uniform interpolation results for the corresponding distributed knowledge extensions: $\san{K}_n\F D$, $\san{D}_n\F D$, and $\san{T}_n\F D$. However, the distributed knowledge extensions $\san{K45}_n \F D$, $\san{KD45}_n \F D$, and $\san{S5}_n \F D$ are more complicated. Distributed knowledge operators require the inclusion relations between sets of agents, and the models of the latter three logics are both transitive and Euclidean, which makes the bisimulation-based model construction rather challenging. We show some canonical formulas in the latter three logics whose uniform interpolants are \emph{inequivalent} to their remainders of literal elimination. Thence, we refine the bisimulation-based model construction. The key notion is to combine the model construction with operating on canonical formulas' modal depth (pruning). We prove that for every canonical formula $\delta$ of depth $k + 1$, its uniform interpolant is constructed as the disjunction of all the remainders of those $(2k + 1)$-depth canonical formulas that entail $\delta$. This means every formula of depth $k + 1$ has a uniform interpolant of depth at least $2k + 1$. Hence, uniform interpolation in $\san{K45}_n \F D$, $\san{KD45}_n \F D$, and $\san{S5}_n \F D$ is proved. The contributions of our paper are summarized as follows:
\begin{itemize}
\item We introduce collective $p$-bisimulation as a combination of collective bisimulation and bisimulation quantifiers, and show its adequacy lemma, which means these logics that we discuss are collective-$p$-bisimulation invariant. (Subsection \ref{subsec: Bisimulation})

\item We introduce distributed knowledge canonical formulas and prove each formula is equivalent to the disjunction of a unique set of such canonical formulas. We also extend the definitions of pruning and literal elimination to the case of distributed knowledge canonical formulas. (Subsection \ref{subsec: CanonForm} \& \ref{subsec: LiteralElimination})

\item We show that in $\san{K}_n\F D$, $\san{D}_n\F D$,  and $\san{T}_n\F D$, every canonical formula's uniform interpolant is equivalent to its remainder of literal elimination (Lemmas \ref{lem: KandD} \& \ref{lem: T}), and thus prove the uniform interpolation property in the three logics (Theorem \ref{thm: KDT}).

\item We prove the uniform interpolation property in $\san{K45}_n \F D$, $\san{KD45}_n \F D$ (Theorem \ref{thm: K45KD45}), and $\san{S5}_n \F D$ (Theorem \ref{thm: K45KD45}) by combining bisimulation-based model construction with pruning (Lemmas \ref{lem: ConstructionK45} \& \ref{lem: ConstructionS5}) and constructing uniform interpolants for distributed canonical formulas (Lemmas \ref{lem: K45KD45} \& \ref{lem: S5}). We also analyze cases where the uniform interpolant is equivalent (Subsection \ref{subsec: SpecialCase}) or inequivalent (Subsection \ref{subsec: Counterexample}) to the remainder of literal elimination.

\item We extend the techniques to the versions augmented with propositional common knowledge, however, except for $\san{K45}_n \F{DPC}$ and $\san{KD45}_n \F{DPC}$ (Section \ref{sec: CommonKnowledge}).

\end{itemize}
The paper is organized as follows. Section \ref{sec: Preliminaries} gives preliminaries. Section \ref{sec: SyntacticMethods} defines collective $p$-bisimulation and provides syntactic methods, including  distributed knowledge canonical formulas, literal elimination, etc. Section \ref{sec: KDT} proves the uniform interpolation property in $\san{K}_n \F D$, $\san{D}_n \F D$, and $\san{T}_n \F D$. Section \ref{sec: K45KD45S5} proves the uniform interpolation property in $\san{K45}_n \F D$, $\san{KD45}_n \F D$, and $\san{S5}_n \F D$. Section \ref{sec: CommonKnowledge} discusses the results about the propositional common knowledge versions. Section \ref{sec: Conclusion} concludes this paper. Full proofs appear in the appendix.

\section{Preliminaries}\label{sec: Preliminaries}
Let's fix $\C A$ to be the set of $n$ agents. The lowercase letter $i, j$ range over agents while the calligraphic letters $\C B, \C C, \C D, \C X$ range over nonempty subsets of $\C A$. The function $\PP(\cdot)$ indicates the nonempty subset of a given set, for example, $\PPA$ as all the nonempty subsets of $\C A$. For every set of agents $\C B$, its complement $\C B^c = \C A \setminus \C B$. 

Let's fix the calligraphic letter $\C P$ to be the countable set of all the atoms (or propositional letters). A \emph{literal} is an atom $p$ (positive literal) or its negation $\neg p$ (negative literal). Given a finite subset $P \subseteq \C P$, a \emph{minterm} of $P$ is a conjunction of literals that uses only the members of $P$ and where each atom in $P$ occurs exactly once.

\begin{defn}
The language $\C L^{\F D}_{\F C}$ is generated by BNF:
\[\phi :: = p \mid \neg \phi \mid \phi \land \phi \mid \phi \lor \phi \mid \F K_i \phi \mid \F D_{\C B} \phi \mid \F{C} \phi\]
for every $p \in \C P$, $i \in \C A$, and $\C B \in \PPA$. 
If for every formula in the form $\F C\phi$, $\phi$ is restricted to be propositional, we denote the language by $\C L^{\F D}_\F{PC}$. If the formulas in the form $\F{C} \phi$ are absent, we denote the language by $\C L_{\F D}$. 

Given finite $P \subseteq \C P$, we say a formula $\phi$ is over $P$ if every atom occurring in $\phi$ belongs to $P$; conversely, denote the atoms occurring in $\phi$ by $\C P(\phi)$. Let $\top$ and $\bot$ be true and false, respectively. The uppercase letters $\Phi$ and $\Psi$ range over finite sets of formulas. The duals, $\hat{\F K}_i$, $\hat{\F D}_{\C B}$, and $\hat{\F C}$, are the short forms of ``$\neg \F K_i \neg$'', ``$\neg \F D_{\C B} \neg$'', and ``$\neg \F C \neg$'', respectively. The modal depth of a formula $\phi$, denoted by $dep(\phi)$, is the depth of nesting of modal operators in $\phi$.
\end{defn}

\begin{defn}
A \emph{frame} is a pair $(S, R)$, where
\begin{itemize}
\item $S$ is a nonempty set of possible worlds;

\item $R$ is a function that maps each agent $i \in \C A$ to a binary relation $R_i \subseteq S \times S$, called \emph{the accessibility relation} for $i$.
\end{itemize}
A \emph{Kripke model} is a triple $M = (S, R, V)$, where $(S, R)$ is a frame and 
\begin{itemize}
\item $V : S \to \Pow(\C P)$ is a valuation map.
\end{itemize}
A \emph{pointed (Kripke) model} is a pair $(M, s)$, where $M$ is a Kripke model and $s \in S$, called the \emph{actual} world.
\end{defn}

Given a model $M$, for each world $t$, let $R_i(t) = \{u \mid (t, u) \in R_i\}$, which are called \emph{$i$-successors} of $t$. For $\C B \in \PPA$, denote $R_{\C B} = \bigcap_{i \in \C B}R_i$. Thus, $\C B$-successors of $t$, $R_{\C B}(t)$, is similarly defined. Furthermore, for any possible relation $R_x$ in our discussion, let $R^{- 1}_x = \{(u, t) \mid (t, u) \in R_x\}$; let $\TC(t, R_x)$ be the smallest set satisfying
\begin{itemize}
\item $R_x(t) \subseteq \TC(t, R_x)$,

\item for every $u \in \TC(t, R_x)$ and $v \in R_x(u)$, we have $v \in \TC(t, R_x)$.
\end{itemize}
Note that $\TC(t, R_x)$ is the set of the worlds reachable from $t$ via $R_x$, or in other words, the image of $t$ under the transitive closure of $R_x$. Let $\RT(t, R_x) = \TC(t, R_x) \cup \{t\}$. $\RT(t, R_x)$ is the image of $t$ under the reflexive-transitive closure of $R_x$. Particularly, when $R_x$ is $\bigcup_{i \in \C A} R_i$, let's write $\RT(t)$ and $\TC(t)$ instead, the latter of which is all the worlds that are reachable from $t$.
\begin{defn}
Let $(M, s) = (S, R, V, s)$ be a pointed model. The satisfaction relation for $\C L^{\F D}_{\F C}$ is inductively defined as follows:
\begin{itemize}
\item $M, s \vDash p$ iff $p \in V(s)$;

\item $M, s \vDash \neg \phi$ iff $M, s \not \vDash \phi$;

\item $M, s \vDash \phi \land \psi$ iff $M, s \vDash \phi$ and $M, s \vDash \psi$;

\item $M, s \vDash \phi \lor \psi$ iff $M, s \vDash \phi$ or $M, s \vDash \psi$;

\item $M, s \vDash \F K_i \phi$ iff for every $t \in R_i(s)$, $M, t \vDash \phi$;

\item $M, s \vDash \F D_{\C B} \phi$ iff for every $t \in R_{\C B}(s)$, $M, t \vDash \phi$;

\item $M, s \vDash \F C \phi$ iff for every $t \in \TC(s)$, $M, t \vDash \phi$;

\end{itemize}
\end{defn}
Observe that for every $i \in \C A$, $M, s \vDash \F D_{\{i\}} \phi$ if and only if $M, s \vDash \F K_i \phi$. Therefore, we take the modalities $\F K_i$ and $\F D_{\{i\}}$ the same thing.

The axioms and modal systems of $\san K$, $\san{D}$, $\san T$, $\san{K45}$, $\san{KD45}$, and $\san{S5}$ are well-known; if $\san L$ is any of the six system, then ``$\san{L}_n$'' refers to the case of $n$ agents, where $n = |\C A|$, ``$\san{L}_n \F D$'' the case with distributed knowledge, and ``$\san{L}_n \F{DPC}$'' the case with both distributed knowledge and propositional common knowledge. Readers can turn to the book \cite{Fagin1995reason} for more information of these systems, like soundness and completeness. We still denote any system above by $\san L$. A model is an $\san L$-model, if it satisfies all the tautologies of the system $\san L$. A formula $\phi$ is $\san L$-satisfiable, if there exists an $\san L$-model $(M, s)$ such that $M, s \vDash \phi$. We write $\phi \vDash_{\san L} \psi$ if for every $\san L$-model $(M, s)$, if $M, s \vDash \phi$ then $M, s \vDash \psi$. We write $\phi \equiv_{\san L} \psi$ if both $\phi \vDash_{\san L} \psi$ and $\psi \vDash_{\san L} \phi$.

\begin{table}[t]
\centering
\begingroup
\setlength{\arrayrulewidth}{1pt}
\renewcommand{\arraystretch}{1.15}

\begin{tabular}{@{} l  *{6}{c} @{}}
\hline
& $\san{K}$ & $\san{D}$ & $\san{T}$ & $\san{K45}$ & $\san{KD45}$ & $\san{S5}$ \\
\hline
serial &  & $\surd$ & $\surd$ & & $\surd$ & $\surd$ \\
reflexive &  &  & $\surd$ &  &  & $\surd$ \\
transitive &  &  &  & $\surd$ & $\surd$ & $\surd$ \\
Euclidean &  &  &  & $\surd$  & $\surd$ & $\surd$ \\
\hline
\end{tabular}

\endgroup
\caption{Systems and any accessibility relation $R_i$}
\label{tab: LvsRelation}
\end{table}

\section{Uniform Interpolation and Syntactic Methods}\label{sec: SyntacticMethods}
In this section, we will give the definitions of uniform interpolation and collective $p$-bisimulation, and thus extend the syntactic methods in \cite{fang2019forgetting} to the distributed knowledge case, and discuss their properties. 

\subsection{Uniform interpolation, bisimulation, and forgetting}\label{subsec: Bisimulation}
\begin{defn}[Uniform Interpolation]\label{defn: UniformInterpolation}
A modal system $\san L$ has the \emph{uniform interpolation property}, if for every $P$, every $\san{L}$-satifiable $\phi$ over $P$, and every $p \in P$, there is an $\san{L}$-satifiable formula $\psi$ over $P \setminus \{p\}$ such that for every formula $\chi$ without $p$ occurring, \[\phi \vDash_{\san L} \chi \Longleftrightarrow \psi \vDash_{\san L} \chi.\]
We call $\psi$ a uniform interpolant of $\phi$ in $\san L$ over $P \setminus \{p\}$.
\end{defn}

Roelofsen extends the concept of bisimulation to collective bisimulation as a sufficient condition of two models satisfying the same distributed knowledge formulas \cite{roelofsen2007distributed}.
We can extend the concept of collective bisimulation to collective $p$-bisimulation, which describes the similarity of two models ignoring the valuation on $p$.

\begin{defn}[Collective $p$-Bisimulation]\label{defn: col-p-bisim-pointed}
Let $(M, s)$ and $(M', s')$ be two pointed models where $M = (S, R, V)$ and $M' = (S', R', V')$. Let $p \in \C P$. A \emph{collective $p$-bisimulation} between $(M, s)$ and $(M', s')$ is a relation $\rho \subseteq S \times S'$ such that $(s, s') \in \rho$ and for $\C C \in \PPA$, $u \in S$, and $u' \in S'$, if $(u, u') \in \rho$, then the following hold:
\begin{itemize}
\item \F{Atoms}: $V(u) \sim_p V'(u')$, that is, for every $q \in \C P \setminus \{p\}$, $q \in V(u)$ if and only if $q \in V'(u')$;

\item \F{Forth}: For every $v \in R_{\C C}(u)$, there exists $v' \in R'_{\C C}(u')$ such that $(v, v') \in  \rho$;

\item \F{Back}: For every $v' \in R'_{\C C}(u')$, there exists $v \in R_{\C C}(u)$ such that $(v, v') \in  \rho$.
\end{itemize}
If this relation $\rho$ exists, we say the two models are collectively $p$-bisimilar, or $\rho$ witnesses the collective $p$-bisimulation between $(M, s)$ and $(M', s')$, denoted by $\rho: (M, s) \bisim^{col}_p (M', s')$.
\end{defn}
If we replace the Atoms condition with ``$V(u) = V'(u')$'', this definition becomes that of collective bisimulation. Naturally, we can consider the adequacy lemma of collective $p$-bisimulation w.r.t $\C L^{\F D}_{\F C}$, which means that $\C L^{\F D}_{\F C}$ is collective-$p$-bisimulation invariant.
\begin{restatable}[the Adequacy Lemma of Collective $p$-Bisimulation]{lem}{Adequacy}
\label{lem: col-p-bisim-adequacy}
Let $(M, s)$ and $(M', s')$ be two models, $p \in \C P$, and $\phi$ be any $\mathcal{L}^{\F D}_{\F C}$ formula without $p$ occurring. 
If $(M, s) \bisim^{col}_p (M', s')$, then 
$M, s \vDash \phi \Longleftrightarrow M', s' \vDash \phi$.
\end{restatable}
With the help of collective $p$-bisimulation, another important concept, forgetting, is introduced as follows. Forgetting is the semantic counterpart of uniform interpolation. The results of forgetting can act as uniform interpolants.
\begin{defn}[Forgetting]\label{defn: forgetting}
Consider the context of a modal system $\san L$. Let $\phi \in \C L^{\F D}_{\F C}$ and $p$ an atom. A formula $\psi$ satisfying that $\C P(\psi) \subseteq \C P(\phi) \backslash \{p\}$ is a result of forgetting $p$ in $\phi$, written as $dforget_{\san L}(\phi, p) \equiv_{\san L} \psi$, if the following conditions hold,
\begin{itemize}
\item \textbf{Forth}: For any $\san L$-models $(M, s)$ and $(M', s')$, if $M, s\vDash \phi$ and $(M, s) \bisim^{col}_p (M', s')$, then $M', s' \vDash \psi$; 

\item \textbf{Back}: For any $\san L$-model $(M', s')$, if $M', s' \vDash_{\san L} \psi$, then there is a $\san L$-model $(M, s)$ s.t. $M, s \vDash \phi$ and  $(M, s) \bisim^{col}_p (M', s')$.
\end{itemize}
We say $dforget_{\san L}(\phi, p)$ exists if there is an $\san L$-satisfiable formula $\psi$ so that $dforget_{\san L}(\phi, p) \equiv_{\san L} \psi$; otherwise, we say $dforget_{\san L}(\phi, p)$ does not exist.
We say $\san L$ is closed under forgetting if for every $\san L$-satisfiable formula $\phi$ and every $p \in \C P$, $dforget_{\san L}(\phi, p)$ exists.
\end{defn}
\begin{restatable}{prop}{UIforgetting}\label{prop: UI-forgetting}
For every modal system $\san L$, if it is closed under forgetting, then it has the uniform interpolation property.  More precisely, for every $P \subseteq \C P$, every $p \in \C P$, every $\san L$-satisfiable formula $\phi$ over $P$, and every $\san L$-satisfiable $\psi$ over $P \setminus \{p\}$, 
if $dforget_{\san L}(\phi, p) \equiv_{\san L} \psi$, then $\psi$ is a uniform interpolant of $\phi$ in $\san L$ over $P \setminus \{p\}$.
\end{restatable}

Observe that forgetting is closed under disjunction:
\begin{restatable}{fact}{DforgetVee}\label{fact: dforget-vee}
For every $\san L$-satisifiable formulas $\phi_1, \phi_2$, and $p \in \C P$, 
\[dforget_{\san L}(\phi_1 \lor \phi_2, p) \equiv_{\san L} dforget_{\san L}(\phi_1, p) \lor dforget_{\san L}(\phi_2, p).\]
\end{restatable}
These inspire us to represent a result of forgetting (uniform interpolant) as those of some more ``elementary'' formulas. That is, d-canonical formulas, which will be defined in the following.

\subsection{d-canonical formulas and pruning}\label{subsec: CanonForm}
In the subsection, let's focus on the distributed knowledge epistemic modal logic, $\C L_{\F D}$.

For the set $\Phi$, denote $\bigvee_{\phi \in \Phi} \phi$ (resp. $\bigwedge_{\phi \in \Phi} \phi$) by $\bigvee \Phi$ (resp. $\bigwedge \Phi$). For $\C B \in \PPA$, let $\hat{\F D}_{\C B} \Phi = \{\hat{\F D}_{\C B} \phi \mid \phi \in \Phi\}$, and let
\[\nabla_{\C B} \Phi = \F D_{\C B}(\bigvee \Phi) \land (\bigwedge \hat{\F D}_{\C B} \Phi)\]
It is easy to check that for every model $(M, s)$, $M, s \vDash \nabla_{\C B} \Phi$, if and only if, for all $t \in R_{\C B}(s)$, there is $\phi \in \Phi$ such that $M, t \vDash \phi$, and for all $\phi \in \Phi$, there is $t \in R_{\C B}(s)$ such that $M, t \vDash \phi$.
\begin{defn}
Let $P \subseteq \C P$ be finite. Define the set $D^P_k$ inductively as follows:
\begin{itemize}
\item $D^P_0 = \{\bigwedge_{p \in S} p \land \bigwedge_{p \in (P \setminus S)} \neg p \mid S \subseteq P\}$, i.e., the set of the minterms of $P$.

\item $D^P_{k+1}$ consists of formulas in the form 
\[\delta_0 \land \bigwedge_{\C B \in \PPA}\nabla_{\C B} \Phi_{\C B},\]
where $\delta_0 \in D^P_0$ and $\Phi_{\C B} \subseteq D^P_k$.
\end{itemize}
Let's call the formulas in the form above \emph{distributed knowledge canonical formulas}, or \emph{d-canonical formulas} for short, of $P$. For $\delta = \delta_0 \land \bigwedge_{\C B \in \PPA}\nabla_{\C B} \Phi_{\C B}$, let's write $w(\delta) = \delta_0$ and $R_{\C B}(\delta) = \Phi_{\C B}$. (In particular, when $k = 0$, $w(\delta) = \delta$.) We write $R_i(\delta)$ rather than $R_{\{i\}}(\delta)$ for simplicity. $R_i(\delta)$ and $R_{\C B}(\delta)$ are also called \emph{$i$-successors} and \emph{$\C B$-successors} of $\delta$, respectively. Given a modal system $\san L$, let's denote by $D^P_k(\san L)$ the members of $D^P_k$ that are satisfiable in $\san L$. For any d-canonical formulas $\delta$ and $\eta$, we say $\eta$ \emph{wholly occurs in} $\delta$, if either $\eta$ is $\delta$ or there exists $\C B \in \PPA$ and $\eta_0 \in R_{\C B}(\delta)$ such that $\eta$ wholly occurs in $\eta_0$.
\end{defn}

Then, we have three important propositions in parallel with the multi-agent epistemic modal logics \cite{fang2019forgetting, moss2007finite}. Every d-canonical formula can act as a model for a formula with no larger depth and every formula can be represented as a unique set of d-canonical formulas.

\begin{restatable}{prop}{Canondichotomy}\label{prop: canon-dichotomy}
Let $\san L$ be any modal system, $\delta \in D^P_k(\san L)$ where $k \in \B N$, and $P \subseteq \C P$ is finite. Let $\phi \in \C L_{\F D}$ such that $dep(\phi) \leq k$ and $\C P(\phi) \subseteq P$. Then, either $\delta \vDash_{\san L} \phi$ or $\delta \vDash_{\san L} \neg \phi$. 
\end{restatable}

\begin{restatable}{prop}{Modelsuniquecanon}\label{prop: models-unique-canon}
Let $(M, s)$ be a pointed model and $k \in \B N$. Let $P \subseteq \C P$ be finite. Then, there exists a unique $\delta \in D^P_k$ such that $M, s \vDash \delta$. 
\end{restatable}

\begin{restatable}{prop}{Uniquecanon}\label{prop: unique-canon}
Let $\san L$ be any modal system, $\phi \in \C L_\F{D}$, $k \geq dep(\phi)$, and $P = \C P(\phi)$. Then, there exists a unique set $\Phi \subseteq D^P_k(\san L)$ such that $\phi \equiv_{\san L} \bigvee \Phi$.
\end{restatable}
This following proposition characterizes the monotonicity of distributed knowledge from the angles of model and formula.
\begin{restatable}{prop}{BCinclude}\label{prop: BC-include}
Let $\C B, \C C \in \PPX$ such that $\C B \subsetneq \C C$.
\begin{enumerate}
\item For every model $(M, s)$, $R_\C{C}(s) \subseteq R_{\C B}(s)$.

\item For every satisfiable d-canonical formula $\delta$, $R_\C{C}(\delta) \subseteq R_{\C B}(\delta)$.
\end{enumerate}
\end{restatable}

Observe that every d-canonical formula $\delta$ can be represented as a tree, with every d-canonical formula $\eta$ that wholly occurs in it as a node ``$w(\eta)$'' and with ``$R_{\C B}$'' labeling its edges. Intuitively, ``$\delta^\downarrow$'' removes the leaves of the largest distance from the root. We also introduce a notation ``$\delta^{\uparrow l}$'' for the remaining tree whose depth is at least $l$. These are illustrated in Figure \ref{fig: pruning}.

\begin{figure}[htbp]

\centering
\resizebox{0.9\linewidth}{!}{
\begin{tikzpicture}[
    block/.style={draw, darkgreen, dashed, thick, rounded corners, inner sep=10pt},
    world/.style={thick, minimum size=8mm},
    edgei/.style={ ->, black, shorten >=2pt, shorten <=2pt},
    label/.style={font=\footnotesize}
]


\node[world] (lh) at (0,0) {\(p \land q\)};

\node[world] (l0) at (-2,-2) {\(\neg p \land q\)};
\node[world] (l1) at (2,-2) {\(p \land \neg q\)};

\node[world] (l00) at (-3,-4) {\( p \land q\)};
\node[world] (l01) at (-1,-4) {\(\neg p \land \neg q\)};

\node[world] (l000) at (-3,-6) {\(\neg p \land q\)};
\node[world] (l010) at (-1,-6) {\(p \land \neg q\)};

\draw[edgei] (lh) -- (l0)node[midway, above left] {\(R_{\C B_1}\)};
\draw[edgei] (lh) -- (l1) node[midway, above right] {\(R_{\C B_2}\)};

\draw[edgei] (l0) -- (l00)node[midway, above left] {\(R_{\C C_1}\)};
\draw[edgei] (l0) -- (l01) node[midway, above right] {\(R_{\C C_2}\)};
\draw[edgei] (l00) -- (l000)node[midway, left] {\(R_{\C D_1}\)};
\draw[edgei] (l01) -- (l010) node[midway, right] {\(R_{\C D_2}\)};


\node[world] (mh) at (8,0) {\(p \land q\)};

\node[world] (m0) at (6,-2) {\(\neg p \land q\)};
\node[world] (m1) at (10,-2) {\(p \land \neg q\)};

\node[world] (m00) at (5,-4) {\( p \land q\)};
\node[world] (m01) at (7,-4) {\(\neg p \land \neg q\)};

\draw[edgei] (mh) -- (m0)node[midway, above left] {\(R_{\C B_1}\)};
\draw[edgei] (mh) -- (m1) node[midway, above right] {\(R_{\C B_2}\)};

\draw[edgei] (m0) -- (m00)node[midway, above left] {\(R_{\C C_1}\)};
\draw[edgei] (m0) -- (m01) node[midway, above right] {\(R_{\C C_2}\)};


\node[world] (rh) at (16,0) {\(p \land q\)};

\node[world] (r0) at (14,-2) {\(\neg p \land q\)};
\node[world] (r1) at (18,-2) {\(p \land \neg q\)};

\draw[edgei] (rh) -- (r0)node[midway, above left] {\(R_{\C B_1}\)};
\draw[edgei] (rh) -- (r1) node[midway, above right] {\(R_{\C B_2}\)};

\node[draw, single arrow, minimum height=1.4cm, minimum width=0.5cm,
      single arrow head extend=0.1cm, rotate=0] (boxarrow1) at (4,-1) {};
\node at (boxarrow1.north) [above=2pt] {\LARGE $\downarrow$};

\node[draw, single arrow, minimum height=1.4cm, minimum width=0.5cm,
      single arrow head extend=0.1cm, rotate=0] (boxarrow2) at (12,-1) {};
\node at (boxarrow2.north) [above=2pt] {\LARGE $\downarrow$};

\node[draw, single arrow, minimum height=12cm, minimum width=0.5cm,
      single arrow head extend=0.1cm, rotate=8] (boxarrow3) at (8.5,-5) {};
\node at (boxarrow3.south) [below = 2pt, rotate=8] {\LARGE $\downarrow 2$ \text{ or } $\uparrow$};

\node[label] at (lh.north) [above=11pt] {\LARGE \(\delta\)};
\node[label] at (mh.north) [above=11pt] {\LARGE \(\delta^\downarrow \text{ and } \delta^{\uparrow 2}\)};

\node[label] at (rh.north) [above=11pt] {\LARGE \(\delta^{\downarrow 2} \text{ and } \delta^{\uparrow}\)};
\end{tikzpicture}
}

  \caption{Pruning}
  \label{fig: pruning}
\end{figure}

\begin{defn}\label{defn: uparrow}
Let $P \subseteq \C P$ be finite. Let $k \in \B N$ and $\delta \in D^P_k$. Then, $\delta^\downarrow$ is inductively defined as follows:
\[
\delta^\downarrow = 
\begin{cases}
\delta & \text{if $k = 0$;}\\
w(\delta) & \text{if $k = 1$;}\\
w(\delta)\land \bigwedge_{\C B \in \PPA}\nabla_{\C B} (R_{\C B}(\delta))^\downarrow & \text{otherwise.}\\
\end{cases}
\]
where for any set $\Phi$, $\Phi^\downarrow = \{\phi^\downarrow \mid \phi \in \Phi\}$.
Let $l \in \B N$ and $\delta^{\downarrow l}$ is inductively defined as follows:
\[
\delta^{\downarrow l} = 
\begin{cases}
\delta & \text{if $l = 0$;}\\
(\delta^{\downarrow l - 1})^\downarrow & \text{otherwise.}\\
\end{cases}
\]
Denote $\Phi^{\downarrow l} = \{\phi^{\downarrow l} \mid \phi \in \Phi\}$. Let $\delta^{\uparrow 0} = w(\delta)$ and for every $l \geq 1$, let
\[
\delta^{\uparrow l} = 
\begin{cases}
\delta & \text{if $dep(\delta) \leq l$}\\

w(\delta)\land \bigwedge_{\C B \in \PPA}\nabla_{\C B} R_{\C B}(\delta)^{\uparrow (l - 1)} & \text{ otherwise}
\end{cases}
\]
where for any set $\Phi$, $\Phi^{\uparrow l} = \{\delta^{\uparrow l} \mid \delta \in \Phi\}$. When $l = 1$, we write $\delta^{\uparrow}$ instead of ``$\delta^{\uparrow 1}$''.
\end{defn}

Naturally, $ \varnothing^{\downarrow l} = \varnothing^{\uparrow l} = \varnothing$. The following two propositions characterize the properties of pruning. Note that ``$R_{\C B}(\delta^{\uparrow l}) = R_{\C B}(\delta)^{\uparrow l}$'' is not true by Definition \ref{defn: uparrow}.

\begin{restatable}{prop}{Distcutinto}\label{prop: dist-cut-into}
Let $P \subseteq \C P$ be finite and let $\delta \in D^P_k$. Then, for any $l < k$ and $\C B \in \PPA$,
\[R_{\C B}(\delta^{\downarrow l}) = R_{\C B}(\delta)^{\downarrow l}.\]
\end{restatable}

\begin{restatable}{prop}{Uparrowproperty}\label{prop: uparrow-property}
Consider the context of a modal system $\san L$. For every $\delta \in D^P_{k}(\san L)$ and $k, l \in \B N$, if $k \geq l$, then the following properties hold:
\begin{enumerate}
\item $\delta \vDash \delta^{\downarrow l}$.

\item $\delta^{\downarrow l} \in D^P_{k - l}(\san L)$.

\item $\delta^{\uparrow l} \in D^P_l(\san L)$.

\item For every $h \in \B N$ and $\gamma \in D^P_{k + h}(\san L)$, $\gamma^{\downarrow h} = \gamma^{\uparrow k} = (\gamma^{\downarrow h})^{\uparrow k}$.

\item For $k_1, k_2 \in \B N$, if $l \leq \min\{k, k_1, k_2\}$, then $
(\delta^{\uparrow k_1})^{\uparrow l} = (\delta^{\uparrow k_2})^{\uparrow l}
$.
\end{enumerate}
\end{restatable}

\subsection{Forgetting as literal elimination}\label{subsec: LiteralElimination}
Every formula's result of forgetting (uniform interpolant) can be represented as the disjunction of those of its d-canonical formulas. 
\begin{restatable}{prop}{Deltaphi}\label{prop: delta-phi}
If $dforget_{\san L}(\delta, p)$ exists for every $p$ and every $\san L$-satisfiable d-canonical formula $\delta$, then $\san L$ is closed under forgetting.
\end{restatable}
\begin{proof}
Let $\phi$ be any $\san L$-satisfiable formula and let $k = dep(\phi)$. Let $\Phi = \{\delta \in D^P_k \mid \delta \vDash_{\san L} \phi\} $. By Proposition \ref{prop: unique-canon}, $\phi \equiv_{\san L} \bigvee \Phi$. Since $dforget_{\san L}(\delta, p)$ exists for every $\delta \in \Phi$ and by Fact \ref{fact: dforget-vee}, we have that
\[dforget_{\san L}(\phi, p) \equiv_{\san L} \bigvee_{\delta \in \Phi}dforget_{\san L}(\delta, p).\]
Then, $dforget_{\san L}(\phi, p)$ exists. Therefore, $\san L$ is closed under forgetting.
\end{proof}
Therefore, we only need to consider how ``$dforget_{\san L}(\delta, p)$'' is obtained. The concept of literal elimination can also be applied to $\C L^{\F D}_{\F C}$.
\begin{defn}[Literal Elimination]
Let $\phi \in \C L^{\F D}_{\F C}$ and $p$ be an atom. We denote by $\phi^{p}$ the formula obtained from $\phi$ by substituting all occurrences of $\neg p$ with $\top$ and \emph{subsequently}, substituting all occurrences of $p$ with $\top$. We call $\phi^p$ the \emph{remainder} of eliminating $p$.
\end{defn}
Similar to ``$\Phi^\downarrow$'', let's denote $\Phi^p = \{\phi^p \mid \phi \in \Phi\}$. It is also easy to check the fact that $(\delta^{\downarrow l})^p = (\delta^p)^{\downarrow l}$ and $(\delta^{\uparrow l})^p = (\delta^p)^{\uparrow l}$.

\begin{restatable}{prop}{Deltapimply}\label{prop: deltap-imply}
If $\delta \in D^P_k$ and $p$ is an atom, then, $\delta \vDash \delta^p$
\end{restatable}

Literal elimination is a robust tool to construct the forgetting result. Intuitively, ``$\delta^p \equiv_{\san L} dforget_{\san L}(\delta, p)$'' may be concluded, but it is not always true (c.f. in \ref{subsection: K45cou}). Here, we provide a sufficient condition for the existence of $dforget_{\san L}(\delta, p)$.

\begin{thm}\label{thm: guideline}
Consider the context of a modal system $\san L$. Let $P \subseteq \C P$ be finite and $p \in P$. For every $k \geq 0$ and every $\delta \in D^P_k(\san L)$, if the following condition holds:
\begin{quote}
There is $l \geq k$ such that for every $\gamma \in D^{P}_{l}(\san L)$ and every $\san L$-model $(M', s')$ of $\gamma^p$, if $\gamma^{\uparrow k} = \delta$, then there is an $\san L$-model $(M, s)$ such that
\begin{itemize}
\item $M, s \vDash \delta$,

\item $(M, s) \bisim^{col}_p (M', s')$.
\end{itemize}
\end{quote}
then $dforget_{\san L}(\delta, p) \equiv_{\san L} \bigvee \{ \gamma \in D^{P}_{l}(\san L) \mid \gamma^{\uparrow k} = \delta\}^p$.
\end{thm}
\begin{proof}
Let $\Phi = \{ \gamma \in D^{P}_{l}(\san L) \mid \gamma^{\uparrow k} = \delta\}$. Let's check $\bigvee \Phi^p$ satisfies the Forth and the Back conditions of forgetting.

For every $\san L$-model $(M, s)$ of $\delta$ and every $\san L$-model $(M', s')$ such that $(M, s) \bisim^{col}_p (M', s')$, there is $\gamma \in D^{P}_{l}(\san L)$ such that $M, s \vDash \gamma$. By Proposition \ref{prop: uparrow-property}, $\gamma \vDash \gamma^{\uparrow k}$ and $\gamma^{\uparrow k} \in D^P_{k}(\san L)$. Then by Proposition \ref{prop: models-unique-canon}, $\delta = \gamma^{\uparrow k}$. So, $\gamma \in \Phi$. Since $(M, s) \bisim^{col}_p (M', s')$, we have $M', s' \vDash \gamma^p$. So, $M', s' \vDash \bigvee \Phi^p$. Then, the Forth condition of forgetting is satisfied. 

Conversely, for every $\san L$-model $(M', s')$ of $\bigvee \Phi^p$, there is $\gamma \in \Phi$ such that $M', s' \vDash \gamma^p$. By the premise, there is an $\san L$-model $(M, s)$ such that
\begin{itemize}
\item $M, s \vDash \delta$,

\item $(M, s) \bisim^{col}_p (M', s')$.
\end{itemize}
Then, the Back condition of forgetting is satisfied. 

Therefore, $dforget_{\san L}(\delta, p) \equiv_{\san L} \bigvee \{ \gamma \in D^{P}_{l}(\san L) \mid \gamma^{\uparrow k} = \delta\}^p$.
\end{proof}
Observe that, if the sufficient condition holds for $l = k$, then by Proposition \ref{prop: models-unique-canon}, $\delta$ is the only member of $\Phi$ and then, $\delta^p \equiv_{\san L} dforget_{\san L}(\delta, p)$. Theorem \ref{thm: guideline} provides us with a guideline for proving the uniform interpolation property. In the following sections, we will settle the parameter $l$ and then construct the model $(M, s)$ to meet the sufficient condition in Theorem \ref{thm: guideline} for each system $\san L$. Thus, it will be proved that each $\san L$ among $\san{K}_n\F D$, $\san{D}_n\F D$, $\san{T}_n\F D$, $\san{K45}_n \F D$, $\san{KD45}_n \F D$, and $\san{S5}_n \F D$ is closed under forgetting and then has the uniform interpolation property.

\section{Uniform Interpolation in $\san K_n \F D$, $\san {KD}_n \F D$, and $\san T_n \F D$}\label{sec: KDT}
The modal systems $\san K_n \F D$, $\san {KD}_n \F D$, $\san T_n \F D$ have the uniform interpolation property. We will prove this by showing the sufficient condition of Theorem \ref{thm: guideline} is satisfied when ``$l = k$''. 
More precisely, we will show that $\delta^p \equiv_{\san L} dforget_{\san L}(\delta, p)$, for every $\san L$ among the three systems. For convenience, let's uniformly denote by $P$ the discussed finite subset of $\C P$ and $p$ is an atom in $P$ to be forgotten.

Let's first consider $\san K_n \F D$ and $\san D_n \F D$.
\begin{lem}[Lemma for $\san K_n \F D$ and $\san D_n \F D$]\label{lem: KandD}
Let $\san L$ be $\san K_n \F D$ or $\san D_n \F D$, and $\delta \in D^P_k(\san L)$ for $k \geq 0$. For any model $(M', s')$, if $M', s' \vDash \delta^p$, then there is an $\san L$-model $(M, s)$ such that $M, s \vDash \delta$ and  $(M, s) \bisim^{col}_p (M', s')$.
\end{lem}
\begin{proof}
Let's prove by induction on $k$.

\F{Base case} ($k = 0$): $\delta$ is a minterm, namely, $\delta = w(\delta)$. Let $M = (S, R, V)$ be a copy of $M'$, and construct a new actual world $s$ such that for every $q \in P$, $q \in V(s)$ if and only if $w(\delta) \vDash q$. For every $i \in \C A$ and $t' \in R'_{i}(s')$ in $M'$, if $t$ is the copy of $t'$ in $M$, add $(s, t)$ to $R_i$ in $M$. Let $\rho = \{(u, u') \mid u \text{ is the copy of }u'\} \cup \{(s, s')\}$. Then, it is easy to see that $(M,s) \bisim^{col}_p (M', s')$ and $M, s \vDash \delta$. (For convenience, let's say $s'$ (resp. $\delta$) constructs $s$.)

\F{Inductive step}: Let $s$ be the constructed new actual world by $s'$ as is in the base case. Initially, let $S = \{s\}$, $R_i = \varnothing$ for every $i \in \C A$, and $\rho = \{(s, s')\}$.

For every $\C B \in \PPA$, every $t' \in R'_{\C B}(s')$, and every $\eta \in R_{\C B}(\delta)$, suppose that $M', t' \vDash \eta^p$ and $\eta^p \in R_{\C B}(\delta^p)$. By the inductive hypothesis, there is a model $(M^t, t)$ and a bisimulation $\rho_t$ such that
\[M^t, t \vDash \eta \text{ and } \rho_t: (M^t, t) \bisim^{col}_p (M', t').\]
Denote $M^t = (S^t, R^t, V^t)$.  Then, concatenate $s$ with $(M^t, t)$ in the following way:
\begin{itemize}
\item Extend $S$ with $S\cup S^t$, $R_i$ with $R_i \cup R_i^t$ for $i \in \C A$, $V$ with $V \cup V^t$, and $\rho$ with $\rho \cup \rho_t$.

\item Add $(s, t)$ to $R_j$ for every $j \in \C B$.
\end{itemize}
Denote by $T^*_{\C B}$ the set of all such $t$ with regard to $\C B$.

Then, let $(M, s) = (S, R, V, s)$. Observe that for every $\C B \in \PPA$, $R_{\C B}(s) = \bigcup_{\C B \subseteq \C B_0} T^*_{\C B_0}$. We need to check that $(M, s) \bisim^{col}_p (M', s')$ and $M, s \vDash \delta$.

We have known that $(s, s') \in \rho$. For every pair $(u, u')\in \rho$, $u$ is either a copy of $u'$ or constructed by $u'$ inductively. It is easy to see that $V(u) \sim_p V(u')$. If $u$ is in some submodel $(M^t, t)$, by the construction, $\rho_t: (M^t, t) \bisim^{col}_p (M', t')$ and then the Forth and Back conditions about $u$ are satisified by the inductive hypothesis. We only need to consider the case where $u$ is $s$:
\begin{itemize}
\item For every $\C B \in \PPA$, every $v \in R_{\C B}(s)$, by the construction, there is $\C B_0 \supseteq \C B$ such that $v \in T^*_{\C B_0}$ and $v$ is constructed by $v' \in R_{\C B_0}(s')$ and $(v, v')\in \rho$. Since $R_{\C B_0}(s') \subseteq R_{\C B}(s')$, then $v' \in R_{\C B}(s')$.

\item For every $\C B \in \PPA$, every $v' \in R'_{\C B}(s')$, since $M', s' \vDash \delta^p$, there is $\eta \in R_{\C B}(\delta)$ such that $M', v' \vDash \eta^p$. By the construction, there is $v \in T^*_{\C B} \subseteq R_{\C B}(s)$ constructed by $\eta$ and $v' \in R_{\C B}(s')$ and $(v, v')\in \rho$. 
\end{itemize}
Therefore, $\rho: (M, s) \bisim^{col}_p (M', s')$

Let's check $M, s \vDash \delta$. We have already known $M, s \vDash w(\delta)$ in the construction.
\begin{itemize}
\item For every $\C B \in \PPA$ and every $\eta \in R_{\C B}(\delta)$, as $\eta^p \in R_{\C B}(\delta^p)$ and $M', s' \vDash \nabla_{\C B} R_B(\delta^p)$, there is $t' \in R'_B(s')$, such that $M', t' \vDash \eta^p$. By the construction, there is $t \in T^*_{\C B} \subseteq R_{\C B}(s)$ constructed by $t'$ and $\eta$ so that $M^t, t \vDash \eta$. Note that the construction only adds the edge $R_{\C B}$ from $s$ to $t$ and leaves others intact. So, $M, t \vDash \eta$.

\item For every $\C B \in \PPA$ and every $t \in R_{\C B}(s)$, there is $\C B_0 \supseteq \C B$ such that $t \in T^*_{\C B_0}$ and $t$ is constructed by $t'$ and $\eta \in R_{\C B_0}(\delta) \subseteq R_{\C B}(\delta)$. By the inductive hypothesis of the construction, $M^t, t \vDash \eta$. The construction only adds the edge $R_{\C B_0}$ from $s$ to $t$ and keeps others intact. So, $M, t \vDash \eta$.
\end{itemize}
So, $M, s \vDash \bigwedge_{\C B \in \PPA} \nabla_{\C B}R_{\C B}(\delta)$. Therefore, $M, s \vDash \delta$.

In this construction, for every $i \in \C A$, every $t, t'$ with $(t, t') \in \rho$, when $R'_{i}(t')$ is nonempty, so is $R_{i}(t)$. So, when $\san L$ is $\san{D}_n \F D$, $(M, s)$ is also a $\san{D}_n \F D$-model (a serial model).

Hence, the lemma is proved for $\san{K}_n \F D$ and $\san{D}_n \F D$.
\end{proof}

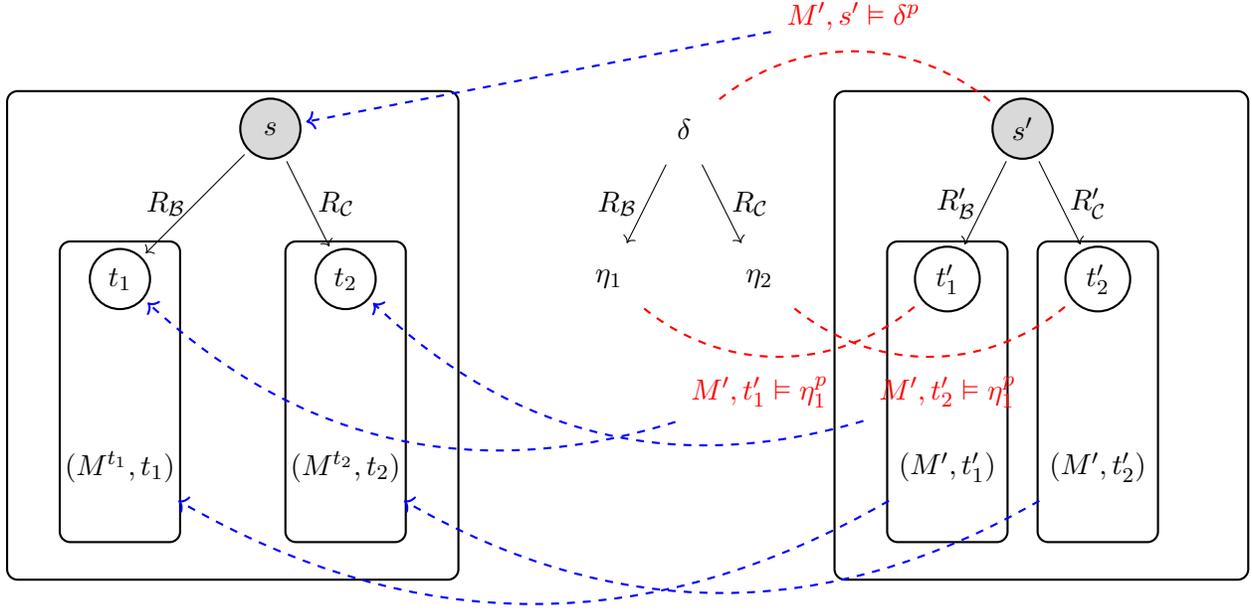
\begin{figure}[htbp]
    \centering
\begin{tikzpicture}[
    block/.style={draw, thick, rounded corners, inner sep=10pt},
    dblock/.style={draw, dashed, thick, rounded corners, inner sep=10pt},
    world/.style={circle, draw, thick, minimum size=8mm},
    formula/.style={thick, minimum size=8mm},
    actual/.style={world, fill=gray!30},
    edgei/.style={->, shorten >=2pt, shorten <=2pt},
    edgej/.style={red, dashed, thick, shorten >=2pt, shorten <=2pt, bend right=40}, 
    edgek/.style={->, blue, dashed, thick, shorten >=2pt, shorten <=2pt, bend left=30}, 
    edgejup/.style={red, dashed, thick, shorten >=2pt, shorten <=2pt, bend left=40}, 
    edgekst/.style={->, blue, dashed, thick, shorten >=2pt, shorten <=2pt}, 
    label/.style={font=\footnotesize}
]

\node[actual] (s') at (0,0) {\(s'\)};
\draw[block] (-2.5, 0.5) rectangle (3,-6);

\node[world] (t'1) at (-1,-2) {\(t'_1\)};
\draw[block] (-1.8, -1.5) rectangle (-0.2,-5.5);
\node[formula] (M't'1) at (-1,-4.5) {\((M', t'_1)\)};
\node[world] (t'2) at (1,-2) {\(t'_2\)};
\draw[block] (0.2, -1.5) rectangle (1.8,-5.5);
\node[formula] (M't'2) at (1,-4.5) {\((M', t'_2)\)};

\draw[edgei] (s') -- (t'1) node[midway, left] {\(R'_{\C B}\)};
\draw[edgei] (s') -- (t'2) node[midway, right] {\(R'_{\C C}\)};

\node[formula] (delta) at (-4.5,0) {\(\delta\)};

\node[formula] (eta1) at (-5.5,-2) {\(\eta_1\)};
\node[formula] (eta2) at (-3.5,-2) {\(\eta_2\)};

\draw[edgei] (delta) -- (eta1) node[midway, left] {\(R_{\C B}\)};
\draw[edgei] (delta) -- (eta2) node[midway, right] {\(R_{\C C}\)};

\node[actual] (s) at (-10,0) {\(s\)};
\draw[block] (-13.5, 0.5) rectangle (-7.5,-6);

\node[world] (t1) at (-12,-2) {\(t_1\)};
\draw[block] (-12.8, -1.5) rectangle (-11.2,-5.5);
\node[formula] (Mt1) at (-12,-4.5) {\((M^{t_1}, t_1)\)};

\node[world] (t2) at (-9,-2) {\(t_2\)};
\draw[block] (-9.8, -1.5) rectangle (-8.2,-5.5);
\node[formula] (Mt2) at (-9,-4.5) {\((M^{t_2}, t_2)\)};

\draw[edgei] (s) -- (t1) node[midway, left] {\(R_{\C B}\)};
\draw[edgei] (s) -- (t2) node[midway, right] {\(R_{\C C}\)};

\draw[edgejup] (delta) to (s');
\node[formula] (deltas') at (-2.25, 1.5) {\color{red}\(M', s' \vDash \delta^p\)};

\draw[edgej] (eta1) to (t'1);
\draw[edgej] (eta2) to (t'2);
\node[formula] (eta1t'1) at (-3.5, -3.5) {\color{red}\(M', t'_1 \vDash \eta_1^p\)};
\node[formula] (eta2t'2) at (-1, -3.5) {\color{red}\(M', t'_2 \vDash \eta_1^p\)};

\draw[edgek] (eta1t'1) to (t1);
\draw[edgek] (eta2t'2) to (t2);
\draw[edgekst] (deltas') to (s);
\draw[edgek] (M't'1) to (Mt1);
\draw[edgek] (M't'2) to (Mt2);



\end{tikzpicture}

  \caption{Inductive step in Lemma \ref{lem: KandD}}
\end{figure}

The case of $\san T_n \F D$ is similar and augmented with the following property about reflexivity:
\begin{restatable}{lem}{Reflexive}\label{lem: reflexive}
Let $\delta \in D^P_k(\san T_n \F D)$ for $k \geq 1$. Let $l \in \{1, \dots, k\}$. Then, for all $\C B \in \PPA$, we have $\delta^{\downarrow l} \in R_{\C B}(\delta^{\downarrow (l-1)})$ and $\delta^{\uparrow (l - 1)} \in R_{\C B}(\delta^{\uparrow l})$.
\end{restatable}

\begin{restatable}[Lemma for $\san T_n \F D$]{lem}{LemmaT}\label{lem: T}
Let $\delta \in D^P_k(\san T_n \F D)$ for some $k \geq 0$. For any $\san T_n \F D$-model $(M', s')$, if $M', s' \vDash \delta^p$, then there is a $\san T_n \F D$-model $(M, s)$ such that $M, s \vDash \delta$ and  $(M, s) \bisim^{col}_p (M', s')$.
\end{restatable}
\begin{proof}
Note that $D^P_k(\san T_n \F D) \subseteq D^P_k(\san K_n \F D)$. By Lemma \ref{lem: KandD}, there exists a $\san K_n \F D$-model there exists $(M^*, s^*) = (S^*, R^*, V^*, s^*)$ such that $M^*, s^* \vDash \delta$ and $(M^*,s^*) \bisim^{col}_p (M', s')$. We further extend $(M^*, s^*)$ by adding all the reflexive edges to every world as follows:
\begin{itemize}
\item Let $S = S^*$, $V = V^*$, and $s = s^*$.

\item For every $i \in \C A$, let $R_i = R_i^* \cup \{(t, t) \mid t \in S\}$ and $R = (R_i)_{i \in \C A}$.
\end{itemize}
Thus, let $(M, s) = (S, R, V, s)$, which is obviously a $\san T_n \F D$-model.

Let $\rho$ be the collective $p$-bisimulation that witnesses $(M^*,s^*) \bisim^{col}_p (M', s')$. The conditions of Definition \ref{defn: col-p-bisim-pointed} obviously remain, except for the Forth condition. For any $u \in S^*$ and $u' \in S'$, suppose $u\rho u'$. For any $\C B \subseteq \C A$, for any $v \in R_{\C B}(u)$, if $u \neq v$, the case is no different from the case of $v \in R^*_{\C B}(u)$ in $(M^*, s^*)$; if $u = v$, because $M'$ is reflexive, then $u' \in R'_{\C B}(u')$ and therefore we still have $u\rho u'$. Therefore, the Forth condition also remains. Therefore, $\rho: (M, s) \bisim^{col}_p (M', s')$.

Let's show that $M, s \vDash \delta$. For every $t \in S$ which is not a copy, let's denote by $\eta_t$ the d-canonical formula that constructs $t$. (If $t$ is $s$, then $\eta_t$ is $\delta$.) We will show that $M, t \vDash \eta_t^{\uparrow l}$ for every $l \in \B N$. Let's prove by induction on $l$.
\begin{itemize}
\item When $l = 0$, $\eta_t^{\uparrow 0} = w(\eta_t)$ and by the construction, $M, t \vDash w(\eta_t)$.

\item Suppose $l > 0$. 

\begin{itemize}
\item $dep(\eta_t) \leq l$: $\eta_t^{\uparrow l} = \eta_t^{\uparrow (l - 1)} = \eta_t$. By the inductive hypothesis, $M, t \vDash \eta_t^{\uparrow l}$.

\item $dep(\eta_t) \geq l > 0$: We have known that $M, t \vDash w(\eta_t)$. For every $u \in R_{\C B}(t)$, if $u \neq t$, by the hypothesis, $M, u \vDash \eta_u^{\uparrow (l - 1)}$; if $u = t$, by Lemma \ref{lem: reflexive}, $\eta_t^{\uparrow (l - 1)} \in R_{\C B}(\eta_t^{\uparrow l})$ and by the hypothesis, $M, t \vDash \eta_t^{\uparrow (l - 1)}$. Conversely, for every $\eta \in R_{\C B}(\eta_t)$, there is a world $u \in R_{\C B}(t)$ constructed by $\eta$ and $M, u \vDash \eta^{\uparrow (l - 1)}$ by the inductive hypothesis. Therefore, $M, t \vDash \eta_t^{\uparrow l}$.
\end{itemize}
 
\end{itemize}
Note that $\delta^{\uparrow k} = \delta$. So, we have $M, s \vDash \delta$.

Hence, the lemma for $\san T_n \F D$ is proved.
\end{proof}

\begin{thm}\label{thm: KDT}
$\san K_n \F D$, $\san {KD}_n \F D$, and $\san T_n \F D$ have the uniform interpolation property.
\end{thm}
\begin{proof}
Let $\san L$ be any of $\san K_n \F D$, $\san {KD}_n \F D$, and $\san T_n \F D$. By Theorem \ref{thm: guideline}, Lemma \ref{lem: KandD}, and Lemma \ref{lem: T}, for every finite $P$, every $p$, and every $\delta \in D^P_k(\san L)$ for $k \geq 0$, $\delta^p \equiv_{\san L} dforget_{\san L}(\delta, p)$. By Proposition \ref{prop: delta-phi}, for every $\san L$-satisifaible formula $\phi$ has its uniform interpolant over $P \setminus \{p\}$. More precisely,
\[dforget_{\san L}(\phi, p) \equiv_{\san L} \bigvee\{\delta^p \mid \delta \vDash_{\san L} \phi \text{ and } \delta \in D^P_{dep(\phi)}(\san L)\}\]
Therefore, $\san K_n \F D$, $\san {KD}_n \F D$, and $\san T_n \F D$ have the uniform interpolation property.
\end{proof}
\section{Uniform Interpolation in $\san{K45}_n \F D$, $\san{KD45}_n \F D$, and $\san{S5}_n \F D$}\label{sec: K45KD45S5}

In the section, let's show $\san{K45}_n \F D$, $\san{KD45}_n \F D$, and $\san{S5}_n \F D$ have the uniform interpolation property. All the three systems require their models to be transitive and Euclidean.  We will first show a counterexamples such that ``$\delta^p \not \equiv_{\san L} dforget_{\san L}(\delta, p)$''. Then, we will provide a general method to obtain $dforget_{\san L}(\delta, p)$, and then conclude uniform interpolation. Finally, we show that in some special cases where $\delta^p \equiv_{\san L} dforget_{\san L}(\delta, p)$ holds. For convenience, let's uniformly denote by $P$ the discussed finite subset of $\C P$ and $p$ is an atom in $P$ to be forgotten.

\subsection{$\delta^p$ may not be the uniform interpolant}\label{subsection: K45cou}\label{subsec: Counterexample}

In this subsection, let's assume $\C A$ contains at least three agents, i.e., $\{1, 2, 3\}$. For simplicity, we also assume $P = \{p, q\}$. The examples in this subsection can be extended to the case where $P$ is larger.

\begin{restatable}{prop}{Sfivecounter}\label{prop: S5counter}
For any $k \geq 2$, there is $\delta^k_{cou} \in D^P_{k}(\san{S5}_n \F D)$ and a $\san{S5}_n \F D$-model $(M', s')$ of $(\delta^k_{cou})^p$ such that for any $\san{K45}_n \F D$-model $(M, s)$, the following two do not hold simultaneously:
\begin{itemize}
\item $(M, s) \vDash \delta^k_{cou}$,

\item $(M, s) \bisim^{col}_p (M', s')$.
\end{itemize}
Thus, $(\delta^k_{cou})^p \not \equiv_{\san L} dforget_{\san L}(\delta^k_{cou}, p)$, for $\san L$ among $\san{K45}_n \F D$, $\san{KD45}_n \F D$, and $\san{S5}_n \F D$.
\end{restatable}

\begin{proof}
Let's only show how $\delta^2_{cou}$ and $(M', s')$ are constructed. The full proof is placed in the appendix. We write ``$\{x, y\} \in R_{\C B}$'' instead of ``$(x, y), (y, x) \in R_{\C B}$''. Consider the model $(M_2, s_2) = (S, R, V, s_2)$ as follows:
\begin{itemize}
\item $S = \{s_2, t_1, t_2, t_3, t_4, v, w\}$.

\item For every $x, y \in \{s_2, t_1, t_2, t_3, t_4\}$, $\{x, y\} \in R_1$; $\{t_1, t_3\}, \{t_2, t_4\} \in R_2$; $\{t_1, w\}, \{t_4, v\} \in R_3$. For every $x \in S$, $(x, x) \in R_{\C A}$. These are all the accessibility relations.

\item $V$ is given by the following:
\begin{itemize}
\item $M_2, x \vDash \neg p \land \neg q$, for $x \in \{t_1, s_2, v\}$,

\item $M_2, t_2 \vDash p \land \neg q$,

\item $M_2, x \vDash \neg p \land q$, for $x \in \{t_3, w\}$,

\item $M_2, t_4 \vDash p \land q$.
\end{itemize}

\end{itemize}
Thus, $(M_2, s_2)$ is an $\san{S5}_n \F D$-model, which is illustrated in Figure \ref{fig: Counterexample}. Let $\delta^2_{cou}$ be the unique member of $D^P_2(\san{S5}_n \F D)$ such that $M_2, s_2 \vDash \delta^{2}_{cou}$. 


Let's $(M', s') = (S', R', V', s')$ be a model of $(\delta^2_{cou})^p$, illustrated in Figure \ref{fig: Counterexample}. $S'$ is a copy of $S$ and $s'$ is a copy of $s_2$. Every world $x' \in S'$ inherits the value of $x$ on $q$. Every accessibility relation is also inherited except that $\{t'_1, t'_4\} \in R'_2$ but $\{t'_1, t'_3\} \notin R'_2$; $\{t'_2, t'_3\} \in R'_2$ but $\{t'_2, t'_4\} \notin R'_2$.
\end{proof}

\begin{figure}[htbp]
  \begin{minipage}{0.48\textwidth}
    \centering
\resizebox{0.9\linewidth}{!}{
\begin{tikzpicture}[
    block/.style={draw, darkgreen, dashed, thick, rounded corners, inner sep=10pt},
    world/.style={circle, draw, thick, minimum size=8mm},
    actual/.style={world, fill=gray!30},
    edgei/.style={blue, shorten >=2pt, shorten <=2pt},
    edgej/.style={red, dashed, thick, shorten >=2pt, shorten <=2pt, bend right=20}, 
    edgek/.style={black, thick, shorten >=2pt, shorten <=2pt, dotted}, 
    label/.style={font=\footnotesize}
]

\node[actual] (s) at (0,0) {\(s_2\)};

\node[world] (t1) at (-3,-3) {\(t_1\)};
\node[world] (t2) at (1,-3) {\(t_2\)};
\node[world] (t3) at (-1,-2) {\(t_3\)};
\node[world] (t4) at (3,-2) {\(t_4\)};

\node[world] (w) at (-4, -5) {\(w\)};
\node[world] (v) at (4, -5) {\(v\)};

\draw[edgei] (s) -- (t1);
\draw[edgei] (s) -- (t2);
\draw[edgei] (s) -- (t3);
\draw[edgei] (s) -- (t4) node[midway, above right] {\(R_1\)};
\draw[edgei] (t1) -- (t2) node[midway, below] {\(R_1\)};
\draw[edgei] (t1) -- (t3);
\draw[edgei] (t1) -- (t4);
\draw[edgei] (t2) -- (t3);
\draw[edgei] (t2) -- (t4);
\draw[edgei] (t3) -- (t4);

\draw[edgej] (t1) to node[midway, below, sloped] {\(R_2\)} (t3) ;
\draw[edgej] (t2) to node[midway, below, sloped] {\(R_2\)} (t4) ;

\draw[edgek] (t1) -- (w) node[midway, above left] {\(R_3\)};
\draw[edgek] (t4) -- (v) node[midway, above right] {\(R_3\)};
%

\node[label] at (s.north) [above=0pt] {\(\neg p \land \neg q\)};
\node[label] at (t1.north) [above=0pt] {\(\neg p \land \neg q\)};
\node[label] at (t2.north) [above=0pt] {\(p \land \neg q\)};
\node[label] at (t3.north) [above=0pt] {\(\neg p \land q\)};
\node[label] at (t4.north) [above=0pt] {\(p \land q\)};
\node[label] at (w.north) [above=-2pt] {\(\neg p \land q\)};
\node[label] at (v.north) [above=-2pt] {\(\neg p \land \neg q\)};
\end{tikzpicture}
}

\caption*{$(M_2, s_2)$}
\end{minipage}
\hfill
\begin{minipage}{0.48\textwidth}
\centering
\resizebox{0.9\linewidth}{!}{
\begin{tikzpicture}[
    block/.style={draw, darkgreen, dashed, thick, rounded corners, inner sep=10pt},
    world/.style={circle, draw, thick, minimum size=8mm},
    actual/.style={world, fill=gray!30},
    edgei/.style={blue, shorten >=2pt, shorten <=2pt},
    edgej/.style={red, dashed, thick, shorten >=2pt, shorten <=2pt, bend right=6}, 
    edgek/.style={black, thick, shorten >=2pt, shorten <=2pt, dotted}, 
    label/.style={font=\footnotesize}
]

\node[actual] (s) at (0,0) {\(s'\)};

\node[world] (t1) at (-3,-3) {\(t'_1\)};
\node[world] (t2) at (1,-3) {\(t'_2\)};
\node[world] (t3) at (-1,-2) {\(t'_3\)};
\node[world] (t4) at (3,-2) {\(t'_4\)};

\node[world] (w) at (-4, -5) {\(w'\)};
\node[world] (v) at (4, -5) {\(v'\)};

\draw[edgei] (s) -- (t1);
\draw[edgei] (s) -- (t2);
\draw[edgei] (s) -- (t3);
\draw[edgei] (s) -- (t4) node[midway, above right] {\(R'_1\)};
\draw[edgei] (t1) -- (t2) node[midway, below] {\(R'_1\)};
\draw[edgei] (t1) -- (t3);
\draw[edgei] (t1) -- (t4);
\draw[edgei] (t2) -- (t3);
\draw[edgei] (t2) -- (t4);
\draw[edgei] (t3) -- (t4);

\draw[edgej] (t1) to node[left, below, sloped] {\(R'_2\)} (t4);
\draw[edgej] (t2) to node[midway, above, sloped] {\(R'_2\)} (t3);

\draw[edgek] (t1) -- (w) node[midway, above left] {\(R'_3\)};
\draw[edgek] (t4) -- (v) node[midway, above left] {\(R'_3\)};
%

\node[label] at (s.north) [above=0pt] {\(\neg q\)};
\node[label] at (t1.north) [above=0pt] {\(\neg q\)};
\node[label] at (t2.north) [above=0pt] {\(\neg q\)};
\node[label] at (t3.north) [above=0pt] {\(q\)};
\node[label] at (t4.north) [above=0pt] {\(q\)};
\node[label] at (w.north) [above=-2pt] {\(q\)};
\node[label] at (v.north) [above=-2pt] {\(\neg q\)};
\end{tikzpicture}
}
\caption*{$(M', s')$}
  \end{minipage}

  \caption{$M_2$ and $M'$ in Proposition \ref{prop: S5counter}. (Reflexive edges $R_{\C A}$ are omitted for readability)}
  \label{fig: Counterexample}
\end{figure}

\begin{restatable}{prop}{Morecounter}\label{prop: More-counter}
Let $\delta^2_{cou}$ and $(M', s')$ be defined as in Proposition \ref{prop: S5counter}. For all $\san{S5}_n \F{D}$-satifiable d-canonical formulas $\gamma$ and $\xi$ such that $dep(\gamma), dep(\xi) > 2$, if 
\begin{itemize}
\item $(\xi^{\uparrow 2})^p = (\delta^2_{cou})^p$ and $M', s' \vDash \xi^p$,

\item $\gamma^{\uparrow 2} = \delta^2_{cou}$,
\end{itemize}
then $\gamma \neq \xi$.
\end{restatable}

\noindent \F{Remark}: Consider the model $(M', s')$ defined as in Proposition \ref{prop: S5counter}. Even if $dforget_{\san L}(\delta^2_{cou}, p)$ exists, Proposition \ref{prop: S5counter} shows that $(M', s')$ is not an $\san L$-model of $dforget_{\san L}(\delta^2_{cou}, p)$. Proposition \ref{prop: More-counter} further explains that for any $\san L$-satisfiable $\gamma$ of larger modal depth than $\delta^2_{cou}$ such that $\gamma \vDash_{\san L} \delta^2_{cou}$, we still have $M', s' \not\vDash \gamma^p$. They inspire us to exclude those ``bad models'' like $(M', s')$ by enlarging the modal depth. When these ``bad models'' are all excluded, we may find the real uniform interpolant. It is actually the notion of Theorem \ref{thm: guideline}. 

Readers may have noticed that Proposition \ref{prop: S5counter} only shows $dforget_{\san L}(\delta^k_{cou}, p)$ is inequivalent to $(\delta^k_{cou})^p$ but does not deny its existence. This subsection claims that $(\delta^k_{cou})^p$ is not a uniform interpolant, but will this claim be completely proved only when we prove $dforget_{\san L}(\delta^k_{cou}, p)$ really exists. 

The next subsection will show that the idea of enlarging the modal depth works and ensures the existence of any $dforget_{\san L}(\delta, p)$, including $dforget_{\san L}(\delta^k_{cou}, p)$.

\subsection{Uniform interpolation in $\san{K45}_n \F D$, $\san{KD45}_n \F D$}\label{subsec: K45KD45}
In this subsection, let's show a general construction of transitive and Euclidean models to meet the sufficient condition of Theorem \ref{thm: guideline}. Let's first consider some key definitions and properties for the transitive and Euclidean case.
\begin{defn}
Let $\Phi$ be a set of formulas and $M = (S, R, V)$ be a model. For any $T \subseteq S$, we say $T$ is $\Phi$-complete, if the following conditions hold
\begin{itemize}
\item for every $t \in T$, there exists $\phi \in \Phi$ such that $M, t \vDash \phi$;

\item for every $\phi \in \Phi$, there exists $t \in T$ such that $M, t \vDash \phi$.
\end{itemize}
\end{defn}
Observe that, $M, s \vDash \nabla_{\C B} \Phi$ iff $R_{\C B}(s)$ is $\Phi$-complete. Naturally, $R_{\C B}(s)$ is $R_{\C B}(\delta)$-complete if $M, s \vDash \delta$. By transitivity and the Euclideanness, it is easy to see that a world $s$ and its $\C B$-successor $t$ share the $\C B$-successors, that is, if $t \in R_{\C B}(s)$, then $R_{\C B}(t) = R_{\C B}(s)$. However, it is slightly different for formulas, because a d-canonical formula $\delta$ and its $\C B$-successor $\eta$ have different modal depths. Then, we have the following property \footnote{In the multi-agent modal logic $\san{K45}_n$ \cite{fang2019forgetting}, this is called ``identical children property'' while $i$-successors are called ``$i$-children'' instead. However, the construction of a $\san{K45}_n \F D$-model becomes so complex that in Lemma \ref{lem: ConstructionK45}, a ``child'' of $t$ may belong to the ``older generation'' than $t$. It may confuse readers if we say ``children''. So, we rename them with another commonly used word, \emph{successor}.}:

\begin{restatable}[Identical Successors Property]{prop}{IdenticalChildren}\label{prop: dist-identical-son}
Let $\delta$ be a $D^P_k(\san{K45}_n \F D)$ for  $k \geq 2$. Let $l \in \B N$ s.t. $1 \leq l < k$. Then, for all $\C B \in \PPA$ and $\eta \in R_{\C B}(\delta)$, $R_{\C B}(\delta)^{\downarrow l} = R_{\C B}(\eta)^{\downarrow l - 1}$.
\end{restatable}

The following proposition describes how the distributed knowledge transitive and Euclidean relations are intertwined, which we have to consider for the model construction.
\begin{restatable}[Sibling Property]{prop}{Brothers}\label{prop: brothers}
Let $\delta \in D^P_k(\san L)$ where $k \geq 2$. Let $\C B, \C C \in \PPA$ such that $\C C_1 = \C C \cap \C B \neq \varnothing$ and $\C C_2 = \C C \cap \C B^c \neq \varnothing$. For every $\eta \in R_{\C B}(\delta)$ and $\chi \in R_{\C C}(\eta)$, there is $\eta_0 \in R_{\C C_1}(\delta)$ such that 
\begin{itemize}
\item $\eta_0^{\downarrow} = \chi \in R_{\C C}(\eta)$,

\item $R_{\C D}(\eta_0) = R_{\C D}(\eta)$ for every $\C D \in \PP(\C C_2)$,
\end{itemize}
\end{restatable}

Transitivity and Euclideanness, as well as reflexivity, are global properties, and it is not practical to keep them during each construction step. Let's introduce the following local properties:
\begin{defn}\label{defn: equivalent}
Let $M$ be any Kripke model and $W$ be any set of worlds of $M$. For $i \in \C A$, we say a set $W$ is \emph{quasi-$i$-equivalent}, if 
\begin{itemize}
\item $W = \RT(t, R_i \cup R_i^{- 1})$ for every $t \in W$;

\item for every $t_1, t_2 \in W$, $R_i(t_1) = R_i(t_2)$.
\end{itemize}
We say $W$ is \emph{$i$-equivalent} if it further satisfies that
\begin{itemize}
\item for every $t \in W$, $(t, t) \in R_i$.
\end{itemize}
For every $\C B \in \PPA$, if we replace the ``$i$'' above with $\C B$, we obtain the definitions of \emph{quasi-$\C B$-equivalence} and \emph{$\C B$-equivalence} respectively.
\end{defn}

\begin{restatable}{prop}{QeTranEu}\label{prop: q-e-tran-Eu}
For any model $M$, these two claims are equivalent
\begin{enumerate}
\item for every $t$ in $M$ and every $i \in \C A$, $\RT(t, R_i \cup R_i^{- 1})$ is quasi-$i$-equivalent;

\item $M$ is transitive and Euclidean.
\end{enumerate}
Furthermore, the following are also equivalent
\begin{enumerate}
\item[3] for every $t$ in $M$ and every $i \in \C A$, $\RT(t, R_i \cup R_i^{- 1})$ is $i$-equivalent;

\item[4] $M$ is transitive, Euclidean, and reflexive.
\end{enumerate}
\end{restatable}

In the later construction, if we have constructed some transitive Euclidean models inductively, we will connect them in the following way, so that the new model is still transitive and Euclidean.
\begin{defn}
Given a Kripke model $M$ and $i \in \C A$, let $W_1, W_2, \dots, W_m$ be any distinct quasi-$i$-equivalent sets ($m \geq 2$) in $M$. \emph{Merging them with the relation $R_i$} is the process as follows: for any $t, u \in W_1 \cup W_2 \cup \dots \cup W_m$, if $(u, u) \in R_i$, then add $(t, u) \in R_i$.
\end{defn}

The process of merging also implies that, as the induction proceeds, a set ``$R_{\C B}(t)$'' may become larger. Then, the notation ``$R_{\C B}(t)$'' may become misleading. Let's introduce the following concepts:
\begin{defn}
Let $M = (S, R, V)$ be any model and $\C X \in \PPA$. A \emph{multi-pointed model} is a model in the form $(M, \C U_{\C X})$, where $\C U_{\C X}$ is a family of subsets of $S$ indexed by $\C P^+(\C X)$, that is,
\[\C U_{\C X} = \{T_{\C B} \subseteq S \mid \C B \in \C P^+(\C X)\}\]
\end{defn}
For example, given a model $(M, s)$, if $\C X = \C A$ and $\C U_{\C A} = \{R_{\C B}(s) \mid \C B \in \PPA\}$, then $(M, \C U_{\C A})$ makes a multi-pointed model. We also rewrite Definition \ref{defn: col-p-bisim-pointed} to fit the multi-pointed models as follows:
\begin{defn}\label{defn: col-p-bisim}
Given $\C X \in \PPA$, suppose $(M, \C U_{\C X})$ and $(M', \C U'_{\C X})$ be two multi-pointed models. A collective $p$-bisimulation between $(M, \C U_{\C X})$ and $(M', \C U'_{\C X})$ is a relation $\rho$ between $S$ and $S'$ satisfying the following:
\begin{itemize}
\item For any $\C B \in \PPX$,
\begin{itemize}
\item for every $t \in T_{\C B}$, there is $t' \in T'_{\C B}$ such that $(t, t') \in \rho$;

\item for every $t' \in T'_{\C B}$, there is $t \in T_{\C B}$ such that $(t, t') \in \rho$.
\end{itemize}

\item For any $u \in S$ and $u' \in S'$, whenever $(u, u') \in \rho$, we have
\begin{itemize}
\item \textbf{Atoms}: $V(u) \sim_p V'(u')$;

\item \textbf{Forth}: for all $\C B \subseteq \C A$ and all $v \in S$, if $(u, v) \in R_{\C B}$, then there is $v'$ such that $(u',  v') \in R'_{\C B}$ and $(v, v') \in \rho$;

\item \textbf{Back}: for all $\C B \subseteq \C A$ and all $v' \in S$, if $(u', v') \in R'_{\C B}$, then there is $v$ such that $(u, v) \in R_{\C B}$ and $(v, v') \in \rho$;
\end{itemize}
\end{itemize}
We also say $\rho$ witnesses the collective $p$-bisimulation between $(M, \C U_{\C X})$ and $(M', \C U'_{\C X})$ and denote this by $\rho: (M, \C U_{\C X}) \bisim^{col}_p (M', \C U'_{\C X})$. 
\end{defn}

Now, we can formally show a general construction for $\san{K45}_n \F D$- (resp. $\san{KD45}_n \F D$-) models. Readers will see, the parameters $\gamma$ and $k + h + 1$ act as the ``$\gamma$'' and ``$l$'' in Theorem \ref{thm: guideline}. We use two formulas $\gamma$ and $\xi$ for this construction instead of $\gamma$ alone, so that $\xi$ helps ensure collective $p$-bisimulation while $\gamma$ ensures the satisfaction relation. Note that $\gamma$ and $\xi$ are not necessarily the same.
\begin{restatable}[The Construction Lemma for $\san{K45}_n \F D$ and $\san{KD45}_n \F D$]{lem}{ConstructionLemma}\label{lem: ConstructionK45} Let $\san L$ be either $\san{K45}_n \F D$ or $\san{KD45}_n \F D$. Let $k, d, h$ be natural numbers such that $0 \leq k \leq d \leq h$. Let $\gamma, \xi \in D^P_{k + h + 1}(\san L)$ such that $(\gamma^{\uparrow (k + d + 1)})^p = (\xi^{\uparrow (k + d + 1)})^p$. Let $(M', s') = (S', R', V', s')$
be an $\san L$-model of $\xi^p$. 

Given $\C X \in \PPA$ and the multi-pointed submodel
\[(M', \C U'_{\C X}) = (S', R', V', \C U'_{\C X})\]
where $\C U'_{\C X} = \{T'_{\C B} \mid T'_{\C B} = R'_{\C B}(s')\text{ for }\C B \in \PPX\}$, there is a multi-pointed $\san L$-model 
\[(M, \C U_{\C X}) = (S, R, V, \C U_{\C X})\]
where $\C U_{\C X}= \{T_{\C B} \subseteq S \mid \C B \in \PPX\}$ such that, 
\begin{itemize}
\item $T_{\C B}$ is $\C B$-equivalent in $(M, \C U_{\C X})$, for $\C B \in \PPX$;

\item $(M, \C U_{\C X}) \bisim^{col}_p (M', \C U'_{\C X})$;

\item $T_{\C B}$ is $R_{\C B}(\gamma)^{\uparrow k}$-complete for $\C B \in \PPX$.
\end{itemize}
\end{restatable}

\begin{proof}
We only show how $(M, \C U_{\C X})$ is constructed. The full proof is placed in the appendix together with some additional lemmas.

Let's construct by induction on $k$ and $d$. We label the constructed model as $(M, \C U_{\C X})^{\gamma, \xi, k, d}$, in order to show the four parameters' change in the induction. The Greek letter $\eta$ ranges over the d-canonical formulas that wholly occur in $\gamma$, while the Greek letter $\zeta$ those that wholly occur in $\xi$.  Initially, let $S$, $V$, $\rho$, and every relation $R_i$ be empty.

\F{Base case ($k = d = 0$)}: Construct $(M, \C U_{\C X})^{\gamma, \xi, 0, 0}$ as follows.

\begin{itemize}
\item[Step 1] 
For every $\C B \in \PPX$, $t' \in R'_{\C B}(s')$, $\eta \in R_{\C B}(\gamma)$, and $\zeta \in R_{\C B}(\xi)$, if $(\eta^{\uparrow (k + d)})^p = (\zeta^{\uparrow (k + d)})^p$ and $M', t' \vDash \zeta^p$, then construct a world $t$ such that for every $q \in P$, $q \in V(t)$ if and only if $w(\eta) \vDash q$. Then, add $(t, t')$ to $\rho$. Let $T^*_{\C B}$ be all such $t$ with respect to this $\C B$.

(Note that when a world $t$ is constructed, we can determine its unique $t'$, $\eta$, and $\zeta$. Let's write this $\eta$ as $\eta_t$ and this $\zeta$ as $\zeta_t$. Let's say $t'$ (resp. $\eta_t$, $\zeta_t$) constructs $t$.)

\item[Step 2] 
For every $\C B_1, \C B_2 \in \PPX$, every $t_1 \in T^*_{\C B_1}$, $t_2 \in T^*_{\C B_2}$, and $i \in \C B_1^c \cap \C B_2^c$, if $(t'_1, t'_2) \in R'_i$, then add $(t_1, t_2)$ to $R_{i}$. 

\item[Step 3] For every $\C B_1, \C B_2 \in \PPX$, every $t_1 \in T^*_{\C B_1}$, $t_2 \in T^*_{\C B_2}$, and $i \in \C B_1 \cap \C B_2$, add $(t_1, t_2)$ to $R_{i}$. 

\item[Step 4] For every $\C B \in \PPX$, let
\[T_{\C B} = \bigcup_{\C B \subseteq \C C} T^*_{\C C}\]
Extend $S$ with $\bigcup_{\C B \in \PPX}T_{\C B}$ and let $\C U_{\C X} = \{T_{\C B} \mid \C B \in \PPX\}$.

Denote by $(M, \C U_{\C X})^{\gamma, \xi, 0, 0}_4$ the submodel constructed via the first four steps.

\item[Step 5] For every $\C B \in \C P^+(\C X)$ and every $t \in T^*_{\C B}$, if there is $i \in \C B^c$ such that $R'_i(t') \neq \varnothing$, consider the multi-pointed submodel 
$(M', \C U^{t'}_{\C B^c})$, where $\C U^{t'}_{\C B^c} = \{R'_{\C D}(t') \mid \C D \in \PP(\C B^c)\}$. Let $(M^t, \C U^t_{\C B^c})$ be a copy of $(M', \C U^{t'}_{\C B^c})$ and let $\rho_t = \{(u, u') \mid \text{$u$ is a copy of $u'$.}\}$.
Add the submodel $(M^t, \C U^t_{\C B^c})$ to $(S, R, V)$, and extend $\rho$ to be $\rho \cup \rho_t$. 

(Note that so far, no edge is added between $t$ and $(M^t, \C U^{t}_{\C B^c})$. We write ``$(M^t, \C U^{t}_{\C B^c})$'' as ``$(M^t, \C U^{t})$'' if it is clear from context.)

\item[Step 6] For every $\C B \in \C P^+(\C X)$, every $t \in T^*_{\C B}$, every $i \in \C B^c$, let's denote the set $\RT(t, R_i \cup R_i^{- 1})$ within $(M, \C U_{\C X})^{\gamma, \xi, 0, 0}_4$ by $\RT(t, R_i \cup R_i^{- 1})_4$. For every $u \in \RT(t, R_i \cup R_i^{- 1})_4$, if $(M^u, \C U^{u})$ exists, which is an $\san L$-model, there is the quasi-$i$-equivalent set $W^u_i$ in $(M^u, \C U^{u})$ such that $W^u_i \supseteq R'_i(u')$. Then, merge all these sets $W^u_i$ and $\RT(t, R_i \cup R_i^{- 1})_4$ with the relation $R_i$.
\end{itemize}
Thus, $(M, \C U_{\C X})^{\gamma, \xi, 0, 0}$ is constructed. Observe that this construction does not specify the value of $d$, so $(M, \C U_{\C X})^{\gamma, \xi, 0, d}$ is also constructed for any $d \geq 0$.

\F{Inductive step}: Steps 1, 3, and 4 are the same as those of the base case. We show Steps 2, 5, and 6, where Step 2 is divided into two substeps.

\begin{itemize}

\item[Step 2-1] 
For every $\C B \in \PPX$ and every $t \in T^*_{\C B}$ that is constructed in Step 1, for every $\C C \in \PPA$ such that $\C C_1 = \C B \cap \C C \neq \varnothing$ and $\C C_2 = \C B^c \cap \C C \neq \varnothing$, and every $u' \in R'_{\C C}(t')$, \emph{we claim that} there exists $\eta \in R_{\C C_1}(\gamma)$ and $\zeta \in R_{\C C_1}(\xi)$ such that
\begin{itemize}
\item $\eta^{\downarrow} \in R_{\C C}(\eta_t)$,

\item $R_{\C D}(\eta) = R_{\C D}(\eta_t)$ for every $\C D \in \PP(\C C_2)$,

\item $M', u' \vDash \zeta^p$

\item $(\eta^{\uparrow (k + d - 1)})^p =  (\zeta^{\uparrow (k + d - 1)})^p$.
\end{itemize}
Construct a new world $u$ such that $q \in P$, $q \in V(u)$ if and only if $w(\eta) \vDash q$; add $u$ to $T^*_{\C C_1}$. Add $(t, u)$ and $(u, u)$ to $R_i$ for every $i \in \C C_2$. Add $(u, u')$ to $\rho$. (We also say $u$ is constructed by $u'$, $\eta$, and $\zeta$, like those constructed in Step 1.)

\item[Step 2-2] For every $\C B \in \PPX$, every $t \in T^*_{\C B}$, every $i \in \PP(\C B^c)$, and every $u_1, u_2 \in R_i(t)$, then add $(u_1, u_2)$ to $R_i$.
\end{itemize}

We need to prove the claim that the formulas $\eta$ and $\zeta$ exist in Step 2-1. On the one hand, since $M', t' \vDash \zeta_t^p$, by Proposition \ref{prop: brothers} (the Sibling Property), there is $\zeta \in R_{\C C_1}(\xi)$ such that 
\begin{itemize}
\item $M', u' \vDash \zeta^p$,

\item $(\zeta^\downarrow)^p \in R_{\C C}(\zeta_t)^p$

\item $R_{\C D}(\zeta)^p = R_{\C D}(\zeta_t)^p$ for every $\C D \in \PP(\C C_2)$.
\end{itemize}
Note that since $t$ is constructed in Step 1, $(\eta_t^{\uparrow (k + d)})^p = (\zeta_t^{\uparrow (k + d)})^p$. So, there is $\chi \in R_{\C C}(\eta_t)$ such that $(\chi^{\uparrow (k + d - 1)})^p = ((\zeta^\downarrow)^{\uparrow (k + d - 1)})^p$. Note that $\chi, \zeta^\downarrow \in D^P_{k + h - 1}$ and $\zeta \in D^P_{k + h}$ while $k + h \geq k + h - 1 \geq k + d - 1$, so by Proposition \ref{prop: uparrow-property}, we have $(\zeta^\downarrow)^{\uparrow (k + d - 1)} = \zeta^{\uparrow (k + d - 1)}$ and then $(\chi^{\uparrow (k + d - 1)})^p = (\zeta^{\uparrow (k + d - 1)})^p$. On the other hand, by Proposition \ref{prop: brothers} (the Sibling Property) again, there is $\eta \in R_{\C C_1}(\gamma)$ such that 
\begin{itemize}
\item $\chi = \eta^\downarrow \in R_{\C C}(\eta_t)$,

\item $R_{\C D}(\eta) = R_{\C D}(\eta_t)$ for every $\C D \in \PP(\C C_2)$.
\end{itemize}
Similarly, since $\chi = \eta^\downarrow \in D^P_{k + h - 1}$ and $\eta \in D^P_{k + h}$ while $k + h \geq k + h - 1 \geq k + d - 1$, by Proposition \ref{prop: uparrow-property}, we have $\chi^{\uparrow (k + d - 1)} = \eta^{\uparrow (k + d - 1)}$. Therefore, $(\eta^{\uparrow (k + d - 1)})^p = (\zeta^{\uparrow (k + d - 1)})^p$. Hence, the formulas $\eta$ and $\zeta$ are what Step 2-1 requires.

Denote by $(M, \C U_{\C X})^{\gamma, \xi, k, d}_4$ the submodel constructed via the first four steps.

\begin{itemize}
\item[Step 5] For every $\C B \in \C P^+(\C X)$ and every $t \in T^*_{\C B}$, if there is $i \in \C B^c$ such that $R'_i(t') \neq \varnothing$,, consider the multi-pointed submodel 
$(M', \C U^{t'}_{\C B^c})$, where $\C U^{t'}_{\C B^c} = \{T^{t'}_{\C D} \mid T^{t'}_{\C D} = R'_{\C D}(t') \text{ for }\C D \in \PP(\C B^c)\}$. Let's discuss the world $t$.
\begin{itemize}
\item $t$ is constructed in Step 1: Then, $(\eta_t^{\uparrow (k + d)})^p = (\zeta_t^{\uparrow (k + d)})^p$. Since $k \leq d$, then $k - 1 \leq d$. By the inductive hypothesis, a multi-pointed model 
\[(M^t, \C U^t_{\C B^c})^{\eta_t, \zeta_t, (k - 1), d} = (S^t, R^t, V^t, \C U^t_{\C B^c})\]
is constructed, where $\C U^{t}_{\C B^c} = \{T^{t}_{\C D} \mid \C D \in \PP(\C B^c)\}$, such that
\begin{itemize}
\item $T^t_{\C D}$ is $\C D$-equivalent in $(M^t, \C U^t_{\C B^c})^{\eta_t, \zeta_t, (k - 1), d}$, for $\C D \in \PP(\C B^c)$;

\item $\rho_t: (M^t, \C U^t_{\C B^c})^{\eta_t, \zeta_t, (k - 1), d} \bisim^{col}_p (M', \C U^{t'}_{\C B^c})$;

\item $T^t_{\C D}$ is $R_{\C D}(\eta_t)^{\uparrow (k - 1)}$-complete, for $\C D \in \PP(\C B^c)$.
\end{itemize}

\item $t$ is constructed in Step 2-1: Then, $(\eta_t^{\uparrow (k + d - 1)})^p = (\zeta_t^{\uparrow (k + d - 1)})^p$. Since $k \leq d$, then $k - 1 \leq d - 1$. By the inductive hypothesis, 
\[(M^t, \C U^t_{\C B^c})^{\eta_t, \zeta_t, (k - 1), (d - 1)} = (S^t, R^t, V^t, \C U^t_{\C B^c})\]
is constructed, where $\C U^{t}_{\C B^c} = \{T^{t}_{\C D} \mid \C D \in \PP(\C B^c)\}$, such that
\begin{itemize}
\item $T^t_{\C D}$ is $\C D$-equivalent in $(M^t, \C U^t_{\C B^c})^{\eta_t, \zeta_t, (k - 1), (d - 1)}$, for $\C D \in \PP(\C B^c)$;

\item $\rho_t: (M^t, \C U^t_{\C B^c})^{\eta_t, \zeta_t, (k - 1), (d - 1)} \bisim^{col}_p (M', \C U^{t'}_{\C B^c})$;

\item $T^t_{\C D}$ is $R_{\C D}(\eta_t)^{\uparrow (k - 1)}$-complete, for $\C D \in \PP(\C B^c)$.
\end{itemize}
\end{itemize}
Add the submodel, $(M^t, \C U^t_{\C B^c})^{\eta_t, \zeta_t, (k - 1), d}$ or $(M^t, \C U^t_{\C B^c})^{\eta_t, \zeta_t, (k - 1), (d - 1)}$, to $(S, R, V)$, and extend $\rho$ to be $\rho \cup \rho_t$. 

(We write ``$(M^t, \C U^t_{\C B^c})^{\eta_t, \zeta_t, (k - 1), d}$'' or ``$(M^t, \C U^t_{\C B^c})^{\eta_t, \zeta_t, (k - 1), (d - 1)}$'' as ``$(M^t, \C U^t)$'' if it is clear from context.)

\item[Step 6] For every $\C B \in \C P^+(\C X)$, every $t \in T^*_{\C B}$, every $i \in \C B^c$, let $R_i(t)_4$ be the members of $R_i(t)$ within $(M, \C U_{\C X})^{\gamma, \xi, k, d}_4$. For every $u \in R_i(t)_4 \cup \{t\}$ such that $(M^u, \C U^u)$ exists, $T^u_i$ exists, which is an $i$-equivalent set. Then, merge all these sets $T^u_i$ and $R_i(t)_4 \cup \{t\}$ with the relation $R_i$.
\end{itemize}
Then, together with the other steps, $(M, \C U_{\C X})^{\gamma, \xi, k, d}$ is constructed for the inductive step.

Let $(M, \C U_{\C X}) = (M, \C U_{\C X})^{\gamma, \xi, k, d}$ and the construction is finished.
\end{proof}

Then, we conclude $\san{K45}_n \F D$ and $\san{KD45}_n \F D$ meet the sufficient condition of Theorem \ref{thm: guideline}.
\begin{lem}[Lemma for $\san{K45}_n \F D$ and $\san{KD45}_n \F D$]\label{lem: K45KD45}
Let $\delta \in D^P_{k + 1}(\san L)$ for $k \geq 0$. For every $\gamma \in D^P_{2k + 1}(\san L)$ and every $\san L$-model $(M', s')$ of $\gamma^p$, if $\gamma^{\uparrow (k + 1)} = \delta$, then there is an $\san L$-model $(M, s)$ such that
\begin{itemize}
\item $M, s \vDash \delta$,

\item $(M, s) \bisim^{col}_p (M', s')$.
\end{itemize}
\end{lem}
\begin{proof}
Consider the multi-pointed submodel of $(M', s')$
\[(M', \C U'_{\C A}) = (S', R', V', \C U'_{\C A})\]
where $\C U'_{\C A} = \{T'_{\C B} = R'_{\C B}(s') \mid \C B \in \PPA\}$. It is trivially true that $(\gamma^{\uparrow (k + k + 1)})^p = (\gamma^{\uparrow (k + k + 1)})^p$. Then, by Lemma \ref{lem: ConstructionK45}, we can construct a model
\[(M, \C U_{\C A})^{\gamma, \gamma, k, k} = (S, R, V, \C U_{\C A}),\]
where $\C U_{\C A} = \{T_{\C B} \subseteq S \mid \C B \in \PPA\}$, such that 
\begin{itemize}
\item $T_{\C B}$ is $\C B$-equivalent in $(M, \C U_{\C A})$, for every $\C B \in \PPX$;

\item $(M, \C U_{\C A})^{\gamma, \gamma, k, k} \bisim^{col}_p (M', \C U'_{\C A})$;

\item $T_{\C B}$ is $R_{\C B}(\gamma)^{\uparrow k}$-complete, for every $\C B \in \PPX$.
\end{itemize}
Note that $R_{\C B}(\gamma)^{\uparrow k} = R_{\C B}(\delta)$ exactly, so $T_{\C B}$ is $R_{\C B}(\delta)$-complete. Let's write $(M, \C U_{\C A})^{\gamma, \gamma, k, k}$ as $(M, \C U_{\C A})$ for simplicity.

Add a new actual world $s$ to $S$ such that for every $q \in P$, $q \in V(s)$ if and only if $w(\delta) \vDash q$. Add $(s, s') \in \rho$. Add $(s, t) \in R_i$ for every $t \in T_{\C B}$, every $i \in \C B$, and every $\C B \in \PPA$. Thus, we obtain the pointed $\san L$-model 
\[(M, s) = (S, R, V, s).\]
It is easy to see that $(M, s) \bisim^{col}_p (M', s')$. 

Observe that $M, s \vDash w(\delta)$ and for every $\C B \in \PPA$, $T_{\C B}$ is exactly $R_{\C B}(s)$, so $R_{\C B}(s)$ is $R_{\C B}(\delta)$-complete. Therefore, $M, s \vDash \delta$.
\end{proof}

\begin{thm}\label{thm: K45KD45}
$\san{K45}_n \F D$ and $\san{KD45}_n \F D$ have the uniform interpolation property.
\end{thm}
\begin{proof}
Let $\san L$ be $\san{K45}_n \F D$ or $\san{KD45}_n \F D$. Let $P \subseteq \C P$ be finite and $p$ be an atom. Observe that every member $\delta$ of $D^P_0(\san L)$, which is a minterm, has a uniform interpolant $\delta^p$. For $k \geq 0$ and $\delta \in D^P_{k + 1}(\san L)$, by Lemma \ref{lem: K45KD45} and Theorem \ref{thm: guideline}, $dforget_{\san L}(\delta, p)$ exists. More precisely,
\[dforget_{\san L}(\delta, p) \equiv_{\san L} \bigvee \{\gamma^p \mid \gamma \in D^P_{2k + 1} \text{ and } \gamma^{\uparrow (k + 1)} = \delta\}\]
By Proposition \ref{prop: UI-forgetting} and Proposition \ref{prop: delta-phi}, $\san{K45}_n \F D$ and $\san{KD45}_n \F D$ have the uniform interpolation property.
\end{proof}

\subsection{Uniform interpolation in $\san{S5}_n \F D$}\label{subsec: S5}
\begin{restatable}[The Construction Lemma for $\san{S5}_n \F D$]{lem}{ConstructionSfive}\label{lem: ConstructionS5}
Let $k, d, h$ be natural numbers such that $0 \leq k \leq d \leq h$. Let $\gamma, \xi \in D^P_{k + h + 1}(\san{S5}_n \F D)$ such that $(\gamma^{\uparrow (k + d + 1)})^p = (\xi^{\uparrow (k + d + 1)})^p$. Let $(M', s') = (S', R', V', s')$
be an $\san{S5}_n \F D$-model of $\xi^p$. 

Given $\C X \in \PPA$ and the multi-pointed submodel
\[(M', \C U'_{\C X}) = (S', R', V', \C U'_{\C X})\]
where $\C U'_{\C X} = \{T'_{\C B} \mid T'_{\C B} = R'_{\C B}(s')\text{ for }\C B \in \PPX\}$, there is a multi-pointed $\san{S5}_n \F D$-model 
\[(M, \C U_{\C X}) = (S, R, V, \C U_{\C X})\]
where $\C U_{\C X}= \{T_{\C B} \subseteq S \mid \C B \in \PPX\}$ such that, 
\begin{itemize}
\item $T_{\C B}$ is $\C B$-equivalent, for every $\C B \in \PPX$;

\item $(M, \C U_{\C X}) \bisim^{col}_p (M', \C U'_{\C X})$;

\item $T_{\C B}$ is $R_{\C B}(\gamma)^{\uparrow k}$-complete, for every $\C B \in \PPX$.
\end{itemize}
\end{restatable}
\begin{proof}
There are only a few nuances in the construction of $\san{S5}_n \F D$-models from the case of $\san{K45}_n \F D$. That is, Step 2-2 of the inductive step should be rewritten as:
\begin{itemize}
\item[Step 2-2] For every $\C B \in \PPX$, every $t \in T^*_{\C B}$, add $(t, t) \in R_j$ for every $j \in \C A$, and subsequently, for every $i \in \PP(\C B^c)$, and every $u_1, u_2 \in R_i(t)$, then add $(u_1, u_2)$ to $R_i$.
\end{itemize}
We leave the full proofs in the appendix. 
\end{proof}

\begin{lem}[Lemma for $\san{S5}_n\F D$]\label{lem: S5}
Let $\delta \in D^P_{k + 1}(\san{S5}_n\F D)$ for $k \geq 0$. For every $\gamma \in D^P_{2k + 1}(\san{S5}_n \F D)$ and every $\san{S5}_n\F D$-model $(M', s')$ of $\gamma^p$, if $\gamma^{\uparrow (k + 1)} = \delta$, then there is an $\san{S5}_n\F D$-model $(M, s)$ such that
\begin{itemize}
\item $M, s \vDash \delta$,

\item $(M, s) \bisim^{col}_p (M', s')$.
\end{itemize}
\end{lem}
\begin{proof}
Analogous to Lemma \ref{lem: K45KD45}, we add a new actual world $s$ to the multi-pointed $\san{S5}_n\F D$-model $(M, \C U_{\C A})$, but in particular, we also add $(s, s)$ to $R_i$ for each $i \in \C A$. It is easy to see that $(M, s) \bisim^{col}_p (M', s')$ and $(M, s)$ is an $\san{S5}_n\F D$-model.

For every $\C B \in \PPA$, $s' \in R_{\C B}(s')$. By the construction, there is $t_s \in T_{\C B}$ such that $(t_s, s') \in \rho$. Then, $V(t_s) = V(s)$. Since $(M, s) \bisim^{col}_p (M', s')$, we have that $t_s$ and $s$ shares the same relations. So, $T_{\C B} \cup \{s\}$ is still $\C B$-equivalent and $R_{\C B}(\delta)$-complete. Therefore, $M, s \vDash \delta$.
\end{proof}

Therefore, we also have the conclusion.
\begin{thm}\label{thm: S5}
$\san{S5}_n\F D$ has the uniform interpolation property.
\end{thm}
\begin{proof}
This proof is analogous to Theorem \ref{thm: K45KD45}.
\end{proof}

\subsection{Special cases where $\delta^p$ is a uniform interpolant}\label{subsec: SpecialCase}

By Theorem \ref{thm: K45KD45} and Theorem \ref{thm: S5}, we can instantly have the following corollary:
\begin{cor}
Let $\san L$ be any of $\san{K45}_n \F D$, $\san{KD45}_n \F D$, and $\san{S5}_n \F D$. For every $\delta \in D^p_1(\san L)$, $\delta^p \equiv_{\san L} dforget_{\san L}(\delta, p)$.
\end{cor}
\begin{proof}
When $k = 0$, then $\delta \in D^P_{1}(\san L)$ and $dforget_{\san L}(\delta, p) \equiv_{\san L} \bigvee \{\gamma^p \mid \gamma \in D^P_1 \text{ and } \gamma^{\uparrow 1} = \delta\} = \delta^p$.
\end{proof}

Consider the case where $\C A$ only has two agents, that is, $\C A = \{1, 2\}$. Thus, the three systems become $\san{K45}_2 \F D$, $\san{KD45}_2 \F D$, and $\san{S5}_2 \F D$, and we still denote them by $\san L_2$. We will show that $\delta^p \equiv_{\san L_2} dforget_{\san L_2}(\delta, p)$ for any $\san L_2$-satisfiable $\delta$. This conclusion is based on the much simplified version of the Sibling Property:
\begin{prop}\label{prop: A2-main-lemma}
Let $\delta \in D^P_k(\san L_2)$ for $k \geq 2$ and $(M', s')$ be any $\san L_2$-model of $\delta^p$. For every $i, j \in \C A$ with $i \neq j$, every $\eta \in R_i(\delta)$, every $\chi \in R_{\C A}(\eta)$, every $t' \in R'_i(s')$, and every $u' \in R'_{\C A}(t')$, if $M', t' \vDash \eta^p$ and $M', u' \vDash \chi^p$, then there is $\eta_0 \in R_{i}(\delta)$ such that
\begin{itemize}
\item $\eta_0^\downarrow = \chi \in R_{\C A}(\eta)$,

\item $R_j(\eta) = R_j(\eta_0)$,

\item $M', u' \vDash \eta_0^p$.
\end{itemize}
\end{prop}
\begin{proof}
By Proposition \ref{prop: brothers} (the Sibling Property), there is $\eta_0 \in R_{i}(\delta)$ such that
\begin{itemize}
\item $\eta_0^\downarrow = \chi \in R_{\C A}(\eta)$,

\item $R_j(\eta) = R_j(\eta_0)$
\end{itemize}
For the model $(M', s')$, there is $\eta_1 \in R_i(\delta)$ such that $M', u' \vDash \eta_1^p$ and $(\eta_1^p)^\downarrow = \chi^p$ and $R'_j(\eta^p) = R'_j(\eta_1^p)$. Furthermore, by Proposition \ref{prop: dist-identical-son}, $R_i(\eta^p) = R_i(\eta_1^p) = R_i(\delta^p)^{\downarrow}$ and $R_i(\eta) = R_i(\eta_0) = R_i(\delta)^{\downarrow}$. Thus, we can see $w(\chi^p) = w(\eta_1^p) = w(\eta_0^p)$, $R_i(\eta_1^p) = R_i(\eta_0^p)$, and $R_j(\eta_1^p) = R_j(\eta_0^p)$. Since $\C A = \{1, 2\}$, we have that $\eta_1^p = \eta_0^p$. Hence, $M', u' \vDash \eta_0^p$.
\end{proof}

\begin{restatable}{lem}{LTwo}\label{lem: L2K45}
Let $\delta \in D^P_{k + 1}(\san L_2)$ for $k \geq 0$ and $(M', s') = (S', R', V', s')$
be an $\san L_2$-model of $\delta^p$. Given $\C X \in \PPA$ and
\[(M', \C U'_{\C X}) = (S', R', V', \C U'_{\C X})\]
where $\C U'_{\C X} = \{T'_{\C B} \mid T'_{\C B} = R'_{\C B}(s') \text{ for } \C B \in \PPX\}$, there is a multi-pointed $\san {L}_2$-model 
\[(M, \C U_{\C X}) = (S, R, V, \C U_{\C X})\]
where $\C U_{\C X}= \{T_{\C B} \subseteq S \mid \C B \in \PPX\}$, and $\rho \subseteq S \times S'$, such that, 
\begin{itemize}
\item $T_{\C B}$ is $\C B$-equivalent in $(M, \C U_{\C X})$, for every $\C B \in \PPX$;

\item $\rho: (M, \C U_{\C X}) \bisim^{col}_p (M', \C U'_{\C X})$;

\item $T_{\C B}$ is $R_{\C B}(\delta)$-complete, for every $\C B \in \PPX$.
\end{itemize}
\end{restatable}

\begin{prop}
For every $k \geq 0$ and $\delta \in D^P_k(\san L_2)$, 
\[\delta^p \equiv_{\san L_2} dforget_{\san L_2}(\delta, p).\]
\end{prop}

\begin{proof}
This is proved via Lemma \ref{lem: K45KD45} (Lemma \ref{lem: S5}) and Lemma \ref{lem: L2K45}
\end{proof}

\section{Uniform Interpolation in Propositional Common Knowledge}\label{sec: CommonKnowledge}
In this section, we extend our conclusions to the epistemic modal logic with both distributed knowledge and propositional common knowledge, namely $\C L^{\F D}_\F{PC}$. Our current methods are also available for the systems $\san{K}_n\F{DPC}$, $\san{D}_n\F{DPC}$, $\san{T}_n\F{DPC}$, and $\san{S5}_n\F{DPC}$, however, not for $\san{KD45}_n\F{DPC}$ or $\san{K45}_n\F{DPC}$. We will briefly explain the reasons for the last two systems. For convenience, let's uniformly denote by $P$ the discussed finite subset of $\C P$ and $p$ is an atom in $P$ to be forgotten.

Let's use $\nabla \Phi$ to denote the formula
\[\F C (\bigvee \Phi) \land (\bigwedge \hat{\F C} \Phi)\]
For any model $(M, s)$, $M, s \vDash \nabla \Phi$, if and only if, for all $t \in \TC(s)$, there is $\phi \in \Phi$ such that $M, t \vDash \phi$, and for all $\phi \in \Phi$, there is $t \in \TC(s)$ such that $M, t \vDash \phi$.

\begin{defn}[dpc-canonical formulas]
Define $C^P_k$ to be the set inductively as follows:
\begin{itemize}
\item $C^P_0 = \{\delta_0 \land \nabla \Phi \mid \delta_0 \in D^P_0 \text{ and }\Phi \subseteq D^P_0\}$.

\item $C^P_{k+1}$ consist of formulas in the form 
\[\delta_0 \land \nabla \Phi \land \bigwedge_{\C B \in \PPA}\nabla_{\C B} \Phi_{\C B}\]
where $\delta_0 \in D^P_0$, $\Phi_{\C B} \subseteq C^P_k$, and $\Phi \subseteq D^P_0$.
\end{itemize}
\end{defn}
Let's call them \emph{canonical formulas of distributed knowledge and propositional common knowledge}, or \emph{dpc-canonical formulas} for short. Similar to our early conventions, if $\delta$ is a dpc-canonical formula in the form above, let's write $w(\delta) = \delta_0$, $R_{\C B}(\delta) = \Phi_{\C B}$, and $\TC(\delta) = \Phi$. For any dpc-canonical formulas $\delta$ and $\eta$, we say $\eta$ \emph{wholly occurs in} $\delta$, if either $\eta$ is $\delta$ or there exists $\C B \in \PPA$ and $\eta_0 \in R_{\C B}(\delta)$ such that $\eta$ wholly occurs $\eta_0$. Let's denote by $C^P_k(\san L)$ the members of $C^P_k$ that are satisfiable in the modal system $\san L$. Since we only consider the propositional common knowledge, any formula of the form $\F C \psi$ has a propositional $\psi$. To simplify the proofs, we slightly modify the definition of modal depth. For every propositional common knowledge formula $\F C \psi$, let
\begin{itemize}
\item $dep(\F C \psi) = dep(\psi) = 0$.
\end{itemize}

It is also natural to extend the concepts of pruning and literal elimination to the dpc-canonical formulas. The only noteworthy point is that ``$\nabla \TC(\delta)$'' is always accompanied by $w(\delta)$, which will never be pruned. That is.
\[
\delta^\downarrow = 
\begin{cases}
\delta & \text{if $k = 0$;}\\
w(\delta)\land \nabla \TC(\delta) & \text{if $k = 1$;}\\
w(\delta) \land \nabla \TC(\delta) \land \bigwedge_{\C B \in \PPA}\nabla_{\C B} (R_{\C B}(\delta))^\downarrow & \text{otherwise.}\\
\end{cases}
\]
Let $\delta^{\uparrow 0} = w(\delta)\land \nabla \TC(\delta)$ and for every $l \geq 1$, let
\[
\delta^{\uparrow l} = 
\begin{cases}
\delta & \text{if $dep(\delta) < l$}\\

w(\delta) \land \nabla \TC(\delta) \land \bigwedge_{\C B \in \PPA}\nabla_{\C B} R_{\C B}(\delta)^{\uparrow (l - 1)} & \text{ otherwise}
\end{cases}
\]

And then, we can extend the five important propositions to dpc-canonical formulas, too.
\begin{restatable}{prop}{Dpcdichotomy}\label{prop: dpc-dichotomy}
Consider the context of a modal system $\san L$. Let $\delta \in C^P_k(\san L)$ where $k \in \B N$. For every $\phi$ of $\C L^{\F D}_\F{PC}$ such that $dep(\phi) \leq k$ and $\C P(\phi) \subseteq P$, either $\delta \vDash \phi$ or $\delta \vDash \neg \phi$. 
\end{restatable}

\begin{restatable}{prop}{DpcModelsUniCan}\label{prop: dpc-models-unique-canon}
Let $(M, s)$ be a pointed model and $k \in \B N$. Then, there exists a unique $\delta \in C^P_k$ such that $M, s \vDash \delta$. 
\end{restatable}

\begin{restatable}{prop}{DpcUniCan}\label{prop: dpc-unique-canon}
Consider the context of a modal system $\san L$. Let $\phi \in \C L^\F{D}_\F{PC}$, $k \geq dep(\phi)$, and $P = \C P(\phi)$. Then, there exists a unique set $\Phi \subseteq C^P_k(\san L)$ such that $\phi \equiv \bigvee \Phi$.
\end{restatable}

\begin{restatable}{prop}{DpcDistcutinto}
Let $P \subseteq \C P$ be finite and let $\delta \in C^P_k$. Then, for any $l < k$ and $\C B \in \PPA$,
\[R_{\C B}(\delta^{\downarrow l}) = R_{\C B}(\delta)^{\downarrow l}.\]
\end{restatable}

\begin{restatable}{prop}{DpcUparrowproperty}
Consider the context of a modal system $\san L$. For every $\delta \in C^P_{k}(\san L)$ and $k, l \in \B N$, if $k \geq l$, then the following properties hold:
\begin{enumerate}
\item $\delta \vDash \delta^{\downarrow l}$.

\item $\delta^{\downarrow l} \in C^P_{k - l}(\san L)$.

\item $\delta^{\uparrow l} \in C^P_l(\san L)$.

\item For every $h \in \B N$ and $\gamma \in C^P_{k + h}(\san L)$, $\gamma^{\downarrow h} = \gamma^{\uparrow k} = (\gamma^{\downarrow h})^{\uparrow k}$.

\item For $k_1, k_2 \in \B N$, if $l \leq \min\{k, k_1, k_2\}$, then $
(\delta^{\uparrow k_1})^{\uparrow l} = (\delta^{\uparrow k_2})^{\uparrow l}
$.
\end{enumerate}
\end{restatable}

Furthermore, we can deconstruct propositional common knowledge successors as follows.
\begin{restatable}{prop}{ReachSplit}\label{prop: Reach-split}
For every satisfiable $\delta \in C^P_k$ where $k \geq 1$ and every $l \geq 1$,
\[\TC(\delta^{\uparrow l}) = \bigcup_{\C B \in \PPA}\bigcup_{\eta \in R_{\C B}(\delta)} \{w(\eta^{\uparrow (l - 1)})\} \cup \TC(\eta^{\uparrow (l - 1)})\]
\end{restatable}

\begin{restatable}{thm}{DPCGuide}\label{thm: dpc-guideline}
For every $k \geq 0$ and every $\delta \in C^P_k(\san L)$, if the following holds:
\begin{quote}
There is $l \geq k$ such that for every $\gamma \in C^{P}_{l}(\san L)$ and every $\san L$-model $(M', s')$ of $\gamma^p$, if $\gamma^{\uparrow k} = \delta$, then there is $(M, s)$ such that
\begin{itemize}
\item $M, s \vDash \delta$,

\item $(M, s) \bisim^{col}_p (M', s')$.
\end{itemize}
\end{quote}
Then, $dforget_{\san L}(\delta, p) \equiv_{\san L} \bigvee \{ \gamma \in C^{P}_{l}(\san L) \mid \gamma^{\uparrow k} = \delta\}^p$.
\end{restatable}

\subsection{Uniform interpolation in $\san{K}_n\F{DPC}$, $\san{D}_n\F{DPC}$, $\san{T}_n\F{DPC}$}
\begin{lem}[Lemma for $\san{K}_n\F{DPC}$, $\san{D}_n\F{DPC}$]\label{lem: dpc-KandD}

Let $\san L$ be $\san{K}_n\F{DPC}$ or $\san{D}_n\F{DPC}$, and $\delta \in C^P_k(\san L)$. For any model $(M', s')$, if $M', s' \vDash \delta^p$, then there is an $\san L$-model $(M, s)$ such that $M, s \vDash \delta$ and  $(M, s) \bisim^{col}_p (M', s')$.
\end{lem}
\begin{proof}
Let's prove by induction on $k$. Initially, let $S$, $V$, $\rho$, and $R_i$ for $i \in \C A$ be empty.

\F{Base case} ($k = 0$):
\begin{itemize}
\item Construct $s$ in $S$, such that for every $q \in P$, $q \in V(s)$ if and only if $w(\delta) \vDash q$. Let $(s, s') \in \rho$

\item For every $t' \in \TC(s')$ and every minterm $\chi \in \TC(\delta)$, if $M', t' \vDash \chi^p$, construct a new world $t$ in $S$ such that for every $q \in P$, $q \in V(t)$ if and only if $\chi \vDash q$. Let $(t, t') \in \rho$

\item For every $u, v \in S$ and every $i \in \C A$, if $(u', v')\in R'_i$, then add $(u, v)$ to $R_i$.
\end{itemize}
It is easy to see that $M, s \vDash \delta$ and $(M, s) \bisim^{col}_p (M', s')$ when $k = 0$. Furthermore, if $(M', s')$ is serial, then $(M, s)$ is also serial.

\F{Inductive steps}: The construction is analogous to that of Lemma \ref{lem: KandD}: For every $\C B \in \PPA$, every $\eta \in R_{\C B}(\delta)$ and every $t'\in R_{\C B}(s')$, if $M', t' \vDash \eta^p$, there has been $t \in R_{\C B}(s)$ and a submodel $(M^t, t)$ such that $M^t, t \vDash \eta$ and $\rho_t: (M^t, t) \bisim^{col}_p (M', t')$. Let $\rho$ be the union of $\{(s, s')\}$ and every inductively constructed $\rho_t$. It is easy to check $\rho: (M, s) \bisim^{col}_p (M', s')$. Since there is no edge added from every $t$ to other parts of $M, s$, we have $M, t \vDash \eta$. Thus, we have $M, s \vDash w(\delta) \land \bigwedge_{\C B \in \PPA} \nabla_{\C B}R_{\C B}(\delta)$ as in Lemma \ref{lem: KandD}. We still need to prove $M, s \vDash \nabla \TC(\delta)$:
\begin{itemize}
\item For every $u \in \TC(s)$, there is $t \in R_{\C B}(s)$ such that $u$ is in the submodel $(M^t, t)$. If $u$ is $t$, then $M, u \vDash w(\eta_t)$. By Proposition \ref{prop: Reach-split}, $w(\eta_t) \in \TC(\delta)$. If $u \in \TC(t)$, since $M, t \vDash \eta_t$, there is a minterm $\chi \in \TC(\eta_t)$ such that $M, u \vDash \chi$. By Proposition \ref{prop: Reach-split}, $\chi \in \TC(\eta_t) \subseteq \TC(\delta)$.

\item For every minterm $\chi \in \TC(\delta)$, by Proposition \ref{prop: Reach-split}, there is $\eta \in R_{\C B}(\delta)$ and either $\chi$ is $w(\eta)$ or in $\TC(\eta)$. There is submodel $(M^t, t)$ such that $M^t, t \vDash \eta$. There is $u$ in $(M^t, t)$ such that $M, u \vDash \chi$. It is easy to see $u \in \TC(s)$.
\end{itemize}
Therefore, $M, s \vDash \nabla \TC(\delta)$ and thus, $M, s \vDash \delta$.

In this construction, for every $i \in \C A$, every $t, t'$ with $(t, t') \in \rho$, when $R'_{i}(t')$ is nonempty, so is $R_{i}(t)$. So, when $\san L$ is $\san{D}_n \F{DPC}$, $(M, s)$ is also a $\san{D}_n \F{DPC}$-model (a serial model).

Hence, this lemma is proved.
\end{proof}

In parallel with Lemmas \ref{lem: reflexive} and \ref{lem: T}, the following lemmas for $\san T_n \F{DPC}$ also hold. We place the proof in the appendix.
\begin{restatable}{lem}{DPCReflexive}\label{lem: dpc-reflexive}
Let $\delta \in C^P_k(\san T_n \F{DPC})$ for $k \geq 1$. Let $l \in \{1, \dots, k\}$. Then, for all $\C B \in \PPA$, we have 
\begin{itemize}
\item $\delta^{\downarrow l} \in R_{\C B}(\delta^{\downarrow (l-1)})$,

\item $\delta^{\uparrow (l - 1)} \in R_{\C B}(\delta^{\downarrow l})$,

\item $w(\delta) \in \TC(\delta)$.

\end{itemize}
\end{restatable}

\begin{restatable}[Lemma for $\san{T}_n\F{DPC}$]{lem}{DPCT}\label{lem: dpc-T}

Let $\delta \in C^P_k(\san{T}_n\F{DPC})$. For any $\san{T}_n\F{DPC}$-model $(M', s')$, if $M', s' \vDash \delta^p$, then there is a $\san{T}_n\F{DPC}$-model $(M, s)$ such that $M, s \vDash \delta$ and $(M, s) \bisim^{col}_p (M', s')$.
\end{restatable}

Finally, we can conclude for the first three systems:
\begin{thm}\label{thm: dpc-KDT}
$\san{K}_n\F{DPC}$, $\san{D}_n\F{DPC}$, $\san{T}_n\F{DPC}$ have the uniform interpolation property.
\end{thm}
\begin{proof}
Let $\san L$ be any of the three systems. By Theorem \ref{thm: dpc-guideline}, Lemma \ref{lem: dpc-KandD}, and Lemma \ref{lem: dpc-T}, for every $p$ and every $\delta \in C^P_k(\san L)$ for $k \geq 0$, $\delta^p \equiv_{\san L} dforget_{\san L}(\delta, p)$. By Proposition \ref{prop: delta-phi}, for every $\san L$-satisifaible formula $\phi$ has its uniform interpolant over $P \setminus \{p\}$. More precisely,
\[dforget_{\san L}(\phi, p) \equiv_{\san L} \bigvee\{\delta^p \mid \delta \vDash_{\san L} \phi \text{ and } \delta \in C^P_{dep(\phi)}(\san L)\}\]
Therefore, $\san{K}_n\F{DPC}$, $\san{D}_n\F{DPC}$, $\san{T}_n\F{DPC}$ have the uniform interpolation property.
\end{proof}


\subsection{Uniform interpolation in $\san{S5}_n\F{DPC}$}
In this subsection, let's focus on $\san{S5}_n\F{DPC}$. This system's Euclideanness and reflexivity ensure that, every dpc-canonical formula shares common knowledge with the dpc-canonical formulas that wholly occur in it.
\begin{prop}\label{prop: KD45S5Reach}\label{prop: ShareCommonKnowledge}
For any $\san{S5}_n\F{DPC}$-satisfiable dpc-canonical formulas $\delta$ and $\eta$, if $\eta$ wholly occurs in $\delta$, then $\TC(\eta) = \TC(\delta)$.
\end{prop}
\begin{proof}
Let's prove by induction.
\begin{itemize}
\item \F{Base case}: $\eta = \delta$. It is trivial that $\TC(\eta) = \TC(\delta)$.

\item \F{Inductive steps}: There is $i \in \C A$ and $\eta_0 \in R_i(\delta)$ such that $\eta$ wholly occurs in $\eta_0$. Let $M, s$ be any $\san L$-model of $\delta$. Since $\eta_0 \in R_i(\delta)$, there are $t \in R_i(s)$ such that $M, t \vDash \eta_0$. Since $R_i$ is reflexive, $(s, s) \in R_i$; since $R_i$ is Euclidean, $(t, s) \in R_i$. Therefore, $\TC(s) = \TC(t)$. Since $M, s \vDash \nabla \TC(\delta)$ and $M, t \vDash \nabla \TC(\eta_0)$, we have $\TC(\eta_0) = \TC(\delta)$. By the inductive hypothesis, $\TC(\eta_0) = \TC(\eta)$. Therefore, $\TC(\eta) = \TC(\delta)$.
\end{itemize}
Hence, this proposition is proved.
\end{proof}

Propositions \ref{prop: dist-identical-son} (the Identical Successors Property) and \ref{prop: brothers} (the Sibling Property) can be extended to $\C L^\F{D}_\F{PC}$. Their proofs are not influenced by the addition of common knowledge, so we do not repeat them here.
\begin{prop}\label{prop: dpc-identical-son}
Let $\delta$ be a $C^P_k(\san{S5}_n \F{DPC})$ for  $k \geq 2$. Let $l \in \B N$ s.t. $1 \leq l < k$. Then, for all $\C B \in \PPA$ and $\eta \in R_{\C B}(\delta)$
\[R_{\C B}(\delta)^{\downarrow l} = R_{\C B}(\eta)^{\downarrow l - 1}\]
\end{prop}

%

\begin{prop}\label{prop: dpc-brothers}
Let $\delta \in C^P_k(\san{S5}_n \F{DPC})$ where $k \geq 2$. Let $\C B, \C C \in \PPA$ such that $\C C_1 = \C C \cap \C B \neq \varnothing$ and $\C C_2 = \C C \cap \C B^c \neq \varnothing$. For every $\eta \in R_{\C B}(\delta)$ and $\zeta \in R_{\C C}(\eta)$, there is $\eta_0 \in R_{\C C_1}(\delta)$ such that 
\begin{itemize}
\item $\eta_0^{\downarrow} = \zeta \in R_{\C C}(\eta)$,

\item $R_{\C D}(\eta_0) = R_{\C D}(\eta)$ for every $\C D \in \PP(\C C_2)$,
\end{itemize}
\end{prop}

Then, we have a construction lemma in parallel with Lemma \ref{lem: ConstructionK45} for $\san{S5}_n \F{DPC}$.
\begin{restatable}[The Construction Lemma for  $\san{S5}_n \F{DPC}$]{lem}{ConstructionCommonSfive}\label{lem: ConstructionCommonSfive}
Let $k, d, h$ be natural numbers such that $0 \leq k \leq d \leq h$. Let $\gamma, \xi \in C^P_{k + h + 1}(\san{S5}_n \F{DPC})$ such that $(\gamma^{\uparrow (k + d + 1)})^p = (\xi^{\uparrow (k + d + 1)})^p$. Let $(M', s') = (S', R', V', s')$
be an $\san{S5}_n \F{DPC}$-model of $\xi^p$. 

Given $\C X \in \PPA$ and the multi-pointed submodel
\[(M', \C U'_{\C X}) = (S', R', V', \C U'_{\C X})\]
where $\C U'_{\C X} = \{T'_{\C B} \mid T'_{\C B} = R'_{\C B}(s')\text{ for }\C B \in \PPX\}$, there is a multi-pointed $\san{S5}_n \F{DPC}$-model 
\[(M, \C U_{\C X}) = (S, R, V, \C U_{\C X})\]
where $\C U_{\C X}= \{T_{\C B} \subseteq S \mid \C B \in \PPX\}$ such that, 
\begin{itemize}
\item $T_{\C B}$ is $\C B$-equivalent in $(M, \C U_{\C X})$, for every $\C B \in \PPX$;

\item $(M, \C U_{\C X}) \bisim^{col}_p (M', \C U'_{\C X})$;

\item $T_{\C B}$ is $R_{\C B}(\gamma)^{\uparrow k}$-complete, for every $\C B \in \PPX$.
\end{itemize}
\end{restatable}
\begin{proof}
Let only show the nuanced steps of construction. The full proof is placed in the appendix.

This construction inherits the steps of Lemma \ref{lem: ConstructionK45}, and we only need to modify the following steps:

\F{Base case}:
\begin{itemize}
\item[Step 5] For every $\C B \in \C P^+(\C X)$ such that $\C B \neq \C A$ and every $t \in T^*_{\C B}$, consider the multi-pointed submodel 
$(M', \C U^{t'}_{\C B^c})$, where $\C U^{t'}_{\C B^c} = \{R'_{\C D}(t') \mid \C D \in \PP(\C B^c)\}$. Construct $(M^t, \C U^t_{\C B^c})$ as follows: Initially, let $S^t$, $V^t$, $\rho_t$, and every edge be empty.
\begin{itemize}
\item For every $u' \in \TC(t')$ and every minterm $\chi \in \TC(\eta_t)$, if $M', u' \vDash \chi^p$, then construct $u$ in $S$ such that for every $q \in P$, $q \in V(u)$ if and only if $\chi \vDash q$. Let $(u, u') \in \rho_t$.

\item For every $\C C \in \PPA$ and every $u, v$ in $S^t$, if $(u', v') \in R'_{\C C}$, add $(u, v)$ to $R_i$ for $i \in \C C$.

\item For every $\C D \in \PP(\C B^c)$, let $U^{t*}_{\C D} = \{u \mid (u, u') \in \rho_t \text{ and } u' \in R_{\C D}(t')\}$ and 
\[U^t_{\C D} = \bigcup_{\C D \subseteq \C D_0}U^{t*}_{\C D_0}\]
Let $\C U^t_{\C B^c} = \{U^t_{\C D} \mid \C D  \in \PP(\C B^c)\}$.

\item For all the worlds $u'$ in $(M', s')$ that are unreachable from $t'$, add the induced submodel $(M', u')$ to $S^t$ as copies, and this submodel inherits their relations and valuation in $M'$ and $(v, v') \in \rho_t$ for every $v'$ and its copy $v$ in the copied submodel.
\end{itemize}
\end{itemize}
For convenience, let's still call all the worlds in $(M^t, \C U^t_{\C B^c})$ the \emph{copies} of their counterparts in $M'$, and call $(M^t, \C U^t_{\C B^c})$ a copy of $(M', \C U^{t'}_{\C B^c})$. Observe that $(M^t, \C U^t_{\C B^c})$ is still an $\san{S5}_n \F{DPC}$-model and every set $U^t_{\C D}$ is $\C D$-equivalent. Recall that we denote the set $\RT(t, R_i \cup R_i^{- 1})$ within $(M, \C U_{\C X})^{\gamma, \xi, 0, 0}_4$ by $\RT(t, R_i \cup R_i^{- 1})_4$. 
\begin{itemize}
\item[Step 6] For every $\C B \in \C P^+(\C X)$, every $t \in T^*_{\C B}$, every $i \in \C B^c$, every $u \in \RT(t, R_i \cup R_i^{- 1})_4$, $U^u_i$ in $(M^u, \C U^u)$ is an $i$-equivalent set. Then, merge all these sets $U^u_i$ and $\RT(t, R_i \cup R_i^{- 1})_4$ with the relation $R_i$.
\end{itemize}

\F{Inductive step}:
\begin{itemize}
\item[Step 2-2] For every $\C B \in \PPX$, every $t \in T^*_{\C B}$, add $(t, t) \in R_j$ for every $j \in \C A$, and subsequently, for every $i \in \PP(\C B^c)$, and every $u_1, u_2 \in R_i(t)$, then add $(u_1, u_2)$ to $R_i$.
\end{itemize}
\end{proof}

\begin{lem}[Lemma for $\san{S5}_n \F{DPC}$]\label{lem: dpc-S5}
Let $\delta$ be any $\san{S5}_n \F{DPC}$-satisfiable dpc-canonical formula.
\begin{itemize}
\item If $\delta \in C^P_0$, then for every $\san L$-model $(M', s')$ of $\delta^p$, there is an $\san{S5}_n \F{DPC}$-model $(M, s)$ such that $M, s \vDash \delta$ and $(M, s) \bisim^{col}_p (M', s')$.

\item Suppose $\delta \in C^P_{k + 1}$ where $k \geq 0$. For every $\gamma \in C^P_{2k + 1}(\san{S5}_n \F{DPC})$ and every $\san{S5}_n \F{DPC}$-model $(M', s')$ of $\gamma^p$, if $\gamma^{\uparrow (k + 1)} = \delta$, then there is an $\san{S5}_n \F{DPC}$-model $(M, s)$ such that $M, s \vDash \delta$ and $(M, s) \bisim^{col}_p (M', s')$.
\end{itemize}
\end{lem}
\begin{proof}
Suppose $\delta \in C^P_0$ and $M', s' \vDash \delta^p$. Let's construct a new actual world $s$ such that for every $q \in P$, $q \in V(s)$ if and only if $w(\delta) \vDash q$. Then, construct a model $(M^s, \C U^s_{\C A})$ as in Step 5 of the base case of Lemma \ref{lem: ConstructionCommonSfive}, such that $(M^s, \C U^s_{\C A}) \bisim^{col}_p (M', \C U'_{\C A})$. Then concatenate $(M^s, \C U^s_{\C A})$ to $s$ such that for any $t$ in $(M^s, \C U^s_{\C A})$ and any $\C B \in \PPA$, $(s, t) \in R_{\C B}$ if and only if $(s', t') \in R'_{\C B}$. It is easy to see that $(M, s)$ such that $M, s \vDash \delta$ and $(M, s) \bisim^{col}_p (M', s')$.

The proof for the case ``$\delta \in C^P_{k + 1}$'' is analogous to that of Lemma \ref{lem: K45KD45} and Lemma \ref{lem: S5}, since we have proved Lemma \ref{lem: ConstructionCommonSfive}.
\end{proof}

\begin{thm}\label{thm: dpc-KD45S5}
$\san{S5}_n\F{DPC}$ has the uniform interpolation property.
\end{thm}
\begin{proof}
Every member $\delta$ of $C^P_0(\san{S5}_n\F{DPC})$, by Lemma \ref{lem: dpc-S5} and Theorem \ref{thm: dpc-guideline}, has a uniform interpolant $\delta^p$. 

For $k \geq 0$ and $\delta \in C^P_{k + 1}(\san{S5}_n\F{DPC})$, by Lemma \ref{lem: dpc-S5} and Theorem \ref{thm: dpc-guideline}, $dforget_{\san{S5}_n\F{DPC}}(\delta, p)$ exists. More precisely,
\[dforget_{\san{S5}_n\F{DPC}}(\delta, p) \equiv_{\san{S5}_n\F{DPC}} \bigvee \{\gamma^p \mid \gamma \in C^P_{2k + 1} \text{ and } \gamma^{\uparrow (k + 1)} = \delta\}\]
Hence, $\san{S5}_n\F{DPC}$ has the uniform interpolation property.
\end{proof}

\subsection{Limitations: $\san{K45}_n\F{DPC}$ and $\san{KD45}_n\F{DPC}$}
The construction lemmas are \emph{not} available for $\san{K45}_n\F{DPC}$ or $\san{KD45}_n\F{DPC}$. Then, the conclusion of uniform interpolation cannot be extended to the two systems.

We conjecture that $\san{KD45}_n\F{DPC}$ has the uniform interpolation property. However, since every model of $\san{KD45}_n\F{DPC}$ is not necessarily reflexive, Proposition \ref{prop: ShareCommonKnowledge} no longer holds. Then, every two constructed worlds may not have the same common knowledge but Euclideanness and transitivity may still connect them. We conjecture that Step 2 of the base case can be revised and then a construction lemma for $\san{KD45}_n\F{DPC}$ is possible. That possible proof has been different from the current ones, so it is beyond the scope of this paper.

The case of $\san{K45}_n\F{DPC}$ is even tougher. Observe Step 5 of the base case: 
\begin{itemize}
\item[Step 5] For every $\C B \in \C P^+(\C X)$ such that $\C B \neq \C A$ and every $t \in T^*_{\C B}$, consider the multi-pointed submodel 
$(M', \C U^{t'}_{\C B^c})$, ...
\end{itemize}
Recall that every model is serial in the cases of $\san{KD45}_n\F{DPC}$ or $\san{S5}_n\F{DPC}$, so $R'_j(t')$ is nonempty for $j \in \C B^c$ and then $(M', \C U^{t'}_{\C B^c})$ exists in $(M', s')$. However, this is no longer ensured in $\san{K45}_n\F{DPC}$. It is likely that for every $j \in \C B^c$, $R'_j(t') = \varnothing$, and then $(M', \C U^{t'}_{\C B^c})$ does not exist at all. Then, the construction will not proceed. Let's see a simple example. Let $P = \{p. q\}$ and $\C A$ contains at least two agents: 1 and 2.
\begin{itemize}
\item Let $\eta = p \land q \land \nabla \{p \land q, \neg p \land q\}$.

\item Let $\gamma = p \land \neg q \land \nabla \{p \land q, \neg p \land q\} \land \nabla_{\{1\}}\{\eta\} \land \bigwedge_{\C B \in \PPA \text{ and }\C B \neq \{1\}} \nabla_{\C B} \varnothing$.
\end{itemize}
Note that both $\gamma$ and $\eta$ are $\san{K45}_n\F{DPC}$-satisfiable, and $\TC(\gamma) = \TC(\eta) = \{p \land q, \neg p \land q\}$. Consider the model $(M', s')$ as follows:
\begin{itemize}
\item $S' = \{s', t'\}$

\item $R'_1 = \{(s', t'), (t', t')\}$ and the other relations are empty.

\item $V'$ is given so that $M', s' \vDash \neg q$ while $M', t' \vDash q$.
\end{itemize}
It is easy to check that $M', s' \vDash \gamma^p$. When we want to construct $(M, \C U_{\C A})^{\gamma, \gamma, 0, 0}$ as we did in Lemma \ref{lem: dpc-S5}, we will find that $S = T_1$ and there is no submodel added in Step 5. It is because $R'_{\C B}(t') = \varnothing$ for every $\C B \neq \{1\}$. For every $t \in T_1$, $M, t \vDash p \land q$, and there is no world in $S$ satisfying the minterm $\neg p \land q$, but $\neg p \land q \in \TC(\gamma)$. The model construction fails. Therefore, the conclusion of uniform interpolation becomes unavailable for $\san{K45}_n\F{DPC}$. The uniform interpolation property for this system is left open.

\section{Conclusion}\label{sec: Conclusion}

We have established the conclusions that distributed knowledge modal logics have the uniform interpolation property. Extending both the bisimulation-quantifier approach and several syntactic techniques—canonical formulas, literal elimination, and the pruning operation—we have shown how uniform interpolants can be constructed as the results of forgetting. This strengthens and generalizes earlier results on uniform interpolation in multi-agent modal logics by incorporating the distributed knowledge modality and also the propositional common knowledge cases, which are summarized in Table \ref{tab: summary}. Beyond the question of uniform interpolation in distributed knowledge, our analysis also clarifies the role of canonical forms in handling epistemic modalities.  

The proofs give rise to an open problem: in $\san{K45}_n\F D$, $\san{KD45}_n\F D$, and $\san{S5}_n\F D$, is $2k + 1$ a strict lower bound on the depth necessary for a uniform interpolant of a formula of depth $k + 1$? Furthermore, the current constructions are not practically efficient, so a natural next step is to seek algorithmic refinements that allow for more computationally feasible construction of interpolants.  Another promising direction is to investigate how these methods extend to richer epistemic languages, such as those with group knowledge or dynamic epistemic logics.

\begin{table}[t]
\centering
\begingroup
\setlength{\arrayrulewidth}{1pt}
\renewcommand{\arraystretch}{1.15}

\begin{tabular}{@{} l | *{6}{c} @{}}
\hline
 & $\san{K}$ & $\san{D}$ & $\san{T}$ & $\san{K45}$ & $\san{KD45}$ & $\san{S5}$ \\
\hline
$\C L^\F{K}_n$ & \circsurd & \circsurd & \circsurd & \circsurd & \circsurd & \circsurd \\
$\C L_\F{D}$ & \circsurd & \circsurd & \circsurd & $\surd$ & $\surd$ & $\surd$ \\
$\C L_\F{D}$ (depth = 1) & \circsurd & \circsurd & \circsurd & \circsurd & \circsurd & \circsurd \\
$\C L_\F{D}$ ($|\C A| = 2$) & \circsurd & \circsurd & \circsurd & \circsurd & \circsurd & \circsurd \\
$\C L^\F{D}_\F{PC}$ & \circsurd & \circsurd & \circsurd & ?  & ? & $\surd$ \\
\hline
\end{tabular}

\endgroup
\caption{The first line contains the conclusions from \cite{fang2019forgetting} where $\C L^\F{K}_n$ is the multi-agent modal language without distributed knowledge or common knowledge. ``$\surd$'' stands for uniform interpolation. ``$\bigcirc$'' means every formula and its uniform interpolants share the same modal depth.}
\label{tab: summary}
\end{table}


\newpage
\appendix
\section{The deterred proofs}
\Adequacy*
\begin{proof}
We will prove by induction on the structure of $\phi$. Let $\rho$ be the collective $p$-bisimulation between $(M, s)$ and $(M', s')$.
\begin{itemize}
\item $\phi$ is an atom $q$ for $q \in \C P \setminus \{p\}$: Since $(s, s') \in \rho$ and by the Atoms condition, we have $M, s \vDash q \Longleftrightarrow M', s' \vDash q$.

\item $\phi$ is in the form $\neg \psi$: Suppose $M, s \vDash \neg \psi$. Then, $M, s \not\vDash \psi$. Observe that $p$ does not occur in $\psi$. By the inductive hypothesis, $M', s' \not\vDash \psi$. Then, $M', s' \vDash \neg \psi$. The other direction is symmetric. So, $M, s \vDash \phi \Longleftrightarrow M', s' \vDash \phi$.

\item $\phi$ is in the form $\psi \land \chi$: Suppose $M, s \vDash \psi \land \chi$. Then, $M, s \vDash \psi$ and $M, s \vDash \chi$. Observe that $p$ occurs neither in $\psi$ nor in $\chi$. By the inductive hypothesis, $M', s' \vDash \psi$ and $M', s' \vDash \chi$. Then, $M', s' \vDash \psi \land \chi$. The other direction is symmetric. So, $M, s \vDash \phi \Longleftrightarrow M', s' \vDash \phi$.

\item $\phi$ is in the form $\F D_{\C B} \psi$: Suppose that $(M, s) \not \vDash \phi$. Then, there exists $t \in R_{\C B}(s)$ such that $(M, t) \not \vDash \psi$. By $(M, s) \bisim^{col}_p (M', s')$, there exists $t' \in R'_{\C B}(s')$ s.t. $(t, t') \in \rho$. By the inductive hypothesis, $M', t' \vDash \neg \psi$. So, $M', s' \not \vDash \phi$. The other direction is symmetric. Therefore, $M, s \vDash \phi \Longleftrightarrow M', s' \vDash \phi$.

\item $\phi$ is in the form $\F C \psi$: Suppose that $(M, s) \not \vDash \phi$. Then, there exists $t \in \TC(s)$ such that $(M, t) \not \vDash \psi$. Then, there is a path
\[s = t_0, t_1, \dots, t_l = t\]
such that for every $i \in \{0, \dots, l - 1\}$, there is $\C B_i \in \PPA$ such that $(t_i, t_{i + 1}) \in R_{\C B_i}$. Let's proceed by induction on $l$.
\begin{itemize}
\item \F{Base case} ($l = 1$): By the Forth condition, there is $t' \in R_{\C B_0}(s') \subseteq \TC(s')$ such that $(t, t') \in \rho$.

\item \F{Inductive steps}: By the inductive hypothesis on $l$, there is $t'_{l - 1} \in \TC(s')$ such that $(t_{l - 1}, t'_{l - 1}) \in \rho$. By the Forth condition, there is $t' \in R_{\C B_0}(t'_{l - 1}) \subseteq \TC(s')$ such that $(t, t') \in \rho$.
\end{itemize}
By the inductive hypothesis, $M', t' \vDash \neg \psi$. So, $M', s' \not \vDash \phi$. The other direction is similar. Therefore, $M, s \vDash \phi \Longleftrightarrow M', s' \vDash \phi$.
\end{itemize}
Hence, the Adequacy Lemma is proved.
\end{proof}

\UIforgetting*
\begin{proof}
We need to verify that $\psi$ meets the definition of uniform interpolant (Definition \ref{defn: UniformInterpolation}). Let $\chi$ be any formula without $p$ occurring.
\begin{itemize}
\item Suppose $\psi \vDash_{\san L} \chi$. Let $(M, s)$ be any $\san L$-model of $\phi$. It is easy to see that $(M, s) \bisim^{col}_p (M, s)$. Then by the Forth condition of forgetting, $M, s \vDash \psi$. So, $\phi \vDash_{\san L} \psi$. Therefore, $\phi \vDash_{\san L} \chi$.

\item Suppose $\phi \vDash_{\san L} \chi$. For every $\san L$-model $(M', s')$ of $\psi$, by the Back condition forgetting, there is an $\san L$-model $(M, s)$ such that $M, s \vDash \phi$ and $(M, s)\bisim^{col}_p (M', s')$. Since $\phi \vDash_{\san L} \chi$, we have $M, s \vDash \chi$. Then by Lemma \ref{lem: col-p-bisim-adequacy} (the Adequacy Lemma), $M', s' \vDash \chi$. So, $\psi \vDash_{\san L} \chi$. 
\end{itemize}
Therefore, $\phi \vDash_{\san L} \chi$ if and only if $\psi \vDash_{\san L} \chi$.

Hence, $\psi$ is a uniform interpolant of $\phi$ in $\san L$ over $P \setminus \{p\}$.
\end{proof}

\DforgetVee*
\begin{proof}
``$\Rightarrow$''. Let $(M', s')$ be any $\san L$-model of $dforget_{\san L}(\phi_1 \lor \phi_2, p)$. By the Back condition of Definition \ref{defn: forgetting}, there is an $\san L$-model $(M, s)$ such that $(M, s) \bisim^{col}_p (M', s')$ and $M, s \vDash \phi_1 \lor \phi_2$. W.l.o.g., we assume $M, s \vDash \phi_1$. By the Forth condition of Definition \ref{defn: forgetting}, $M', s' \vDash dforget_{\san L}(\phi_1, p)$. Then, $M', s' \vDash dforget_{\san L}(\phi_1, p) \lor dforget_{\san L}(\phi_2, p)$.

``$\Leftarrow$''. Let $(M', s')$ be any $\san L$-model of $dforget_{\san L}(\phi_1, p) \lor dforget_{\san L}(\phi_2, p)$. W.l.o.g., we assume $M', s' \vDash dforget_{\san L}(\phi_1, p)$. By the Back condition of Definition \ref{defn: forgetting}, there is an $\san L$-model $(M, s)$ such that $(M, s) \bisim^{col}_p (M', s')$ and $M, s \vDash \phi_1$. Then, $M, s \vDash \phi_1 \lor \phi_2$. By the Forth condition of Definition \ref{defn: forgetting}, $M', s' \vDash dforget_{\san L}(\phi_1 \lor \phi_2, p)$.

Hence, $dforget_{\san L}(\phi_1 \lor \phi_2, p) \equiv_{\san L} dforget_{\san L}(\phi_1, p) \lor dforget_{\san L}(\phi_2, p)$.
\end{proof}

\Canondichotomy*
\begin{proof}
Let's prove by induction on $\phi$.

\F{Base:} Suppose $\phi$ is an atom. As $\C P(\phi) \subseteq P$ and $w(\delta)$ is a minterm of $P$, it is obvious that either $\delta \vDash_{\san L} \phi$ or $\delta \vDash_{\san L} \neg \phi$.

\F{Inductive step:} Suppose $\phi$ is in the form $\neg \psi$. By the inductive hypothesis, either $\delta \vDash_{\san L} \psi$ or $\delta \vDash_{\san L} \neg \psi$. If $\delta \vDash_{\san L} \neg \psi$, then $\delta \vDash_{\san L} \phi$. If $\delta \vDash_{\san L} \psi$, observe that in every modal system $\san L$ within our discussion, $\neg \phi = \psi$, so $\delta \vDash_{\san L} \neg \phi$.

Suppose $\phi$ is in the form $\psi_1\land \psi_2$. If $\delta \vDash_{\san L} \psi_1$ and $\delta \vDash_{\san L} \psi_2$, then $\delta \vDash_{\san L} \phi$. Otherwise, without loss of generality, we assume $\delta \vDash_{\san L} \neg \psi_1$. Then, $\delta \vDash_{\san L} \neg \psi_1 \lor \neg \psi_2$. Then, $\delta \vDash_{\san L} \neg (\psi_1 \land \psi_2)$.

Suppose $\phi$ is in the form $\F D_{\C B} \psi$. Note that $1 \leq dep(\phi) \leq k$. By the inductive hypothesis, there are two possible cases:
\begin{itemize}
\item There is $\eta_0 \in R_{\C B}(\delta)$ such that $\eta_0 \vDash_{\san L} \neg \psi$: For any $\san L$-model $(M, s)$, if $M, s \vDash \delta$, then
\[M, s \vDash \bigwedge \hat{\F D}_{\C B} R_{\C B}(\delta).\]
There exists $t_0 \in R_{\C B}(s)$ such that $M, t_0 \vDash \eta_0$. By $\eta_0 \vDash_{\san L} \neg \psi$, $M, t_0 \vDash \neg \psi$. So, $M, s \vDash \hat{\F D}_{\C B} \neg \psi$. Therefore, $M, s \vDash \neg \F D_{\C B} \psi$

\item For any $\eta \in R_{\C B}(\delta)$, $\eta \vDash_{\san L} \psi$: For any $\san L$-model $(M, s)$, if $M, s \vDash \delta$, then
\[M, s \vDash \F D_{\C B} \bigvee R_{\C B}(\delta).\]
For any $t \in R_{\C B}(s)$, there exists $\gamma \in R_{\C B}(\delta)$ such that $M, t \vDash \gamma$. Because ``for any $\eta \in R_{\C B}(\delta)$, $\eta \vDash_{\san L} \psi$'', we have that $\gamma \vDash_{\san L} \psi$. Note that $\bigvee R_{\C B}(\delta) \vDash_{\san L} \gamma$. So, $\bigvee R_{\C B}(\delta) \vDash_{\san L} \psi$.  According to Axiom $\san K$, we have $\F D_{\C B}\bigvee R_{\C B}(\delta) \vDash_{\san L} \F D_{\C B}\psi$. (In the special case where $R_{\C B}(s) = R_{\C B}(\delta) = \varnothing$, observe that $\nabla_{\C B} R_{\C B}(\delta) = \F D_{\C B} \bot$. As $\vDash_{\san L} \bot \to \psi$, we have $\vDash_{\san L} \F D_{\C B}(\bot \to \psi)$. According to Axiom $\san K$, we have that $\vDash_{\san L} \F D_{\C B}\bot \to \F D_{\C B}\psi$.) Therefore, $M, s \vDash \F D_{\C B} \psi$.
\end{itemize}
Therefore, either $\delta \vDash_{\san L} \phi$ or $\delta \vDash_{\san L} \neg \phi$.
\end{proof}

\Modelsuniquecanon*
\begin{proof}
Let's prove by induction on $k$.
\begin{itemize}
\item \F{Base case} ($k = 0$): Let $\delta$ be 
\[\bigwedge_{p \in P \text{ and } M, s \vDash p} p \land \bigwedge_{p \in P \text{ and } M, s \not \vDash p} \neg p\]
Obviously, this $\delta$ is unique with regard to $(M, s)$.

\item \F{Inductive step} ($k > 0$): Let $w(\delta)$ be 
\[\bigwedge_{p \in P \text{ and } M, s \vDash p} p \land \bigwedge_{p \in P \text{ and } M, s \not \vDash p} \neg p\]
For any $\C B \in \PPA$ and $t \in R_{\C B}(s)$, by the hypothesis, there exists a unique $\eta \in D^P_{k - 1}$ such that $M, t \vDash \eta$. Let $\Phi_{\C B}(\delta)$ be the set of these unique $\eta$'s in $D^P_{k - 1}$. In the special case where $R_{\C B}(s) = \varnothing$, let $\Phi_{\C B}(\delta) = \varnothing$. Then, let
\[\delta = w(\delta) \land \bigwedge_{\C B \in \PPA} \nabla_{\C B} \Phi_{\C B}(\delta)\]
Then, $M, s \vDash \delta$. Note that $w(\delta)$ is unique and by the inductive hypothesis, $\Phi_{\C B}(\delta)$ is unique with regard to $R_{\C B}(s)$. Therefore, $\delta$ is unique.
\end{itemize}
\end{proof}

\Uniquecanon*
\begin{proof}
The proof consists of the following two parts:
\begin{itemize}
\item \F{Existence}: Let $\Phi = \{\delta \mid \delta \vDash_{\san L} \phi \text{ and } \delta \in D^P_k(\san L)\}$. On the one hand, it is obvious that $\bigvee \Phi \vDash_{\san L} \phi$. On the other hand, for any model $(M, s)$, if $M, s \vDash \phi$, by Proposition \ref{prop: models-unique-canon}, there exists $\delta_0 \in D^P_k(L)$ such that $M, s \vDash \delta_0$. By Proposition \ref{prop: canon-dichotomy}, for any $\delta' \in D^P_k(\san L) - \Phi$, $\delta' \vDash_{\san L} \neg \phi$. So, $\delta_0 \in \Phi$. So, $M, s \vDash \bigvee \Phi$. Therefore, $\phi \equiv_{\san L} \bigvee \Phi$.

\item \F{Uniqueness}: Suppose there is another $\Phi'$ such that $\bigvee \Phi' \equiv_{\san L} \phi$. Let $\delta_1 \in \Phi' - \Phi$. By the construction of $\Phi$, $\delta_1 \vDash_{\san L} \neg \phi$. That means there is a model $(M_1, s_1)$ of $\delta_1$ such that $M_1, s_1 \vDash \bigvee \Phi'$ and $M_1, s_1 \vDash \neg \phi$ and then $M_1, s_1 \vDash \neg \phi \land \phi$. This is contradictory. Conversely, let $\delta_2 \in \Phi - \Phi'$ and $(M_2, s_2)$ be any model of $\delta_2$. Because $\delta_2 \in D^P_k(\san L)$ and by Proposition \ref{prop: models-unique-canon}, $\delta_2$ is the unique element of $D^P_k(\san L)$ such that $M_2, s_2 \vDash \delta_2$. So $M_2, s_2 \not \vDash \bigvee \Phi'$. Note that $M_2, s_2 \vDash \phi$. This is contradictory to ``$\bigvee \Phi' \equiv_{\san L} \phi$''. Therefore, $\Phi = \Phi'$.
\end{itemize}
Hence, this proposition is proved.

\end{proof}

\BCinclude*
\begin{proof}
For every model $(M, s)$, observe that
\[R_{\C C}(s) = \bigcap_{i \in \C C}R_i(s) \subseteq \bigcap_{i \in \C B}R_i(s) = R_{\C B}(s)\]
So, item 1 holds.

Suppose $\delta \in D^P_k$ such that $M, s \vDash \delta$ and $k \geq 1$. For every $\eta \in R_{\C C}(\delta)$, by the semantics, there is $t \in R_{\C C}(s)$ such that $M, t \vDash \eta$. As $R_{\C C}(s) \subseteq R_{\C B}(s)$, $t \in R_{\C B}(s)$. Then, there is $\eta_0 \in R_{\C B}(\delta) \subseteq D^P_{k - 1}$ such that $M, t \vDash \eta_0$. By Proposition \ref{prop: models-unique-canon}, $\eta = \eta_0$. Therefore, $R_\C{C}(\delta) \subseteq R_{\C B}(\delta)$. 
\end{proof}

\Distcutinto*
\begin{proof}
Let's prove by induction on $l$.
\begin{itemize}
\item \F{Base}: When $l = 1$, $R_{\C B}(\delta^{\downarrow}) = (R_{\C B}(\delta))^{\downarrow}$ by Definition \ref{defn: uparrow}.

\item \F{inductive step}: By the hypothesis, $R_{\C B}(\delta^{\downarrow l - 1}) = R_{\C B}(\delta)^{\downarrow l - 1}$. Then, 
\[R_{\C B}(\delta^{\downarrow l}) = R_{\C B}((\delta^{\downarrow l - 1})^\downarrow) = R_{\C B}(\delta^{\downarrow l - 1})^{\downarrow} = (R_{\C B}(\delta)^{\downarrow l - 1})^{\downarrow} =  (R_{\C B}(\delta))^{\downarrow l}.\]
\end{itemize}
Therefore, for any $l < k$ and $\C B \in \PPA$, $R_{\C B}(\delta^{\downarrow l}) = R_{\C B}(\delta)^{\downarrow l}$.
\end{proof}

\Uparrowproperty*
\begin{proof}
For item 1: Let's first prove that $\delta \vDash \delta^\downarrow$ by induction on $k$. When $k \leq 1$, then $w(\delta) = \delta^\downarrow$, so $\delta \vDash \delta^\downarrow$. Suppose $k > 1$ and $(M, s)$ be any model of $\delta$. Observe that $w(\delta^\downarrow) = w(\delta)$, so $M, s \vDash w(\delta^\downarrow)$. Let $\C B \in \PPA$. 
\begin{itemize}
\item For every $t \in R_{\C B}(s)$, since $M, s \vDash \nabla_{\C B}R_{\C B}(\delta)$, there is $\eta \in R_{\C B}(\delta) \subseteq D^P_{k - 1}(\san L)$ such that $M, t \vDash \eta$. Note that, $\eta^\downarrow \in R_{\C B}(\delta)^\downarrow = R_{\C B}(\delta^\downarrow)$. By the inductive hypothesis on $k$, $\eta \vDash \eta^\downarrow$. Therefore, $M, s \vDash \F D_{\C B} \bigvee R_{\C B}(\delta^\downarrow)$.

\item For every $\gamma \in R_{\C B}(\delta^\downarrow) = R_{\C B}(\delta)^\downarrow$, there is $\eta \in R_{\C B}(\delta) \subseteq D^P_{k - 1}(\san L)$ such that $\eta^\downarrow = \gamma$. By the inductive hypothesis on $k$, $\eta \vDash \gamma$. As $M, s \vDash \nabla_{\C B}R_{\C B}(\delta)$, there is $t \in R_{\C B}(s)$ such that $M, t \vDash \eta$. So, $M, t \vDash \gamma$. Therefore, $M, s \vDash \bigwedge \hat{\F D}_{\C B} R_{\C B}(\delta^\downarrow)$.
\end{itemize}
Therefore, $\delta \vDash \delta^\downarrow$. For every $l > 0$, by the definition, $\delta^{\downarrow l} = (\delta^{\downarrow l - 1})^\downarrow$, so $\delta^{\downarrow l - 1} \vDash \delta^{\downarrow l}$. Inductively, we have that $\delta \vDash \delta^{\downarrow l}$.

For item 2: Let's first prove ``$\delta^{\downarrow l} \in D^P_{k - l}$'' by induction on $l$. Suppose item 2 holds for every $l \in \{0, \dots, k - 1\}$. By Proposition \ref{prop: dist-cut-into}, 
\[
\delta^{\downarrow l} = (\delta^{\downarrow l - 1})^\downarrow
 						= w(\delta^{\downarrow l - 1}) \land \bigwedge_{\C B \in \PPA}\nabla_{\C B} (R_{\C B}(\delta^{\downarrow l - 1}))^\downarrow
\]
By the inductive hypothesis, $R_{\C B}(\delta^{\downarrow l - 1}) \subseteq D^P_{k - 1 - l}$. Therefore, $\delta^{\downarrow l} \in D^P_{k - l}$. Since $\delta$ is $\san L$-satisfiable and by item 1, we have $\delta^{\downarrow l} \in D^P_{k - l}(\san L)$.

For item 3: Let's first prove ``$\delta^{\uparrow l} \in D^P_l$'' by induction on $l$. If $l = 0$, by Definition \ref{defn: uparrow}, $\delta^{\uparrow 0} = w(\delta) \in D^P_0$. Suppose $l > 0$. By the inductive hypothesis, for every $\C B \in \PPA$, $R_{\C B}(\delta)^{\uparrow (l - 1)} \subseteq D^P_{l - 1}$. Then, $\delta^{\uparrow l} \in D^P_l$. By item 2, $\delta^{\downarrow (k - l)} \in D^P_{k - (k - l)}(\san L) = D^P_l(\san L) \subseteq D^P_l$. By Proposition \ref{prop: models-unique-canon}, $\delta^{\uparrow l} = \delta^{\downarrow (k - l)}$. Therefore, $\delta^{\uparrow l} \in D^P_l(\san L)$

For item 4: This is a corollary of items 2 and 3. Observe that $dep(\gamma^{\downarrow h}) \leq k$. By Definition \ref{defn: uparrow}, $(\gamma^{\downarrow h})^{\uparrow k} = \gamma^{\downarrow h}$. Therefore, $\gamma^{\downarrow h} = \gamma^{\uparrow k} = (\gamma^{\downarrow h})^{\uparrow k}$.

For item 5: By item 3, it is easy to see that $(\delta^{\uparrow k_1})^{\uparrow l}, (\delta^{\uparrow k_1})^{\uparrow l} \in D^P_l(\san L)$. By Proposition \ref{prop: models-unique-canon}, $(\delta^{\uparrow k_1})^{\uparrow l} = (\delta^{\uparrow k_2})^{\uparrow l}$.
\end{proof}

\Deltapimply*
\begin{proof}
Let $(M, s)$ be a model of $\delta$. Let's prove by induction on $k$.
\begin{itemize}
\item \F{Base} ($k = 0$): $\delta$ is a minterm. $\delta^p$ is the remaining part of $\delta$ after deleting the literal $p$ or $\neg p$, so $M, s \vDash \delta^p$.

\item \F{inductive step} ($k > 0$): Similar to the inductive base, $M, s \vDash w(\delta^p)$. Let $\C B \in \PPA$.
\begin{itemize}
\item For every $\gamma \in R_{\C B}(\delta^p)$, by the definition, there is $\eta \in R_{\C B}(\delta) \subseteq D^P_{k - 1}$ such that $\eta^p = \gamma$. By the hypothesis, $\eta \vDash \eta^p$. Because $M, s \vDash \delta$, there is $t \in R_{\C B}(s)$ such that $M, t \vDash \eta$. So, $M, t \vDash \eta^p$. Therefore, $M, s \vDash \bigwedge \hat{\F{D}}_{\C B}R_{\C B}(\delta^p)$.

\item For every $t \in R_{\C B}(s)$, because $M, s \vDash \delta$, there is $\eta \in R_{\C B}(\delta)  \subseteq D^P_{k - 1}$ such that $M, t \vDash \eta$. By the hypothesis, $\eta \vDash \eta^p$. So, $M, t \vDash \eta^p$. Therefore, $M, s \vDash \F{D}_{\C B}\bigvee R_{\C B}(\delta^p)$.
\end{itemize}
Therefore, $M, s \vDash w(\delta^p) \land \bigwedge_{\C B \in \PPA} \nabla_{\C B} R_{\C B}(\delta^p)$.
\end{itemize}
Hence, $\delta \vDash \delta^p$.
\end{proof}

\Reflexive*
\begin{proof}
Since $\delta \in D^P_k(\san T_n \F D)$, there is a $\san T_n \F D$-model $(M, s)$ such that $M, s \vDash \delta$. 

By Proposition \ref{prop: uparrow-property}, we have $M, s \vDash \delta^{\downarrow (l - 1)}$ and $M, s \vDash \delta^{\downarrow l}$. Since $(M, s)$ is reflexive, then $s \in R_{\C B}(s)$ and there exists $\eta \in R_{\C B}(\delta^{\downarrow (l - 1)})$ such that $M, s \vDash \eta$. Note that both $\delta^{\downarrow l}$ and $\eta$ belong to $D^P_{k - l}$. By Proposition \ref{prop: models-unique-canon}, they are the unique canonical formula of $D^P_{k - l}$ for $(M, s)$, that is, $\delta^{\downarrow l} = \eta$. Therefore, $\delta^{\downarrow l} \in R_{\C B}(\delta^{\downarrow (l-1)})$. 

When $l \neq k$, we have $\delta^{\downarrow ((k - l) + 1)} \in R_{\C B}(\delta^{\downarrow (k - l)})$ and by Proposition \ref{prop: uparrow-property}, $\delta^{\uparrow (l - 1)} \in R_{\C B}(\delta^{\uparrow l})$; note that $\delta^{\uparrow (k - 1)} = \delta^{\downarrow} \in R_{\C B}(\delta) = R_{\C B}(\delta^{\uparrow k})$. Therefore, $\delta^{\uparrow (l - 1)} \in R_{\C B}(\delta^{\uparrow l})$.
\end{proof}


\Sfivecounter*
\begin{proof}
Let's construct $\delta^{k}_{cou}$ by an $\san{S5}_n \F D$-model $(M_k, s_k)$ and then let $\delta^{k}_{cou}$ is the unique member of $D^P_{k}(\san{S5}_n \F D)$ such that $M_k, s_k \vDash \delta^{k}_{cou}$. We write ``$\{x, y\} \in R_{\C B}$'' instead of ``$(x, y), (y, x) \in R_{\C B}$''. Consider the model $(M_2, s_2) = (S, R, V, s_2)$ satisfying the following:
\begin{itemize}
\item $S = \{s_2, t_1, t_2, t_3, t_4, v, w\}$.

\item For every $x, y \in \{s_2, t_1, t_2, t_3, t_4\}$, $\{x, y\} \in R_1$; $\{t_1, t_3\}, \{t_2, t_4\} \in R_2$; $\{t_1, w\}, \{t_4, v\} \in R_3$. For every $x \in S$, $(x, x) \in R_{\C A}$. These are all the accessibility relations.

\item $V$ is given by the following:
\begin{itemize}
\item $M_2, x \vDash \neg p \land \neg q$, for $x \in \{t_1, s_2, v\}$,

\item $M_2, t_2 \vDash p \land \neg q$,

\item $M_2, x \vDash \neg p \land q$, for $x \in \{t_3, w\}$,

\item $M_2, t_4 \vDash p \land q$.
\end{itemize}

\end{itemize}
Thus, $(M_2, s_2)$ is an $\san{S5}_n \F D$-model. Let $\delta^2_{cou}$ be the unique member of $D^P_2(\san{S5}_n \F D)$ such that $M_2, s_2 \vDash \delta^{2}_{cou}$. For $x \in \{1, 2, 3, 4\}$, let $\eta_x$ be the unique member of $R_1(\delta^2_{cou}) \subseteq D^P_1(\san{S5}_n \F D)$ such that $M, t_x \vDash \eta_x$. Let $\eta_0 = (\delta^2_{cou})^{\downarrow}$. Note that $\eta_0, \eta_1, \eta_2, \eta_3, \eta_4$ are mutully inequivalent.

Let's consider the model $(M', s') = (S', R', V', s')$ of $(\delta^2_{cou})^p$. $S'$ is a copy of $S$ and $s'$ is a copy of $s_2$. Every world $x' \in S'$ inherits the value of $x$ on $q$. Every accessibility relation is also inherited except that
\begin{itemize}
\item $\{t'_1, t'_4\} \in R'_2$ but $\{t'_1, t'_3\} \notin R'_2$,

\item $\{t'_2, t'_3\} \in R'_2$ but $\{t'_2, t'_4\} \notin R'_2$,
\end{itemize}
However, we can see that for every $x \in \{1, 2, 3, 4\}$, $M', t'_x \vDash \eta_x^p$ and $M', s' \vDash \eta_0^p$. The key point is that, for example, $R_{\{1, 2\}}(\eta_1) = \{\neg q, q\}$ but both $M', t'_3 \vDash q$ and $M', t'_4 \vDash q$. Therefore, we have $M', s' \vDash (\delta^2_{cou})^p$. Note that $\eta_0^p, \eta_1^p, \eta_2^p, \eta_3^p, \eta_4^p$ are still mutully inequivalent.

Suppose there exists $\san{K45}_n \F D$-model $(M, s)$ so that both $(M, s) \vDash \delta^2_{cou}$ and $(M, s) \bisim^{col}_p (M', s')$. Since $M, s \vDash \delta^2_\text{cou}$, define
\[H_x = \{t \in R_1(s) \mid M, t \vDash \eta_x\}\]
for $x \in \{0, 1, 2, 3, 4\}$. Since $\eta_0, \eta_1, \eta_2, \eta_3, \eta_4$ are mutually inequivalent and by Proposition \ref{prop: models-unique-canon}, $H_0, H_1, H_2, H_3, H_4$ are mutually disjoint.

Let $\rho$ be the relation that witnesses $(M, s) \bisim^{col}_p (M', s')$. Observe that, for any $x, y \in \{0, 1, 2, 3, 4\}$, if $x \neq y$, then $\eta_x \vDash \eta_x^p$ but $\eta_x \vDash \neg \eta_y^p$. Observe that, for any world $z \in R_1(s)$ and $z' \in S'$, if $(z, z') \in \rho$, then $M, z \vDash \eta_x \Longleftrightarrow M', z' \vDash \eta_x^p$. Therefore, we have
\begin{itemize}
\item $z \in H_0 \Longleftrightarrow (z, s') \in \rho$,

\item $z \in H_x \Longleftrightarrow (z, t'_x) \in \rho$.
\end{itemize}

Let $t$ be any world in $H_2$. Thus, $M, t \vDash \eta_2$ and $(t, t'_2) \in \rho$. Note that $p \land q \in R_{\{1, 2\}}(\eta_2)$. There is $u \in S$ such that $M, u \vDash p \land q$ and $u \in R_{\{1, 2\}}(t)$. Note that $u \in R_{\{1, 2\}}(t) \subseteq R_1(t)$. Since $(M, s)$ is a $\san{K45}_n$-model, the relation $R_1$ is transitive. So, $u \in R_1(s)$. Every world in $H_0 \cup H_1 \cup H_2 \cup H_3$ does not satisfy $p \land q$. Therefore, $u \in H_4$. Since $(t, t'_2) \in \rho$, by the Forth condition of $p$-bisimulation, there is $u' \in R'_{\{1, 2\}}(t'_2)$ such that $(u, u') \in \rho$. However, the only elements of $R'_{\{1, 2\}}(t'_2)$ are $t'_2$ and $t'_3$. So, $u'$ is either $t'_2$ or $t'_3$, and then $u \in H_2 \cup H_3$. This is contradictory to the fact that $u \in H_4$, because $H_0, H_1, H_2, H_3, H_4$ are mutually disjoint.

Hence, ``$(M, s) \vDash \delta^2_{cou}$'' and ``$(M, s) \bisim^{col}_p (M', s')$'' cannot hold simultaneously for any $\san{K45}_n$-model $(M, s)$.

When $k > 2$, let's add the following new worlds
\[s_3, \dots, s_k\]
to $S$ such that for $h \in \{3, \dots, k\}$
\begin{itemize}
\item $\{s_{h - 1}, s_h\} \in R_1$ if $h$ is even, while $\{s_{h - 1}, s_h\} \in R_2$ if $h$ is odd;

\item $p, q \notin V(s_h)$ if $h$ is even, while $p, q \in V(s_h)$ if $h$ is odd;

\item $(s_h, s_h) \in R_{\C A}$.
\end{itemize}
Let $M_k$ be the new model and $s_k$ be the new actual world. let $\delta^{k}_{cou}$ is the unique member of $D^P_{k}(\san{S5}_n \F D)$ such that $M_k, s_k \vDash \delta^{k}_{cou}$. We can similarly prove the conclusion for $\delta^{k}_{cou}$.

Since $\delta^k_{cou}$ breaches the Back condition of Definition \ref{defn: forgetting}, we have that $(\delta^k_{cou})^p \not \equiv_{\san L} dforget_{\san L}(\delta^k_{cou}, p)$, for every $k \geq 2$ and  $\san L$ among $\san{K45}_n \F D$, $\san{KD45}_n \F D$, and $\san{S5}_n \F D$. 
\end{proof}

\Morecounter*
\begin{proof}
Suppose there is $\san{S5}_n \F{D}$-satifiable $\gamma$ such that $\gamma^{\uparrow 2} = \delta^2_{cou}$ and $\gamma = \xi$. Let $(M, s) = (S, R, V, s)$ be any $\san{K45}_n \F{D}$-model of $\gamma$. Since $\gamma^{\uparrow 2} = \delta^2_{cou}$, then $M, s \vDash \delta^{2}_{cou}$. The formulas $\eta_0, \eta_1, \eta_2, \eta_3, \eta_4$ are the same as in Proposition \ref{prop: S5counter} and we also partition $R_1(s)$ into $H_0, H_1, H_2, H_3, H_4$. Let $t \in H_2$ and $\eta_t$ is the member of $R_1(\gamma)$ such that $M, t \vDash \eta_t$. Note that $\eta_t^{\uparrow} = \eta_2$. Because  $w(\eta_4) \in R_{\{1, 2\}}(\eta_2)$, then $R_{\{1, 2\}}(t) \cap H_4 \neq \varnothing$. Let $u \in R_{\{1, 2\}}(t) \cap H_4$ and $\eta_u$ is the member of $R_1(\gamma)$ such that $M, u \vDash \eta_u$. Since $u \in R_{\{1, 2\}}(t)$, we have $\eta_u^{\downarrow} \in R_{\{1, 2\}}(\eta_t)$; since $u \in H_4$, we have $\eta_u^{\uparrow} = \eta_4$. Since $\gamma = \xi$ and $M', s' \vDash \xi^p$, then we have $M', t'_2 \vDash \eta_t^p$. Since $\eta_u^{\downarrow} \in R_{\{1, 2\}}(\eta_t)$ and $R_{\{1, 2\}}(t'_2) = \{t'_2, t'_3\}$, either $M', t'_2 \vDash \eta_u^p$ or $M', t'_3 \vDash \eta_u^p$. Then, either $M', t'_2 \vDash \eta_4^p$ or $M', t'_3 \vDash \eta_4^p$. This is contradictory to the proof of Proposition \ref{prop: S5counter}. Therefore, $\gamma \neq \xi$.
\end{proof}

\begin{fact}\label{fact: shareSuccessors}
Let $M = (S, R, V)$ be a transitive and Euclidean model. For every $t_1, t_2 \in S$ and $\C B \in \PPA$, if $t_2 \in R_{\C B}(t_1)$, then $R_{\C B}(t_1) = R_{\C B}(t_2)$.
\end{fact}
\begin{proof}
For every $u \in R_{\C B}(t_2)$, since $(t_1, t_2), (t_2, u) \in R_{B}$ and $R_{\C B}$ is transitive, then $u \in R_{\C B}(t_1)$. $R_{\C B}(t_1) \supseteq R_{\C B}(t_2)$. For every $u \in R_{\C B}(t_1)$, since $(t_1, u), (t_1, t_2) \in R_{B}$ and $R_{\C B}$ is Euclidean, then $u \in R_{\C B}(t_2)$. So, $R_{\C B}(t_1) \subseteq R_{\C B}(t_2)$.

Therefore, $R_{\C B}(t_1) = R_{\C B}(t_2)$.
\end{proof}

\IdenticalChildren*
\begin{proof}
Let $(M, s)$ be any $\san{K45}_n\F D$-model of $\delta$. Then, there is $t \in R_{\C B}(s)$ such that $M, t \vDash \eta$. By Fact \ref{fact: shareSuccessors}, $R_{\C B}(s) = R_{\C B}(t)$. Note that $R_{\C B}(s)$ is $R_{\C B}(\delta)$-complete and $R_{\C B}(t)$ is $R_{\C B}(\eta)$-complete. Therefore, $R_{\C B}(s)$ is both $R_{\C B}(\delta)$-complete and $R_{\C B}(\eta)$-complete. Note that $R_{\C B}(\delta) \subseteq D^{P}_k$ and $R_{\C B}(\eta) \subseteq D^{P}_{k - 1}$. By Proposition \ref{prop: uparrow-property}, we have that $R_{\C B}(\delta)^{\downarrow l} \subseteq D^{P}_{k - l} $, $R_{\C B}(\eta)^{\downarrow l - 1} \subseteq D^{P}_{k - l}$, and $R_{\C B}(s)$ is both $R_{\C B}(\delta)^{\downarrow l}$-complete and $R_{\C B}(\eta)^{\downarrow l - 1}$-complete. Observe that by Proposition \ref{prop: models-unique-canon}, there is a unique subset $\Phi$ of $D^P_{k - l}$ that $R_{\C B}(s)$ is $\Phi$-complete. Hence, $R_{\C B}(\delta)^{\downarrow l} = R_{\C B}(\eta)^{\downarrow l - 1}$.
\end{proof}

\Brothers*
\begin{proof}
Let $(M, s)$ be any $\san L$-model of $\delta$. There is $t \in R_{\C B}(s)$ and $t_0 \in R_{\C C}(t)$ such that $M, t \vDash \eta$ and $M, t_0 \vDash \chi$. Note that $R_{\C C}(t) \subseteq R_{\C C_1}(t)$. Since $R_{\C C_1}$ is transitive, we have $t_0 \in R_{\C C_1}(s)$. Then, there is $\eta_0 \in R_{\C C_1}(\delta)$ such that $M, t_0 \vDash \eta_0$.
\begin{itemize}
\item Note that $R_{\C C}(\eta) \subseteq R_{\C C_1}(\eta)$ and by Proposition \ref{prop: dist-identical-son} (the Identical Successors Property), $R_{\C C_1}(\eta) = R_{\C C_1}(\delta)^\downarrow$. So $\chi \in R_{\C C_1}(\delta)^\downarrow$. Since $M, t_0 \vDash \chi$ and $M, t_0 \vDash \eta_0$, and by Proposition \ref{prop: models-unique-canon}, we have that $\chi = \eta_0^\downarrow$.

\item For every $\C D \in \PP(\C C_2)$, $t_0 \in R_{\C C}(t) \subseteq R_{\C C_2}(t) \subseteq R_{\C D}(t)$. Since $R_{\C D}$ is both transitive and Euclidean, then by Fact \ref{fact: shareSuccessors}, we have $R_{\C D}(t) = R_{\C D}(t_0)$. Then, $R_{\C D}(t)$ is $R_{\C D}(\eta)$-complete while $R_{\C D}(t_0)$ is $R_{\C D}(\eta_0)$-complete. Note that $\eta, \eta_0 \in D^P_{k - 1}$. By Proposition \ref{prop: models-unique-canon}, they are the unique subset of $D^P_{k - 2}$. Therefore, $R_{\C D}(\eta) = R_{\C D}(\eta_0)$
\end{itemize}
Hence, the proposition is proved.
\end{proof}

\QeTranEu*
\begin{proof}
``$1 \Rightarrow 2$''. Let $i \in \C A$ and $u, v, w$ be any worlds in $M$. 
\begin{itemize}
\item Suppose $(u, v), (v, w) \in R_i$. Then, $v, w \in \RT(u, R_i \cup R_i^{- 1})$. Because $\RT(u, R_i \cup R_i^{- 1})$ is quasi-$i$-equivalent, we have $R_i(u) = R_i(v)$. So, $w \in R_i(u)$, namely, $(u, w) \in R_i$.

\item Suppose $(u, v), (u, w) \in R_i$. Then, $v, w \in \RT(u, R_i \cup R_i^{- 1})$. Because $\RT(u, R_i \cup R_i^{- 1})$ is quasi-$i$-equivalent, we have $R_i(u) = R_i(v)$. Since $w \in R_i(u)$, then, $w \in R_i(v)$, namely, $(v, w) \in R_i$
\end{itemize}
Hence, $M$ is transitive and Euclidean.

``$2 \Rightarrow 1$''. Let $t$ be any world in $M$ and $i \in \C A$. Note that for every $u \in \RT(t, R_i \cup R_i^{- 1})$, $\RT(t, R_i \cup R_i^{- 1}) = \RT(u, R_i \cup R_i^{- 1})$. Let $t_1, t_2$ be any members of $\RT(t, R_i \cup R_i^{- 1})$. Then, there is a sequence 
\[t_1 = u_0, u_1, \dots u_{m_1}, u_m = t_2\]
such that for every $j \in \{0, \dots, m - 1\}$, either $(u_j, u_{j + 1}) \in R_i$ or $(u_{j + 1}, u_j) \in R_i$. Let's prove by induction on $m$:
\begin{itemize}
\item \F{Base case} ($m = 1$): Since $R_i$ is both transitive and Euclidean, by Fact \ref{fact: shareSuccessors}, we have $R_i(u_0) = R_i(u_1)$.

\item \F{Inductive step} ($m > 1$): Since $R_i$ is both transitive and Euclidean, then by Fact \ref{fact: shareSuccessors}, we have $R_i(u_{m - 1}) = R_i(u_m)$. By the hypoethesis, $R_i(u_0) = R_i(u_{m - 1})$. So, $R_i(u_0) = R_i(u_m)$.
\end{itemize}
Therefore, $R_i(t_1) = R_i(t_2)$. Hence, $\RT(t, R_i \cup R_i^{- 1})$ is quasi-$i$-equivalent.

``$3 \Rightarrow 4$''. By ``$1 \Rightarrow 2$'', $M$ is transitive and Euclidean. For every $t$ in $M$ and every $i \in \C A$, since $\RT(t, R_i \cup R_i^{- 1})$ is $i$-equivalent, then $(t, t) \in R_i$. So, $M$ is reflexive.

``$4 \Rightarrow 3$''. For every $t$ in $M$ and every $i \in \C A$, by ``$2 \Rightarrow 1$'', $\RT(t, R_i \cup R_i^{- 1})$ is quasi-$i$-equivalent. Since $M$ is reflexive, $(t, t) \in R_i$. Then, $\RT(t, R_i \cup R_i^{- 1})$ is $i$-equivalent.
\end{proof}

\begin{prop}
Let $M$ be any Kripke model, $\C B \in \PPA$, and $t$ be any world in $M$. 
\begin{itemize}
\item If for every $i \in \C B$, $\RT(t, R_i \cup R_i^{- 1})$ is quasi-$i$-equivalent, then $\RT(t, R_{\C B} \cup R_{\C B}^{- 1})$ is quasi-$\C B$-equivalent.

\item If for every $i \in \C B$, $\RT(t, R_i \cup R_i^{- 1})$ is $i$-equivalent, then $\RT(t, R_{\C B} \cup R_{\C B}^{- 1})$ is $\C B$-equivalent.
\end{itemize}
\end{prop}
\begin{proof}
Let $u \in \RT(t, R_{\C B} \cup R_{\C B}^{- 1})$. It is easy to see that $\RT(t, R_{\C B} \cup R_{\C B}^{- 1}) = \RT(u, R_{\C B} \cup R_{\C B}^{- 1})$. For $i \in \C B$, since $R_{\C B} \subseteq R_i$, then we have $\RT(t, R_{\C B} \cup R_{\C B}^{- 1}) \subseteq \RT(t, R_i \cup R_i^{- 1})$; since $\RT(t, R_i \cup R_i^{- 1})$ is quasi-$i$-equivalent, then $R_i(t) = R_i(u)$. Therefore, $R_{\C B}(t) = \bigcap_{i \in \C B}R_i(t) = \bigcap_{i \in \C B}R_i(u) = R_{\C B}(u)$. Therefore, $\RT(t, R_{\C B} \cup R_{\C B}^{- 1})$ is quasi-$\C B$-equivalent.

Furthermore, if for every $i \in \C B$, $\RT(t, R_i \cup R_i^{- 1})$ is $i$-equivalent, then $(u, u) \in R_i$ for each $i \in \C B$. Then, $(u, u) \in R_{\C B}$. Then, $\RT(t, R_{\C B} \cup R_{\C B}^{- 1})$ is $\C B$-equivalent.
\end{proof}

\begin{fact}\label{fact: merge}
Given a Kripke model $M$ and any quasi-$i$-equivalent sets $W_1, W_2, \dots, W_m$, 
\begin{enumerate}
\item for every $u_1, u_2 \in W_1$, if $(u_1, u_2) \notin R_i$ before merging, then $(u_1, u_2) \notin R_i$ after merging;

\item if there is at least one world $u_0 \in W_1 \cup W_2 \cup \dots \cup W_m$ such that $(u_0, u_0) \in R_i$, and we merge them with the relation $R_i$, then $W_1 \cup W_2 \cup \dots \cup W_m$ becomes quasi-$i$-equivalent;

\item if they are $i$-equivalent and we merge them with the relation $R_i$, then $W_1 \cup W_2 \cup \dots \cup W_m$ becomes $i$-equivalent.
\end{enumerate}
\end{fact}
\begin{proof}
For item 1, if $(u_1, u_2)$ is added to $R_i$ after merging, then $(u_2, u_2) \in R_i$ but $R_i(u_1) \neq R_i(u_2)$ before merging. This is contrary to the definition of quasi-$i$-equivalence.

For item 2, let $U_j = \{u \mid u \in W_j \text{ and }(u, u) \in R_i\}$ for every $j \in \{1, \dots, m\}$, and $U = \bigcup_{j \in \{1, \dots, m\}} U_j$.
So, $u_0 \in U$. After merging, for every $t \in W_1 \cup W_2 \cup \dots \cup W_m$, $(t, u_0) \in R_i$, so $W_1 \cup W_2 \cup \dots \cup W_m = \RT(u_0, R_i \cup R_i^{- 1}) = \RT(t, R_i \cup R_i^{- 1})$. So after merging, for every $t \in W_1 \cup W_2 \cup \dots \cup W_m$, we have $R_i(t) = U$. Therefore, $W_1 \cup W_2 \cup \dots \cup W_m$ becomes quasi-$i$-equivalent after merging.

For item 3, if each of $W_1, W_2, \dots, W_m$ is empty, then $W_1 \cup W_2 \cup \dots \cup W_m$ is vacuously $i$-equivalent; otherwise, by item 1, $W_1 \cup W_2 \cup \dots \cup W_m = U$, so $W_1 \cup W_2 \cup \dots \cup W_m$ becomes $i$-equivalent after merging.
\end{proof}

\ConstructionLemma*

\begin{proof}
Let's construct by induction on $k$ and $d$. We label the constructed model as $(M, \C U_{\C X})^{\gamma, \xi, k, d}$, in order to show the four parameters' change in the induction. The Greek letter $\eta$ ranges over the d-canonical formulas that wholly occur in $\gamma$, while the Greek letter $\zeta$ those that wholly occur in $\xi$.  Initially, let $S$, $V$, $\rho$, and every relation $R_i$ be empty.

\F{Base case ($k = d = 0$)}: Construct $(M, \C U_{\C X})^{\gamma, \xi, 0, 0}$ as follows.

\begin{itemize}
\item[Step 1] 
For every $\C B \in \PPX$, $t' \in R'_{\C B}(s')$, $\eta \in R_{\C B}(\gamma)$, and $\zeta \in R_{\C B}(\xi)$, if $(\eta^{\uparrow (k + d)})^p = (\zeta^{\uparrow (k + d)})^p$ and $M', t' \vDash \zeta^p$, then construct a world $t$ such that for every $q \in P$, $q \in V(t)$ if and only if $w(\eta) \vDash q$. Then, add $(t, t')$ to $\rho$. Let $T^*_{\C B}$ be all such $t$ with respect to this $\C B$.

(Note that when a world $t$ is constructed, we can determine its unique $t'$, $\eta$, and $\zeta$. Let's write this $\eta$ as $\eta_t$ and this $\zeta$ as $\zeta_t$. Let's say $t'$ (resp. $\eta_t$, $\zeta_t$) constructs $t$.)

\item[Step 2] 
For every $\C B_1, \C B_2 \in \PPX$, every $t_1 \in T^*_{\C B_1}$, $t_2 \in T^*_{\C B_2}$, and $i \in \C B_1^c \cap \C B_2^c$, if $(t'_1, t'_2) \in R'_i$, then add $(t_1, t_2)$ to $R_{i}$. 

\item[Step 3] For every $\C B_1, \C B_2 \in \PPX$, every $t_1 \in T^*_{\C B_1}$, $t_2 \in T^*_{\C B_2}$, and $i \in \C B_1 \cap \C B_2$, add $(t_1, t_2)$ to $R_{i}$. 

\item[Step 4] For every $\C B \in \PPX$, let
\[T_{\C B} = \bigcup_{\C B \subseteq \C C} T^*_{\C C}\]
Extend $S$ with $\bigcup_{\C B \in \PPX}T_{\C B}$ and let $\C U_{\C X} = \{T_{\C B} \mid \C B \in \PPX\}$.

Denote by $(M, \C U_{\C X})^{\gamma, \xi, 0, 0}_4$ the submodel constructed via the first four steps.

\item[Step 5] For every $\C B \in \C P^+(\C X)$ and every $t \in T^*_{\C B}$, if there is $i \in \C B^c$ such that $R'_i(t') \neq \varnothing$, consider the multi-pointed submodel 
$(M', \C U^{t'}_{\C B^c})$, where $\C U^{t'}_{\C B^c} = \{R'_{\C D}(t') \mid \C D \in \PP(\C B^c)\}$. Let $(M^t, \C U^t_{\C B^c})$ be a copy of $(M', \C U^{t'}_{\C B^c})$ and let $\rho_t = \{(u, u') \mid \text{$u$ is a copy of $u'$.}\}$.
Add the submodel $(M^t, \C U^t_{\C B^c})$ to $(S, R, V)$, and extend $\rho$ to be $\rho \cup \rho_t$. 

(Note that so far, no edge is added between $t$ and $(M^t, \C U^{t}_{\C B^c})$. We write ``$(M^t, \C U^{t}_{\C B^c})$'' as ``$(M^t, \C U^{t})$'' if it is clear from context.)

\item[Step 6] For every $\C B \in \C P^+(\C X)$, every $t \in T^*_{\C B}$, every $i \in \C B^c$, let's denote the set $\RT(t, R_i \cup R_i^{- 1})$ within $(M, \C U_{\C X})^{\gamma, \xi, 0, 0}_4$ by $\RT(t, R_i \cup R_i^{- 1})_4$. For every $u \in \RT(t, R_i \cup R_i^{- 1})_4$, if $(M^u, \C U^{u})$ exists, which is an $\san L$-model, there is the quasi-$i$-equivalent set $W^u_i$ in $(M^u, \C U^{u})$ such that $W^u_i \supseteq R'_i(u')$. Then, merge all these sets $W^u_i$ and $\RT(t, R_i \cup R_i^{- 1})_4$ with the relation $R_i$.
\end{itemize}
Thus, $(M, \C U_{\C X})^{\gamma, \xi, 0, 0}$ is constructed. Observe that this construction does not specify the value of $d$, so $(M, \C U_{\C X})^{\gamma, \xi, 0, d}$ is also constructed for any $d \geq 0$.

\F{Inductive step}: Steps 1, 3, and 4 are the same as those of the Base. We show Steps 2, 5, and 6, where Step 2 is divided into two substeps.

\begin{itemize}

\item[Step 2-1] 
For every $\C B \in \PPX$ and every $t \in T^*_{\C B}$ that is constructed in Step 1, for every $\C C \in \PPA$ such that $\C C_1 = \C B \cap \C C \neq \varnothing$ and $\C C_2 = \C B^c \cap \C C \neq \varnothing$, and every $u' \in R'_{\C C}(t')$, \emph{we claim that} there exists $\eta \in R_{\C C_1}(\gamma)$ and $\zeta \in R_{\C C_1}(\xi)$ such that
\begin{itemize}
\item $\eta^{\downarrow} \in R_{\C C}(\eta_t)$,

\item $R_{\C D}(\eta) = R_{\C D}(\eta_t)$ for every $\C D \in \PP(\C C_2)$,

\item $M', u' \vDash \zeta^p$

\item $(\eta^{\uparrow (k + d - 1)})^p =  (\zeta^{\uparrow (k + d - 1)})^p$.
\end{itemize}
Construct a new world $u$ such that $q \in P$, $q \in V(u)$ if and only if $w(\eta) \vDash q$; add $u$ to $T^*_{\C C_1}$. Add $(t, u)$ and $(u, u)$ to $R_i$ for every $i \in \C C_2$. Add $(u, u')$ to $\rho$. (We also say $u$ is constructed by $u'$, $\eta$, and $\zeta$, like those constructed in Step 1.)

\item[Step 2-2] For every $\C B \in \PPX$, every $t \in T^*_{\C B}$, every $i \in \PP(\C B^c)$, and every $u_1, u_2 \in R_i(t)$, then add $(u_1, u_2)$ to $R_i$.
\end{itemize}

Note that, we need to prove the claim that the formulas $\eta$ and $\zeta$ exist in Step 2-1. On the one hand, since $M', t' \vDash \zeta_t^p$, by Proposition \ref{prop: brothers} (the Sibling Property), there is $\zeta \in R_{\C C_1}(\xi)$ such that 
\begin{itemize}
\item $M', u' \vDash \zeta^p$,

\item $(\zeta^\downarrow)^p \in R_{\C C}(\zeta_t)^p$

\item $R_{\C D}(\zeta)^p = R_{\C D}(\zeta_t)^p$ for every $\C D \in \PP(\C C_2)$.
\end{itemize}
Note that since $t$ is constructed in Step 1, $(\eta_t^{\uparrow (k + d)})^p = (\zeta_t^{\uparrow (k + d)})^p$. So, there is $\chi \in R_{\C C}(\eta_t)$ such that $(\chi^{\uparrow (k + d - 1)})^p = ((\zeta^\downarrow)^p)^{\uparrow (k + d - 1)}$. Note that $\chi \in D^P_{k + h - 1}$ and $\zeta \in D^P_{k + h}$ while $k + h \geq k + h - 1 \geq k + d - 1$, so by Proposition \ref{prop: uparrow-property}, we have $(\chi^{\uparrow (k + d - 1)})^p = (\zeta^{\uparrow (k + d - 1)})^p$. On the other hand, by Proposition \ref{prop: brothers} (the Sibling Property) again, there is $\eta \in R_{\C C_1}(\gamma)$ such that 
\begin{itemize}
\item $\chi = \eta^\downarrow \in R_{\C C}(\eta_t)$,

\item $R_{\C D}(\eta) = R_{\C D}(\eta_t)^p$ for every $\C D \in \PP(\C C_2)$.
\end{itemize}
So, $\chi^{\uparrow (k + d - 1)} = \eta^{\uparrow (k + d - 1)}$. Therefore, $(\eta^{\uparrow (k + d - 1)})^p = (\zeta^{\uparrow (k + d - 1)})^p$. Hence, the formulas $\eta$ and $\zeta$ are what Step 2-1 requires.

Denote by $(M, \C U_{\C X})^{\gamma, \xi, k, d}_4$ the submodel constructed via the first four steps.

\begin{itemize}
\item[Step 5] For every $\C B \in \C P^+(\C X)$ and every $t \in T^*_{\C B}$, if there is $i \in \C B^c$ such that $R'_i(t') \neq \varnothing$,, consider the multi-pointed submodel 
$(M', \C U^{t'}_{\C B^c})$, where $\C U^{t'}_{\C B^c} = \{T^{t'}_{\C D} \mid T^{t'}_{\C D} = R'_{\C D}(t') \text{ for }\C D \in \PP(\C B^c)\}$. Let's discuss the world $t$.
\begin{itemize}
\item $t$ is constructed in Step 1: Then, $(\eta_t^{\uparrow (k + d)})^p = (\zeta_t^{\uparrow (k + d)})^p$. Since $k \leq d$, then $k - 1 \leq d$. By the inductive hypothesis, a multi-pointed model 
\[(M^t, \C U^t_{\C B^c})^{\eta_t, \zeta_t, (k - 1), d} = (S^t, R^t, V^t, \C U^t_{\C B^c})\]
is constructed, where $\C U^{t}_{\C B^c} = \{T^{t}_{\C D} \mid \C D \in \PP(\C B^c)\}$, such that
\begin{itemize}
\item $T^t_{\C D}$ is $\C D$-equivalent in $(M^t, \C U^t_{\C B^c})^{\eta_t, \zeta_t, (k - 1), d}$, for $\C D \in \PP(\C B^c)$;

\item $\rho_t: (M^t, \C U^t_{\C B^c})^{\eta_t, \zeta_t, (k - 1), d} \bisim^{col}_p (M', \C U^{t'}_{\C B^c})$;

\item $T^t_{\C D}$ is $R_{\C D}(\eta_t)^{\uparrow (k - 1)}$-complete, for $\C D \in \PP(\C B^c)$.
\end{itemize}

\item $t$ is constructed in Step 2-1: Then, $(\eta_t^{\uparrow (k + d - 1)})^p = (\zeta_t^{\uparrow (k + d - 1)})^p$. Since $k \leq d$, then $k - 1 \leq d - 1$. By the inductive hypothesis, 
\[(M^t, \C U^t_{\C B^c})^{\eta_t, \zeta_t, (k - 1), (d - 1)} = (S^t, R^t, V^t, \C U^t_{\C B^c})\]
is constructed, where $\C U^{t}_{\C B^c} = \{T^{t}_{\C D} \mid \C D \in \PP(\C B^c)\}$, such that
\begin{itemize}
\item $T^t_{\C D}$ is $\C D$-equivalent in $(M^t, \C U^t_{\C B^c})^{\eta_t, \zeta_t, (k - 1), (d - 1)}$, for $\C D \in \PP(\C B^c)$;

\item $\rho_t: (M^t, \C U^t_{\C B^c})^{\eta_t, \zeta_t, (k - 1), (d - 1)} \bisim^{col}_p (M', \C U^{t'}_{\C B^c})$;

\item $T^t_{\C D}$ is $R_{\C D}(\eta_t)^{\uparrow (k - 1)}$-complete, for $\C D \in \PP(\C B^c)$.
\end{itemize}
\end{itemize}
Add the submodel, $(M^t, \C U^t_{\C B^c})^{\eta_t, \zeta_t, (k - 1), d}$ or $(M^t, \C U^t_{\C B^c})^{\eta_t, \zeta_t, (k - 1), (d - 1)}$, to $(S, R, V)$, and extend $\rho$ to be $\rho \cup \rho_t$. 

(We write ``$(M^t, \C U^t_{\C B^c})^{\eta_t, \zeta_t, (k - 1), d}$'' or ``$(M^t, \C U^t_{\C B^c})^{\eta_t, \zeta_t, (k - 1), (d - 1)}$'' as ``$(M^t, \C U^t)$'' if it is clear from context.)

\item[Step 6] For every $\C B \in \C P^+(\C X)$, every $t \in T^*_{\C B}$, every $i \in \C B^c$, let $R_i(t)_4$ be the members of $R_i(t)$ within $(M, \C U_{\C X})^{\gamma, \xi, k, d}_4$. For every $u \in R_i(t)_4 \cup \{t\}$ such that $(M^u, \C U^u)$ exists, $T^u_i$ exists, which is an $i$-equivalent set. Then, merge all these sets $T^u_i$ and $R_i(t)_4 \cup \{t\}$ with the relation $R_i$.
\end{itemize}
Then, together with the other steps, $(M, \C U_{\C X})^{\gamma, \xi, k, d}$ is constructed for the inductive step.

Let $(M, \C U_{\C X}) = (M, \C U_{\C X})^{\gamma, \xi, k, d}$ and the construction is finished.

When it is clear from context, we write ``$(M, \C U_{\C X})_4$'' instead of ``$(M, \C U_{\C X})^{\gamma, \xi, k, d}_4$''. 

For intuition, we describe the internal structure of $(M, \C U_{\C X})$ in genealogical terms: Let's also call $T_{\C B}$ the $\C B$-sons of $s$ and $s$ is their \emph{father}, although we have not yet constructed $s$. Let's call $T^*_{\C B}$ the \emph{exact $\C B$-sons} of $s$. Similarly, 
\begin{itemize}
\item if $t$ is constructed in the \F{base case}, then let's call the copies of $R'_{\C B}(t')$ the \emph{$\C B$-sons} of $t$, while $t$ is their \emph{father}.

\item if $t$ is constructed in the \F{inductive step}, let's say the members of $T^t_{\C B}$ are the \emph{$\C B$-sons} of $t$, while $t$ is their \emph{father}; denote by $T^{t*}_{\C B}$ the \emph{exact $\C B$-sons} of $t$.

\end{itemize}
If two worlds have the same father, let's call them \emph{siblings}. Observe that the model $(M, \C U_{\C X})_4$ contains all the sons of $s$, and its members are all the siblings of one another. We distinguish the $\C B$-successors of $t$, ``$R_{\C B}(t)$'', from the $\C B$-sons of $t$, ``$T^t_{\C B}$''. It is easy to see $T^t_{\C B} \subseteq R_{\C B}(t)$.

The proof is long, so we deconstruct it into the following lemma:

Lemma \ref{lem: trans-Euc} shows that $T_{\C B}$ is $\C B$-equivalent and $(M, \C U_{\C X})$ is transitive and Euclidean, namely, a $\san{K45}_n \F D$-model.

Lemma \ref{lem: Common-p-bisim} shows that $(M, \C U_{\C X}) \bisim^{col}_p (M', \C U'_{\C X})$. Now suppose both $\gamma$ and $\xi$ are $\san{KD45}_n \F D$-satisfiable and $(M', s')$ is serial. For every $t \in S$ and every $i \in \C A$, because $R'_i(t') \neq \varnothing$, by the construction and the collectve $p$-bisimulation, we have $R_i(t) \neq \varnothing$. Then, $(M, \C U_{\C X})$ is serial, namely, a $\san{KD45}_n \F D$-model.

Lemma \ref{lem: Tcomplete} shows that $T_{\C B}$ is $R_{\C B}(\gamma)^{\uparrow k}$-complete.

Then, the Construction Lemma is proved.
\end{proof}

\begin{figure}[htbp]
    \centering
\resizebox{0.67\linewidth}{!}{
\begin{tikzpicture}[
    block/.style={draw, thick, rounded corners, inner sep=10pt},
    dblock/.style={draw, dashed, thick, rounded corners, inner sep=10pt},
    world/.style={circle, draw, thick, minimum size=8mm},
    formula/.style={thick, minimum size=8mm},
    actual/.style={world, fill=gray!30},
    edgei/.style={->, shorten >=2pt, shorten <=2pt},
    edgej/.style={red, dashed, thick, shorten >=2pt, shorten <=2pt, bend right=40}, 
    edgek/.style={->, blue, dashed, thick, shorten >=2pt, shorten <=2pt, bend left=60}, 
    edgejup/.style={red, dashed, thick, shorten >=2pt, shorten <=2pt, bend left=40}, 
    edgekst/.style={->, blue, dashed, thick, shorten >=2pt, shorten <=2pt}, 
    label/.style={font=\footnotesize}
]

\node[formula] (M') at (0,1) {\Large \(M'\)};
\node[actual] (s') at (0,0) {\(s'\)};
\draw[block] (-1.5, -1) rectangle (1.8,-3);
\node[formula] (R'Bs') at (1,-2) {\Large\(R'_{\C B}(s')\)};

\node[world] (t') at (-1,-2) {\(t'\)};

\draw[edgei] (s') -- (-1.5, -1) {};
\draw[edgei] (s') -- (1.8, -1) {};

\node[formula] (xi) at (-4,0) {\(\xi\)};

\node[formula] (zeta) at (-4,-2) {\(\zeta_t\)};

\draw[edgei] (xi) -- (zeta) node[midway, right] {\(R_{\C B}\)};

\node[formula] (gamma) at (-8,0) {\(\gamma\)};

\node[formula] (eta) at (-8,-2) {\(\eta_t\)};

\draw[edgei] (gamma) -- (eta) node[midway, left] {\(R_{\C B}\)};


\node[world] (t) at (-10,-2) {\(t\)};
\draw[dblock] (-13, -1) rectangle (-9.5,-3);
\node[formula] (T*B) at (-12,-2) {\Large\(T^*_{\C B}\)};

\draw[edgejup] (xi) to (s');
\node[formula] (deltas') at (-2.25, 1.2) {\color{red}\(M', s' \vDash \xi^p\)};
\draw[edgejup] (gamma) to (xi);
\node[formula] at (-6, 1.2) {\color{red}\small\((\gamma^{\uparrow(k + d + 1)})^p = (\xi^{\uparrow(k + d + 1)})^p\)};
\draw[edgejup] (zeta) to (t');
\node[formula] at (-2.5, -2) {\color{red}\small\(M', t' \vDash \zeta_t^p\)};

\draw[edgejup] (eta) to (zeta);
\node[formula] at (-6, -0.8) {\color{red}\small\((\eta_t^{\uparrow(k + d)})^p = (\zeta_t^{\uparrow(k + d)})^p\)};


\node[formula] (M) at (-12, 0) {\Large\(M\)};

\draw[edgek] (eta) to (t);
\draw[edgek] (zeta) to (t);
\draw[edgek] (t') to (t);

\end{tikzpicture}
}

  \caption{Step 1 of Lemma \ref{lem: ConstructionK45}}
\end{figure}
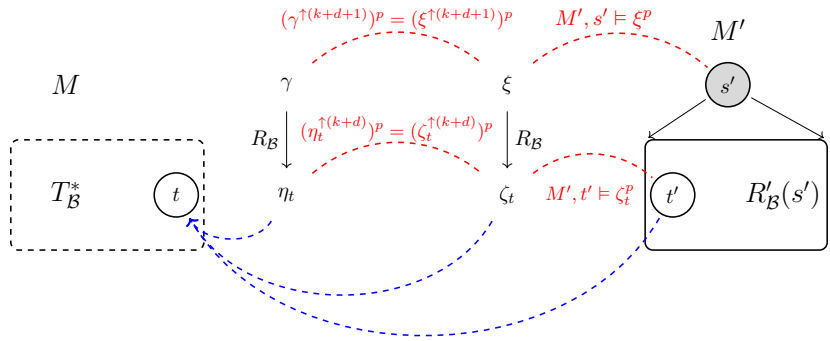

\begin{figure}[htbp]
\begin{minipage}{0.61\textwidth}
    \centering
\resizebox{0.9\linewidth}{!}{
\begin{tikzpicture}[
    block/.style={draw, thick, rounded corners, inner sep=10pt},
    dblock/.style={draw, dashed, thick, rounded corners, inner sep=10pt},
    world/.style={circle, draw, thick, minimum size=8mm},
    formula/.style={thick, minimum size=8mm},
    actual/.style={world, fill=gray!30},
    edgei/.style={->, shorten >=2pt, shorten <=2pt},
    edgej/.style={red, dashed, thick, shorten >=2pt, shorten <=2pt, bend right=40}, 
    edgek/.style={->, blue, dashed, thick, shorten >=2pt, shorten <=2pt, bend left=60}, 
    edgejup/.style={red, dashed, thick, shorten >=2pt, shorten <=2pt, bend left=40}, 
    edgekst/.style={->, blue, dashed, thick, shorten >=2pt, shorten <=2pt}, 
    label/.style={font=\footnotesize}
]

\node[formula] (M') at (0,1) {\Large \(M'\)};
\node[actual] (s') at (0,0) {\(s'\)};

\node[world] (t'_1) at (-1,-2) {\(t'_1\)};

\draw[edgei] (s') -- (t'_1) node[midway, left] {\(R'_{\C B_1}\)};

\node[world] (t'_2) at (1,-2) {\(t'_2\)};

\draw[edgei] (s') -- (t'_2) node[midway, right] {\(R'_{\C B_2}\)};

\draw[edgei] (t'_1) -- (t'_2) node[midway, below] {\(R'_{i}\)};
\node[formula] (notewhere) at (-0,-3) {\color{red}\(\text{where } i \in \C B_1^c \cap \C B_2^c\)};


\node[world] (t_1) at (-8.5,-2) {\(t_1\)};
\draw[dblock] (-10.5, -1) rectangle (-7.7,-3);
\node[formula] (T*B1) at (-9.5,-2) {\Large\(T^*_{\C B_1}\)};

\node[world] (t_2) at (-6,-2) {\(t_2\)};
\draw[dblock] (-6.8, -1) rectangle (-4,-3);
\node[formula] (T*B2) at (-5,-2) {\Large\(T^*_{\C B_2}\)};
\draw[edgei] (t_1) -- (t_2) node[midway, below] {\(R_{i}\)};

\node[formula] (M) at (-7.2, 0) {\Large\(M\)};

\draw[edgek] (-0.1, -2) to (-7.1, -2);

\end{tikzpicture}
}
  \caption{Step 2 of the base case of Lemma \ref{lem: ConstructionK45}}
\end{minipage}
\hfill
\begin{minipage}{0.38\textwidth}
    \centering
\resizebox{0.9\linewidth}{!}{
\begin{tikzpicture}[
    block/.style={draw, thick, rounded corners, inner sep=10pt},
    dblock/.style={draw, dashed, thick, rounded corners, inner sep=10pt},
    world/.style={circle, draw, thick, minimum size=8mm},
    formula/.style={thick, minimum size=8mm},
    actual/.style={world, fill=gray!30},
    edgei/.style={->, shorten >=2pt, shorten <=2pt},
    edgej/.style={red, dashed, thick, shorten >=2pt, shorten <=2pt, bend right=40}, 
    edgek/.style={->, blue, dashed, thick, shorten >=2pt, shorten <=2pt, bend left=60}, 
    edgejup/.style={red, dashed, thick, shorten >=2pt, shorten <=2pt, bend left=40}, 
    edgekst/.style={->, blue, dashed, thick, shorten >=2pt, shorten <=2pt}, 
    label/.style={font=\footnotesize}
]


\node[world] (t_1) at (-9,-2) {\(t_1\)};
\draw[dblock] (-11, -1) rectangle (-8.2,-3);
\node[formula] (T*B1) at (-10,-2) {\Large\(T^*_{\C B_1}\)};

\node[world] (t_2) at (-5.5,-2) {\(t_2\)};
\draw[dblock] (-6.3, -1) rectangle (-3.5,-3);
\node[formula] (T*B2) at (-4.5,-2) {\Large\(T^*_{\C B_2}\)};
\draw[edgei] (t_1) -- (t_2) node[midway, above] {\(R_{\C B_1 \cap \C B_2}\)};
\draw[edgei] (t_2) -- (t_1); 

\node[formula] (M) at (-7.2, 0) {\Large\(M\)};

\node[formula] (height) at (-7.2, 1) {};

\node[formula] (low) at (-7.2, -4) {};
\end{tikzpicture}
}

\caption{Step 3 of Lemma \ref{lem: ConstructionK45}}
\end{minipage}

\end{figure}

\begin{figure}[htbp]
    \centering
\resizebox{0.5\linewidth}{!}{
\begin{tikzpicture}[
    block/.style={draw, thick, rounded corners, inner sep=10pt},
    dblock/.style={draw, dashed, thick, rounded corners, inner sep=10pt},
    world/.style={circle, draw, thick, minimum size=8mm},
    formula/.style={thick, minimum size=8mm},
    actual/.style={world, fill=gray!30},
    edgei/.style={->, shorten >=2pt, shorten <=2pt},
    edgej/.style={red, dashed, thick, shorten >=2pt, shorten <=2pt, bend right=40}, 
    edgek/.style={->, blue, dashed, thick, shorten >=2pt, shorten <=2pt, bend left=60}, 
    edgejup/.style={red, dashed, thick, shorten >=2pt, shorten <=2pt, bend left=40}, 
    edgekst/.style={->, blue, dashed, thick, shorten >=2pt, shorten <=2pt}, 
    label/.style={font=\footnotesize}
]

\node[formula] (M') at (0,1) {\Large \(M'\)};
\node[actual] (s') at (0,0) {\(s'\)};

\node[world] (t') at (0,-2) {\(t'\)};

\draw[edgei] (s') -- (t') node[midway, left] {\(R'_{\C B}\)};

\draw[edgei] (t') -- (-1, -4) {};
\draw[edgei] (t') -- (1, -4) {};

\draw[block] (-1, -4) rectangle (1,-5);
\draw[block] (-1, -4) rectangle (1,-7);
\node[formula] (Ut'Bc) at (0, -4.5) {\(\C U^{t'}_{\C B^c}\)};
\node[formula] (Mt') at (0, -6) {\((M', \C U^{t'}_{\C B^c})\)};

\node[formula] (M) at (-7, 0) {\Large\(M\)};
\node[formula] (Ux) at (-10, -1.5) {\Large\(\C U_{\C X}\)};

\node[world] (t) at (-5,-2) {\(t\)};
\draw[dblock] (-7, -1.5) rectangle (-4.5,-2.5);
\node[formula] (T*B) at (-6, -2) {\Large\(T^*_{\C B}\)};

\draw[block] (-11, -1) rectangle (-3.5,-3);

\draw[block] (-6, -4) rectangle (-4,-5);
\draw[block] (-6, -4) rectangle (-4,-7);
\node[formula] (UtBc) at (-5, -4.5) {\(\C U^{t}_{\C B^c}\)};
\node[formula] (Mt) at (-5, -6) {\((M^t, \C U^{t}_{\C B^c})\)};

\draw[edgekst] (Mt') to (Mt);
\node[formula] (copy) at (-2.5, -5.8) {\color{blue}copy};
\end{tikzpicture}
}

  \caption{Step 5 of the base case of Lemma \ref{lem: ConstructionK45}}
\end{figure}

\begin{figure}[htbp]
    \centering
\resizebox{0.67\linewidth}{!}{
\begin{tikzpicture}[
    block/.style={draw, thick, rounded corners, inner sep=10pt},
    dblock/.style={draw, dashed, thick, rounded corners, inner sep=10pt},
    world/.style={circle, draw, thick, minimum size=8mm},
    formula/.style={thick, minimum size=8mm},
    actual/.style={world, fill=gray!30},
    edgei/.style={->, shorten >=2pt, shorten <=2pt},
    edgej/.style={red, dashed, thick, shorten >=2pt, shorten <=2pt, bend right=20}, 
    edgek/.style={->, blue, dashed, thick, shorten >=2pt, shorten <=2pt, bend left=60}, 
    edgejup/.style={blue, dashed, thick, shorten >=2pt, shorten <=2pt, bend left=20}, 
    edgekst/.style={->, blue, dashed, thick, shorten >=2pt, shorten <=2pt}, 
    label/.style={font=\footnotesize}
]

\node[formula] (M) at (0, 0) {\Large\(M\)};

\node[formula] (Ux) at (-4, -1) {\Large\(\C U_{\C X}\)};
\draw[block] (-5, -0.5) rectangle (5,-3);

\node[world] (t1) at (-2,-2) {\(t_1\)};

\node[world] (t2) at (0,-2) {\(t_2\)};

\node[world] (t3) at (2,-2) {\(t_3\)};

\draw[dblock] (-3, -1.5) rectangle (3,-2.5);
\node[formula] (RT) at (3, -1.2) {\small\(\RT(t_1, R_i \cup R_i^{-1})_4\)};
\draw[edgei] (t1) -- (t2) node[midway, below] {\(R_i\)};
\draw[edgei] (t3) -- (t2) node[midway, below] {\(R_i\)};

\draw[block] (-2, -4) rectangle (2,-7);
\node[formula] (Mt2) at (0, -6.5) {\((M^{t_2}, \C U^{t_2})\)};
\draw[block] (-0.8, -4.4) rectangle (0.8,-5.4);
\node[formula] (R'it'2) at (0, -4.9) {\(R'_i(t'_2)\)};

\draw[dblock] (-1.5, -4.2) rectangle (1.5,-5.6);
\node[formula] (Wt2i) at (-1.2, -4.9) {\(W^{t_2}_i\)};

\draw[block] (-7, -4) rectangle (-3,-7);
\node[formula] (Mt1) at (-5, -6.5) {\((M^{t_1}, \C U^{t_1})\)};
\draw[block] (-5.8, -4.4) rectangle (-4.2,-5.4);
\node[formula] (R'it'1) at (-5, -4.9) {\(R'_i(t'_1)\)};

\draw[dblock] (-6.5, -4.2) rectangle (-3.5,-5.6);
\node[formula] (Wt1i) at (-6.2, -4.9) {\(W^{t_1}_i\)};

\draw[block] (7, -4) rectangle (3,-7);
\node[formula] (Mt3) at (5, -6.5) {\((M^{t_3}, \C U^{t_3})\)};
\draw[block] (5.8, -4.4) rectangle (4.2,-5.4);
\node[formula] (R'it'3) at (5, -4.9) {\(R'_i(t'_3)\)};

\draw[dblock] (6.5, -4.2) rectangle (3.5,-5.6);
\node[formula] (Wt3i) at (6.2, -4.9) {\(W^{t_3}_i\)};


\draw[edgejup] (-3, -2.5) to (-6.5, -4.2);
\draw[edgejup] (-3.5,-5.6) to (-1.5, -5.6);
\draw[edgejup] (1.5,-5.6) to (3.5, -5.6);
\draw[edgejup] (6.5,-4.2) to (3,-2.5);
\node[formula] (Merge) at (2, -3.5) {\color{blue}Merge};

\end{tikzpicture}
}

  \caption{Step 6 of the base case of Lemma \ref{lem: ConstructionK45}}
\end{figure}

\begin{figure}[htbp]
    \centering
\resizebox{0.7\linewidth}{!}{
\begin{tikzpicture}[
    block/.style={draw, thick, rounded corners, inner sep=10pt},
    dblock/.style={draw, dashed, thick, rounded corners, inner sep=10pt},
    world/.style={circle, draw, thick, minimum size=8mm},
    formula/.style={thick, minimum size=8mm},
    actual/.style={world, fill=gray!30},
    edgei/.style={->, shorten >=2pt, shorten <=2pt},
    edgej/.style={red, dashed, thick, shorten >=2pt, shorten <=2pt, bend right=40}, 
    edgek/.style={->, blue, dashed, thick, shorten >=2pt, shorten <=2pt, bend left=60}, 
    edgejup/.style={red, dashed, thick, shorten >=2pt, shorten <=2pt, bend left=40}, 
    edgekst/.style={->, blue, dashed, thick, shorten >=2pt, shorten <=2pt}, 
    label/.style={font=\footnotesize}
]

\node[formula] (M') at (0,1) {\Large \(M'\)};
\node[actual] (s') at (0,0) {\(s'\)};

\node[world] (t') at (-1,-2) {\(t'\)};

\draw[edgei] (s') -- (t') node[midway, left] {\(R'_{\C B}\)};

\node[world] (u') at (1,-3) {\(u'\)};

\draw[edgei] (s') -- (u') node[midway, right] {\(R'_{\C C_1}\)};

\draw[edgei] (t') -- (u') node[midway, below] {\(R'_{\C C_2}\)};

\node[formula] (xi) at (-3,0) {\(\xi\)};

\node[formula] (zetat) at (-4,-2) {\(\zeta_t\)};

\draw[edgei] (xi) -- (zetat) node[midway, left] {\(R_{\C B}\)};

\node[formula] (zeta) at (-2,-3) {\(\zeta\)};

\draw[edgei] (xi) -- (zeta) node[midway, right] {\(R_{\C C_1}\)};

\node[formula] (zetadown) at (-3,-5) {\(\zeta^{\downarrow}\)};

\draw[edgei] (zetat) -- (zetadown) node[midway, left] {\(R_{\C C_2}\)};


\node[formula] (gamma) at (-6,0) {\(\gamma\)};

\node[formula] (etat) at (-7,-2) {\(\eta_t\)};

\draw[edgei] (gamma) -- (etat) node[midway, left] {\(R_{\C B}\)};

\node[formula] (eta) at (-5,-3) {\(\eta\)};

\draw[edgei] (gamma) -- (eta) node[midway, right] {\(R_{\C C_1}\)};

\node[formula] (etadown) at (-6,-5) {\(\eta^{\downarrow}\)};

\draw[edgei] (etat) -- (etadown) node[midway, left] {\(R_{\C C_2}\)};



\node[world] (t) at (-11,-2) {\(t\)};
\draw[dblock] (-13.5, 0) rectangle (-10.5,-2.5);
\node[formula] (T*B) at (-12.5,-2) {\Large\(T^*_{\C B}\)};

\node[world] (u) at (-9,-4) {\color{blue}\(u\)};
\draw[dblock] (-11.5, -3.5) rectangle (-8.5,-5);
\node[formula] (T*C1) at (-11,-4) {\Large\(T^*_{\C C_1}\)};
\draw[edgei] (t) -- (u) node[midway, above] {\color{blue}\(R_{\C C_2}\)};

\node[formula] (M) at (-10, 1) {\Large\(M\)};

\draw[edgejup] (gamma) to (xi);
\draw[edgejup] (xi) to (s');
\node[formula] at (-5.2, 1.2) {\color{red}\((\gamma^{\uparrow (k + d + 1)})^p = (\xi^{\uparrow (k + d + 1)})^p\)};
\node[formula] at (-1.5, 1.2 ) {\color{red}\(M', s' \vDash \xi^p\)};

\draw[edgejup] (etat) to (zetat);
\draw[edgejup] (zetat) to (t');

\draw[edgejup] (eta) to (zeta);
\draw[edgejup] (zeta) to (u');

\draw[edgek] (eta) to (u);
\draw[edgek] (zeta) to (u);
\draw[edgek] (u') to (u);
\node[formula] at (-4, -6 ) {\color{blue}construct};

\node[formula] at (-5, -7) {\color{red}where $(\eta^{\uparrow (k + d - 1)})^p = (\zeta^{\uparrow (k + d - 1)})^p$};

\end{tikzpicture}
}
  \caption{Step 2-1 of the inductive step of Lemma \ref{lem: ConstructionK45}}
\end{figure}

\begin{figure}[htbp]
    \centering
\resizebox{0.67\linewidth}{!}{
\begin{tikzpicture}[
    block/.style={draw, thick, rounded corners, inner sep=10pt},
    dblock/.style={draw, dashed, thick, rounded corners, inner sep=10pt},
    world/.style={circle, draw, thick, minimum size=8mm},
    formula/.style={thick, minimum size=8mm},
    actual/.style={world, fill=gray!30},
    edgei/.style={->, shorten >=2pt, shorten <=2pt},
    edgej/.style={red, dashed, thick, shorten >=2pt, shorten <=2pt, bend right=40}, 
    edgek/.style={->, blue, dashed, thick, shorten >=2pt, shorten <=2pt, bend left=60}, 
    edgejup/.style={red, dashed, thick, shorten >=2pt, shorten <=2pt, bend left=40}, 
    edgekst/.style={->, blue, dashed, thick, shorten >=2pt, shorten <=2pt}, 
    label/.style={font=\footnotesize}
]

\node[formula] (M') at (1.5,0) {\Large \(M'\)};
\node[actual] (s') at (1.5, -1) {\(s'\)};

\node[world] (t') at (0,-2) {\(t'\)};

\draw[edgei] (s') -- (t');

\draw[edgei] (t') -- (-1, -4) {};
\draw[edgei] (t') -- (1, -4) {};

\draw[block] (-1, -4) rectangle (1,-5);
\draw[block] (-1, -4) rectangle (1,-7);
\node[formula] (Ut'Bc) at (0, -4.5) {\(\C U^{t'}\)};
\node[formula] (Mt') at (0, -6) {\((M', \C U^{t'})\)};

\node[world] (u') at (3,-2) {\(u'\)};

\draw[edgei] (s') -- (u');

\draw[edgei] (u') -- (2, -4) {};
\draw[edgei] (u') -- (4, -4) {};

\draw[block] (2, -4) rectangle (4,-5);
\draw[block] (2, -4) rectangle (4,-7);
\node[formula] (Uu'Bc) at (3, -4.5) {\(\C U^{u'}\)};
\node[formula] (Mu') at (3, -6) {\((M', \C U^{u'})\)};

\node[formula] (M) at (-7, 0) {\Large\(M\)};
\node[formula] (Ux) at (-10, -0.5) {\Large\(\C U_{\C X}\)};
\draw[block] (-10.5, -1) rectangle (-3.5,-3);

\node[world] (u) at (-5,-2) {\(u\)};
\node[formula] (ucon2) at (-5, -1.5) {\small\color{red}constructed in Step 2-1};

\draw[block] (-7, -4) rectangle (-3,-5);
\draw[block] (-7, -4) rectangle (-3,-7);
\node[formula] (UuBc) at (-5, -4.5) {\(\C U^{u}\)};
\node[formula] (Mu) at (-5, -6) {\small\((M^u, \C U^{u})^{\eta_u, \zeta_u, (k - 1), (d - 1)}\)};

\node[world] (t) at (-9,-2) {\(t\)};
\node[formula] (ucon2) at (-9, -1.5) {\small\color{red}constructed in Step 1};

\draw[block] (-11, -4) rectangle (-7.2,-5);
\draw[block] (-11, -4) rectangle (-7.2,-7);
\node[formula] (UtBc) at (-9, -4.5) {\(\C U^{t}\)};
\node[formula] (Mt) at (-9, -6) {\small\((M^t, \C U^{t})^{\eta_u, \zeta_u, (k - 1), d}\)};

\draw[edgek] (Mt') to (Mt);
\draw[edgek] (Mu') to (Mu);
\node[formula] (copy) at (-2, -7.5) {\color{blue}inductively construct};
\end{tikzpicture}
}

  \caption{Step 5 of the inductive step of Lemma \ref{lem: ConstructionK45}}
\end{figure}

\begin{figure}[htbp]
    \centering
\resizebox{0.67\linewidth}{!}{
\begin{tikzpicture}[
    block/.style={draw, thick, rounded corners, inner sep=10pt},
    dblock/.style={draw, dashed, thick, rounded corners, inner sep=10pt},
    world/.style={circle, draw, thick, minimum size=8mm},
    formula/.style={thick, minimum size=8mm},
    actual/.style={world, fill=gray!30},
    edgei/.style={->, shorten >=2pt, shorten <=2pt},
    edgej/.style={red, dashed, thick, shorten >=2pt, shorten <=2pt, bend right=20}, 
    edgek/.style={->, blue, dashed, thick, shorten >=2pt, shorten <=2pt, bend left=60}, 
    edgejup/.style={blue, dashed, thick, shorten >=2pt, shorten <=2pt, bend left=20}, 
    edgekst/.style={->, blue, dashed, thick, shorten >=2pt, shorten <=2pt}, 
    label/.style={font=\footnotesize}
]

\node[formula] (M) at (0, 0) {\Large\(M\)};

\node[formula] (Ux) at (-4, -1) {\Large\(\C U_{\C X}\)};
\draw[block] (-5, -0.5) rectangle (5,-3);

\node[world] (t) at (-2,-1.75) {\(t\)};

\node[world] (u1) at (0.5,-1.5) {\(u_1\)};

\node[world] (u2) at (2,-2) {\(u_2\)};

\draw[dblock] (-2.5, -1) rectangle (2.5,-2.5);
\node[formula] (Rit) at (3.5, -1.2) {\small\(R_i(t)_4 \cup \{t\}\)};
\draw[edgei] (t) -- (u1);
\draw[edgei] (t) -- (u2);
\draw[edgei] (u2) -- (u1);
\draw[edgei] (u1) -- (u2);

\draw[block] (-2, -4) rectangle (2,-7);
\node[formula] (Mu1) at (0, -6.5) {\((M^{u_1}, \C U^{u_1})\)};
\node[formula] (Tu1i) at (0, -4.9) {\(T^{u_1}_i\)};
\draw[dblock] (-1.5, -4.2) rectangle (1.5,-5.6);

\draw[block] (-7, -4) rectangle (-3,-7);
\node[formula] (Mt) at (-5, -6.5) {\((M^{t}, \C U^{t})\)};
\node[formula] (Tti) at (-5, -4.9) {\(T^t_i\)};
\draw[dblock] (-6.5, -4.2) rectangle (-3.5,-5.6);

\draw[block] (7, -4) rectangle (3,-7);
\node[formula] (Mu2) at (5, -6.5) {\((M^{u_2}, \C U^{u_2})\)};
\node[formula] (Tu2i) at (5, -4.9) {\(T^{u_2}_i\)};
\draw[dblock] (6.5, -4.2) rectangle (3.5,-5.6);


\draw[edgejup] (-2.5, -2.5) to (-6.5, -4.2);
\draw[edgejup] (-3.5,-5.6) to (-1.5, -5.6);
\draw[edgejup] (1.5,-5.6) to (3.5, -5.6);
\draw[edgejup] (6.5,-4.2) to (2.5,-2.5);
\node[formula] (Merge) at (2, -3.5) {\color{blue}Merge};

\end{tikzpicture}
}

  \caption{Step 6 of the inductive step of Lemma \ref{lem: ConstructionK45}}
\end{figure}

\begin{fact}\label{fact: flow-width}
Consider the model $(M, \C U_{\C X})$ of Lemma \ref{lem: ConstructionK45}. Let $i \in \C A$ and $\C B, \C B_1, \C B_2 \in \PPA$.
\begin{enumerate}
\item For every $t_1 \in T^*_{\C B_1}$, $t_2 \in T^*_{\C B_2}$, if $(t_1, t_2) \in R_i$, then
\begin{itemize}
\item $i \in \C B_1 \Longleftrightarrow i \in \C B_2$;

\item $i \in \C B_1^c \Longleftrightarrow i \in \C B_2^c$.

\end{itemize}
\item For every constructed world $t$ and every $u \in T^{t*}_{\C B}$, $(t, u) \in R_i \Longleftrightarrow i \in \C B$.
\end{enumerate}

\end{fact}
\begin{proof}
For item 1, let's discuss the relationship between $i$ and $\C B_1$.
\begin{itemize}
\item $i \in \C B_1$: $(t_1, t_2) \in R_i$ is added by Step 3. So $i \in \C B_1 \cap \C B_2 \subseteq \C B_2$. 

\item $i \in \C B^c$: If $(t_1, t_2) \in R_i$ is added by Step 2 of the Base case, we have $i \in \C B_1^c \cap \C B_2^c \subseteq \C B_2^c$. If it is added by Step 2-1 of the inductive step, we have $\C B_2 \subseteq \C B_1$, so $i \in \C B_1^c \subseteq \C B_2^c$. If it is added by Step 2-2, then there is $\C B_3 \in \PPX$ and $t_3 \in T^*_{\C B_3}$ such that $\C B_1, \C B_2 \subseteq \C B_3$, and $t_1, t_2 \in R_i(t_3)$. So, $i \in \C B_3^c \subseteq \C B_2^c$. 
\end{itemize}

For item 2, note that by Step 6 of the inductive construction,  $(t, u) \in R_i$ if and only if $u \in T^t_i$ if and only if $i \in \C B$.
\end{proof}

\begin{lem}\label{lem: trans-Euc-1}
In the construction of Lemma \ref{lem: ConstructionK45}, 
$(M, \C U_{\C X})_4$ is transitive and Euclidean.
\end{lem}

\begin{proof}
Let $i \in \C A$ and $u, v, w$ be any worlds in $(M, \C U_{\C X})_4$. There are $\C B_1, \C B_2, \C B_3 \in \PPX$ such that $u \in T^*_{\C B_1}$, $v \in T^*_{\C B_2}$, and $w \in T^*_{\C B_3}$.

Let's check $R_i$ is transitive in $(M, \C U_{\C X})_4$. Suppose $(u, v) \in R_i$ and $(v, w) \in R_i$. We need to show that $(u, w) \in R_i$.
\begin{itemize}
\item $i \in \C B_1$: By Fact \ref{fact: flow-width}, $i \in \C B_2$. Similarly, $i \in \C B_3$ and then $i \in \C B_1 \cap \C B_3$. By Step 3, we have $(u, w) \in R_{\C B_1 \cap \C B_3} \subseteq R_i$.

\item $i \in \C B_1^c$: Consider the inductive construction:
\begin{itemize}
\item In the \F{base case}: By Step 2, $(u', v') \in R'_i$ and $(v', w') \in R'_i$. Since $(M', s')$ is transitive, $(u', w') \in R_i$. By Step 2, we have that $(u, w) \in R_i$.

\item In the \F{inductive step}: If $u$ is constructed in Step 1, then $v, w$ are constructed in Step 2-1 and $(v, w) \in R_i$ is added in Step 2-2, so $(u, w) \in R_i$. If $u$ is constructed in Step 2-1, there is $t$ such that $u, v, w \in R_i(t)$ and then by Step 2-2, we have $(u, w) \in R_i$.

\end{itemize}
\end{itemize}
Therefore, $(M, \C U_{\C X})_4$ is transitive.

Let's check $R_i$ is Euclidean in $(M, \C U_{\C X})_4$. Suppose $(u, v) \in R_i$ and $(u, w) \in R_i$. We need to show that $(v, w) \in R_i$.
\begin{itemize}
\item $i \in \C B_1$: By Fact \ref{fact: flow-width}, $i \in \C B_2$ and $i \in \C B_3$ and then $i \in \C B_2 \cap \C B_3$. By Step 3, we have $(v, w) \in R_{\C B_2 \cap \C B_3} \subseteq R_i$.

\item $i \in \C B_1^c$, consider the inductive construction:
\begin{itemize}
\item In the \F{base case}: by Step 2, $(u', v') \in R'_i$ and $(u', w') \in R'_i$. Since $(M', s')$ is Euclidean, $(v', w') \in R_i$. By Step 2, we have that $(v, w) \in R_i$.

\item In the \F{inductive step}: If $u$ is constructed in Step 1, then $v, w$ are constructed in Step 2-1, so by Step 2-2, we have $(v, w) \in R_i$. If $u$ is constructed in Step 2-1, there is $t$ such that $u, v, w \in R_i(t)$ and then by Step 2-2, we have $(v, w) \in R_i$.
\end{itemize}
\end{itemize}
Therefore, $(M, \C U_{\C X})_4$ is Euclidean.
\end{proof}

\begin{lem}\label{lem: trans-Euc}
In the construction of Lemma \ref{lem: ConstructionK45}, the following holds:
\begin{itemize}
\item $T_{\C B}$ is $\C B$-equivalent, for every $\C B \in \PPX$.

\item $(M, \C U_{\C X})$ is transitive and Euclidean.
\end{itemize}
\end{lem}

\begin{proof}
For every $\C B \in \PPX$, by Lemma \ref{lem: trans-Euc-1}, we have that $T_{\C B}$ is quasi-$\C B$-equivalent in $(M, \C U_{\C X})_4$. Observe that, in Steps 5 and 6, no edge of $R_{\C B}$ is added between any member of $T_{\C B}$ and other worlds. By Step 3, every $t \in T_{\C B}$, $(t, t) \in R_{\C B}$. Therefore, $T_{\C B}$ is $\C B$-equivalent in the whole model $(M, \C U_{\C X})$.

We need to prove $\RT(u, R_i \cup R_i^{- 1})$ is quasi-$i$-equivalent for every $u \in S$ and $i \in \C A$. Let's discuss the cases of $u$ and $\RT(u, R_i \cup R_i^{- 1})$.
\begin{itemize}

\item $\RT(u, R_i \cup R_i^{- 1}) = \RT(u, R_i \cup R_i^{- 1})_4$:
\begin{itemize}
\item $u$ is in some $T^*_{\C B}$ and $i \in \C B$: Observe that $T^*_{\C B} \subseteq T_{\C B} \subseteq T_i = \RT(u, R_i \cup R_i^{- 1})$. We have known that $T_i$ is $i$-equivalent.

\item $u$ is in some $T^*_{\C B}$ and $i \in \C B^c$: Then, no $i$-edges or concatenated submodels are added to any member of $\RT(u, R_i \cup R_i^{- 1})$ in Step 6. By Lemma \ref{lem: trans-Euc-1}, $\RT(u, R_i \cup R_i^{- 1})$ is quasi-$i$-equivalent.

\end{itemize}

\item $\RT(u, R_i \cup R_i^{- 1})$ is only in some submodel $(M^t, \C U^t)$ for $t$ in $(M, \C U_{\C X})_4$: Observe that $t$ is not in $(M^t, \C U^t)$. No $i$-edges are added to any member of $\RT(u, R_i \cup R_i^{- 1})$ in Step 6. Therefore, $\RT(u, R_i \cup R_i^{- 1})$ is the same as that in $(M^t, \C U^t)$. Whether it is either a copy or inductively constructed, $(M^t, \C U^t)$ is transitive and Euclidean, so $\RT(u, R_i \cup R_i^{- 1})$ is quasi-$i$-equivalent.

\item Otherwise: By Lemma \ref{lem: trans-Euc-1}, $\RT(u, R_i \cup R_i^{- 1})_4$ is quasi-$i$-equivalent. For every $t \in \RT(u, R_i \cup R_i^{- 1})_4$, if $(M^t, \C U^t)$ exists, consider the inductive construction:
\begin{itemize}
\item In the \F{base case}: Since $(M^t, \C U^t)$ is a copy of $(M', \C U^{t'})$ which is $\san L$-model, then $W^t_i$ is quasi-$i$-equivalent in $(M', \C U^{t'})$.

\item In the \F{Inductive step}: By the inductive hypothesis of the construction, if $(M^t, \C U^t)$ constructed , then $T^t_i$ is $i$-equivalent in $(M^t, \C U^t)$.
\end{itemize}
Step 6 has merged all these quasi-$i$-equivalent sets. By Fact \ref{fact: merge}, $\RT(u, R_i \cup R_i^{- 1})$ is quasi-$i$-equivalent.
\end{itemize}
By Proposition \ref{prop: q-e-tran-Eu}, $(M, \C U_{\C X})$ is transitive and Euclidean.
\end{proof}

\begin{lem}\label{lem: Common-p-bisim}
In the construction of Lemma \ref{lem: ConstructionK45},
\[\rho: (M, \C U_{\C X}) \bisim^{col}_p (M', \C U'_{\C X})\]
\end{lem}
\begin{proof}
Let's check the conditions in Definition \ref{defn: col-p-bisim}.

For $t \in T_{\C B} \in \C U_{\C X}$, there is a $\C B_0 \supseteq \C B$ such that $t \in T^*_{\C B_0}$. By Step 1 and Step 2, there is $t' \in R_{\C B_0}(s')$ that constructs $t$. As $\C B_0 \supseteq \C B$, $R_{\C B_0}(s') \subseteq R_{\C B}(s')$. So. $t' \in R_{\C B}(s')$ and $(t, t') \in \rho$.

Conversely, for every $t' \in R_{\C B}(s')$, because $M', s' \vDash \xi^p$ and $(\gamma^{\uparrow (k + d + 1)})^p = (\xi^{\uparrow (k + d + 1)})^p$, there is $\eta \in R_{\C B}(\gamma)$ and $\zeta \in R_{\C B}(\xi)$ such that $(\eta^p)^{\uparrow (k + d)} = (\zeta^p)^{\uparrow (k + d)}$ and $M', t' \vDash \zeta^p$. By Step 1, there is $t$ constructed by $t'$, $\eta$, and $\zeta$, and $(t, t') \in \rho$.

Let $u, u'$ be any worlds in $S$ and $S'$ respectively so that $(u, u') \in \rho$. 

\F{Atoms}: If $u$ is a copy of $u'$, then $V(u) = V(u')$. If it is constructed, observe that $w(\eta_u)^p = w(\zeta_u)^p$, and $M, u \vDash w(\eta_u)$ while $M', u' \vDash w(\zeta_u)^p$. Therefore, $V(u) \sim_p V(u')$.

\F{Forth}: For all $\C C \in \PPA$ and all $v \in S$, suppose $(u, v) \in R_\C{C}$. Let $v'$ be the world that constructs $v$. Then, $(v, v') \in \rho$. We need to show $(u', v') \in R'_{\C C}$. Let's discuss the cases of $u, v$.
\begin{itemize}
\item $u, v$ are in $(M, \C U_{\C X})_4$: Suppose $u \in T^*_{\C B_1}$ and $v \in T^*_{\C B_2}$ for some $\C B_1, \C B_2 \in \PPX$. Let's discuss the inclusion relationship between $\C C$ and $\C B_1$.
\begin{itemize}
\item $\C C \subseteq \C B_1$: By Fact \ref{fact: flow-width}, $\C C \subseteq \C B_2$ and then $\C C \subseteq \C B_1 \cap \C B_2$. By the construction, $u' \in R'_{\C B_1}(s')$ and $v' \in R'_{\C B_2}(s')$. Because $M'$ is Euclidean, we have $(u', v') \in R'_{\C B_1 \cap \C B_2} \subseteq R'_{\C C}$.

\item $\C C \subseteq \C B_1^c$: If $(u, v)$ is added in Step 2 of the base case, or it is added in Step 2-1 of the inductive step, it is easy to see that $(u', v') \in R'_{\C C}$. Suppose $(u, v)$ is added to $R_{\C C}$ in Step 2-2. Then, there is $t$ in $(M, \C U_{\C X})_4$ such that $u$ and $v$ are in $R_{\C C}(t)$. Then, we have $u'$ and $v'$ are in $R'_{\C C}(t')$ in $M'$. Since $R'_{\C C}$ is Euclidean, $(u', v') \in R'_{\C C}$.

\item Otherwise: Let $\C C_1 = \C C \cap \C B_1$ and $\C C_2 = \C C \cap \C B_1^c$. With the help of the cases ``$\C C \subseteq \C B_1$'' and ``$\C C \subseteq \C B_1^c$'' above, we have $(u', v') \in R'_{\C C_1} \cup R'_{\C C_2} = R'_{\C C}$.

\end{itemize}

\item $u, v$ are in the same submodel $(M^t, \C U^t)$ for to some $t$ in $(M, \C U_{\C X})_4$: Consider the construction.
\begin{itemize}
\item In the \F{base case}: $(M^t, \C U^t)$ is a copy of $(M', \C U^{t'})$. Then, $u, v$ are copies of $u', v'$ respectively. So, we have $(u', v') \in R'_{\C C}$

\item In the \F{inductive step}: The submodel $(M^t, \C U^t)$ satisfies that
\[\rho_t: (M^t, \C U^t) \bisim^{col}_p (M', \C U^{t'})\]
By the inductive hypothesis, $(u', v') \in R'_{\C C}$ and $(v, v') \in \rho_t \subseteq \rho$.
\end{itemize}

\item $u$ is in $(M, \C U_{\C X})_4$ but $v$ is not:
\begin{itemize}
\item $v$ is a $\C C$-son of $u$: Consider the construction.
\begin{itemize}
\item In the \F{base case}: $v$ is a copy of $v'$. Note that $(u, v) \in R_{\C C}$ can only be added in Step 6, and Step 6 requires that every $(v', v') \in R'_i$ for every $i \in \C C$. Since $W^u_i$ is quasi-$i$-equivalent in $(M^u, \C U^u)$, then there is $v'_0 \in R'_{\C C}(u')$ such that $(v', v'_0) \in R'_i$. Since $R'_{\C C}$ is Euclidean, $(v'_0, v') \in R'_i$. Since $R'_{\C C}$ is transitive, we have $(u', v') \in R'_i$ in $(M', s')$. Then, $v' \in  R'_{\C C}(u')$.

\item In the \F{inductive step}: By the inductive hyposis, $(v, v') \in \rho_u \subseteq \rho$ and $v' \in R_{\C C}(u')$.

\end{itemize}

\item Otherwise: $v$ is a $\C C$-son of some sibling $t_1$ of $u$. By the case ``$v$ is a $\C C$-son of $u$'', we have known that $(t'_1, v') \in R'_{\C C}$. There are three possible cases:
\begin{itemize}
\item $(u, t_1) \in R_{\C C}$: Since $t_1, u$ are in $(M, \C U_{\C X})_4$, by the case ``$u, v$ are in $(M, \C U_{\C X})_4$'', we have $(u', t'_1) \in R'_{\C C}$. Since $R'_{\C C}$ is transitive, we have that $(u', v') \in R'_{\C C}$.

\item $(t_1, u) \in R_{\C C}$: Since $t_1, u$ are in $(M, \C U_{\C X})_4$, by the case ``$u, v$ are in $(M, \C U_{\C X})_4$'', we have $(t'_1, u') \in R'_{\C C}$. Since $R'_{\C C}$ is Euclidean, we have that $(u', v') \in R'_{\C C}$.

\item There is $t_2$ in $(M, \C U_{\C X})_4$ such that $(t_1, t_2), (u, t_2) \in R_{\C C}$: Since $t_1$, $t_2$, and $u$ are in $(M, \C U_{\C X})_4$, by the case ``$u, v$ are in $(M, \C U_{\C X})_4$'', we have that $(t'_1, t'_2), (u', t'_2) \in R'_{\C C}$. Since $R'_{\C C}$ is Euclidean, we have that $(t'_2, v') \in R'_{\C C}$. Since $R'_{\C C}$ is transitive, we have that $(u', v') \in R'_{\C C}$.

\end{itemize}
\end{itemize}

\item $v$ is in $(M, \C U_{\C X})_4$ but $u$ is not: $(u, v) \in R_{\C C}$ is added in Step 6 and $(v, v) \in R_{\C C}$ holds in $(M, \C U_{\C X})_4$. Then, $u$ is a $\C C$-son of $t_3$, where either $t_3 = v$ or $t_3$ is a sibling of $v$. By the discussion above, we have that $(t'_3, u') \in R'_{\C C}$ and $(t'_3, v') \in R'_{\C C}$. Since $R'_{\C C}$ is Euclidean, we have that $(u', v') \in R'_{\C C}$.

\item $u$ and $v$ are in different submodels of Step 5: $(u, v) \in R_{\C C}$ is added in Step 6. There are $t_4, t_5$ in $(M, \C U_{\C X})_4$ such that $u$ is a $\C C$-son of $t_4$, $v$ is a $\C C$-son of $t_5$, $(t'_4, u') \in R'_{\C C}$, and $(t'_5, v') \in R'_{\C C}$. There are three possible cases:
\begin{itemize}
\item $(t_4, t_5) \in R_{\C C}$: By the case ``$u, v$ are in $(M, \C U_{\C X})_4$'', we have $(t'_4, t'_5) \in R'_{\C C}$. Since $R'_{\C C}$ is transitive, we have that $(t'_4, v') \in R'_{\C C}$. Since $R'_{\C C}$ is Euclidean, we have that $(u', v') \in R'_{\C C}$.

\item $(t_5, t_4) \in R_{\C C}$: By the case ``$u, v$ are in $(M, \C U_{\C X})_4$'', we have $(t'_5, t'_4) \in R'_{\C C}$. Since $R'_{\C C}$ is transitive, we have that $(t'_5, u') \in R'_{\C C}$. Since $R'_{\C C}$ is Euclidean, we have that $(u', v') \in R'_{\C C}$.

\item There is $t_6$ in $(M, \C U_{\C X})_4$ such that $(t_4, t_6), (t_5, t_6) \in R_{\C C}$: Since $t_4$, $t_5$, and $t_6$ are in $(M, \C U_{\C X})_4$, we have that $(t'_4, t'_6), (t'_5, t'_6) \in R'_{\C C}$. Since $R'_{\C C}$ is Euclidean, we have that $(t'_6, u'), (t'_6, v') \in R'_{\C C}$ and then $(u', v') \in R'_{\C C}$.
\end{itemize}

\end{itemize}
Therefore, $(u', v') \in R'_{\C C}$ and the Forth condition is proved.

\F{Back}: For all $\C C \in \PPA$ and all $v' \in S'$, suppose $(u', v') \in R'_\C{C}$. We need to prove that there is $v$ such that $(u, v) \in R_\C{C}$ and $(v, v') \in \rho$. Let's discuss the cases of $u$.
\begin{itemize}
\item $u$ is in $(M, \C U_{\C X})_4$: Let's suppose $u \in T^{*}_{\C B}$. Let's discuss the inclusion relationship between $\C C$ and $\C B$.
\begin{itemize}
\item $\C C \subseteq \C B$: Since $R'_{\C C}$ is transitive, $(s', v') \in R'_{\C C}$. Since $(M', s') \vDash \xi^p$, there is $\zeta \in R_{\C C}(\xi)$ such that $(M', t') \vDash \zeta^p$. Since $(\gamma^{\uparrow (k + d + 1)})^p = (\xi^{\uparrow (k + d + 1)})^p$, there is $\eta \in R_{\C C}(\gamma)$ such that $(\eta^{\uparrow (k + d)})^p = (\zeta^{\uparrow (k + d)})^p$. By Step 1, there is $v \in T^*_{\C C}$ that is constructed by $v'$ and $(v, v') \in \rho$. By Step 3, we have $(u, v) \in R_{\C C}$.

\item $\C C \subseteq \C B^c$: Since $(u', v') \in R_{\C C}$, by Steps 5 and 6, there is $v$ as a $\C C$-son of $u$ such that $(u, v) \in R_{\C C}$ and $(v, v') \in \rho_u \subseteq \rho$.

\item Otherwise: Let $\C C_1 = \C C \cap \C B$ and let $\C C_2 = \C C \cap \C B^c$. Because $(s', u') \in R'_{\C B}$ and $(u', v') \in R'_{\C C}$, we have that $(s', v') \in R'_{\C C_1}$. Consider the construction.
\begin{itemize}
\item In the \F{base case}: By Step 1, there is $v \in T^*_{\C C}$ that is constructed by $v'$ and $(v, v') \in \rho$,  and by Step 2, $(u, v) \in R_{\C C_2}$.

\item In the \F{inductive Step}: Let's consider the case of $u$
\begin{itemize}
\item $u$ is constructed in Step 1: $v$ constructed in Step 2-1. Then, we have $(u, v) \in R_{\C C_2}$.

\item $u$ is constructed in Step 2-1: There is $t$ that is constructed in Step 1 such that $u \in R_{\C C_2}(t)$ and $(t', u') \in R_{\C C}$. Since $(u', v') \in R'_{\C C}$ and $R'_{\C C}$ is transitive, then $(t', v') \in R_{\C C}$. Then, there is $v \in T^*_{\C C_1} \cap R_{\C C_2}(t)$ that is constructed in Step 2-1. By Step 2-2, we have $(u, v) \in R_{\C C_2}$.

\end{itemize}
\end{itemize}
Then by Step 3, $(u, v) \in R_{\C C_1}$. Therefore, $(u, v) \in R_{\C C}$.

\end{itemize}

\item $u$ is not in $(M, \C U_{\C X})_4$: There is $t$ in $(M, \C U_{\C X})_4$ such that $u$ is in the submodel $(M^t, \C U^t)$. Consider the construction.
\begin{itemize}
\item In the \F{base case}: $u$ is just a copy of $u'$. Since $(u', v') \in R'_{\C C}$, there is a copy $v$ of $v'$ in $(M^t, \C U^t)$ such that $(u, v) \in R_{\C C}$.

\item In the \F{inductive step}:  By the inductive hypothesis, 
\[\rho_t: (M^t, \C U^t) \bisim^{col}_p (M', \C U^{t'})\]
and there is $v$ in $(M^t, \C U^t)$ such that $(u, v) \in R_{\C C}$ and $(v, v') \in \rho_t \subseteq \rho$.
\end{itemize}
\end{itemize}
Therefore, there is $v \in S$ such that $(u, v) \in R_{\C C}$ and $(v, v') \in \rho$. The Back condition is proved.

Hence, $\rho: (M, \C U_{\C X}) \bisim^{col}_p (M', \C U'_{\C X})$.
\end{proof}

\begin{lem}\label{lem: Tcomplete}
If $(M, \C U_{\C X})$ is the model constructed in Lemma \ref{lem: ConstructionK45}, then for every $\C B \in \PPX$, $T_{\C B}$ is $R_{\C B}(\gamma)^{\uparrow k}$-complete.
\end{lem}

\begin{proof}
Let's discuss the cases of $k$. 

Suppose $k = 0$. By the base case of the construction, for every $\C B \in \PPA$:
\begin{itemize}
\item For every $t \in T_{\C B}$, $t$ is constructed by $\eta_t \in R_{\C B}(\gamma)$ and $M, t \vDash w(\eta_t)$. Since $\eta_t^{\uparrow 0} = w(\eta_t)$, we have $M, t \vDash \eta_t^{\uparrow 0}$.

\item For every $\eta \in R_{\C B}(\gamma)$, since $(\xi^{\uparrow (0 + d + 1)})^p = (\gamma^{\uparrow (0 + d + 1)})^p$, there is $\zeta \in R_{\C B}(\xi)$ such that $(\zeta^{\uparrow (0 + d)})^p = (\eta^{\uparrow (0 + d)})^p$. By Step 1, there is $t \in T^*_{\C B} \subseteq T_{\C B}$ that is constructed by $\eta$ and $M, t \vDash w(\eta)$. Since $\eta^{\uparrow 0} = w(\eta)$, we have $M, t \vDash \eta^{\uparrow 0}$.
\end{itemize}
Therefore, when $k = 0$, for every $\C B \in \PPX$, $T_{\C B}$ is $R_{\C B}(\gamma)^{\uparrow k}$-complete.

Suppose $k > 0$. We will prove the following fact:
\begin{quotation}
For every $l \geq 0$, every $u \in S$ that is constructed but not a copy, 
\[M, u \vDash (\eta_u^{\downarrow h})^{\uparrow l}\]
\end{quotation}
Because $\gamma \in D^P_{k + h + 1}$ and $\eta_u$ wholly occurs in $\gamma$, then there is $m \in \{0, \dots, k\}$ such that $\eta_u \in D^P_{m + h}$. Then, $\eta_u^{\downarrow h} \in D^P_{m}$. Let's prove this fact by induction on $l$.

\F{Base case} ($l = 0$): Similar to the case of ``$k = 0$'' and regardless of the value of $m$, $M, u \vDash w(\eta_u)$. So, $M, u \vDash (\eta_u^{\downarrow h})^{\uparrow 0}$.

\F{Inductive step} ($l > 0$): If $m < l$, since $\eta_u^{\downarrow h} \in D^P_m$, we have $(\eta_u^{\downarrow h})^{\uparrow l} = (\eta_u^{\downarrow h})^{\uparrow (l - 1)} = \eta_u^{\downarrow h}$ by Proposition \ref{prop: uparrow-property}. Then, we have $M, u \vDash (\eta_u^{\downarrow h})^{\uparrow l}$ by the inductive hypothesis on $l$. 

In the following, we assume $m \geq l$. We have known $M, u \vDash w(\eta_u)$. It remains to show for any $\C C \in \PPA$, $M, u \vDash \nabla_{\C C} R_{\C C}(\eta^{\downarrow h}_u)^{\uparrow (l - 1)}$. According to the inductive construction, let's suppose $u$ is in an exact $\C B$-son of some $t$, namely, $u \in T^{t*}_{\C B}$. (In particular, when $m = k$, then $T^{t*}_{\C B}$ is $T^*_{\C B}$, $u$ is in $(M, \C U_{\C X})_4$, and actually, $t$ is $s$.) 

Let's first prove ``$M, u \vDash \F{D}_{\C C} \bigvee R_{\C C}(\eta^{\downarrow h}_u)^{\uparrow (l - 1)}$''. That is to prove for every $v \in R_{\C C}(u)$, there is a formula $\eta_v$ that constructs $v$, such that $(\eta_v^{\downarrow h})^{\uparrow (l - 1)} \in R_{\C C}(\eta_u^{\downarrow h})^{\uparrow (l - 1)}$ and $M, v \vDash (\eta_v^{\downarrow h})^{\uparrow (l - 1)}$.
Let's discuss the inclusion relationship between $\C C$ and $\C B$.

\begin{figure}[htbp]
\begin{minipage}{0.48\textwidth}
    \centering
\resizebox{0.9\linewidth}{!}{
\begin{tikzpicture}[
    block/.style={draw, thick, rounded corners, inner sep=10pt},
    dblock/.style={draw, dashed, thick, rounded corners, inner sep=10pt},
    world/.style={circle, draw, thick, minimum size=8mm},
    formula/.style={thick, minimum size=8mm},
    actual/.style={world, fill=gray!30},
    edge/.style={thick, shorten >=2pt, shorten <=2pt},
    edgei/.style={->, shorten >=2pt, shorten <=2pt},
    edgej/.style={red, dashed, thick, shorten >=2pt, shorten <=2pt, bend right=20}, 
    edgek/.style={->, blue, dashed, thick, shorten >=2pt, shorten <=2pt, bend left=30}, 
    edgejup/.style={blue, dashed, thick, shorten >=2pt, shorten <=2pt, bend left=20}, 
    edgekst/.style={->, blue, dashed, thick, shorten >=2pt, shorten <=2pt}, 
      reflexive/.style={
    ->,                     
    >=Stealth,              
    blue, dashed,
    line width=0.8pt,
    shorten >=1pt, shorten <=1pt,
    loop above,             
    distance=6mm,           
    looseness=8             
  },
  reflexive right/.style={reflexive, loop right},
  reflexive below/.style={reflexive, loop below},
  reflexive left/.style={reflexive, loop left},
    label/.style={font=\footnotesize}
]

\node[formula] (s) at (1, 2) {};
\node[world] (t) at (0, 0) {\(t\)};
\node[world] (t1) at (3, 0) {\(t_1\)};
\draw[block] (-0.5, 0.5) rectangle (5,-0.5);

\node[world] (u) at (-1, -3) {\(u\)};
\node[world] (u1) at (0.5, -3) {\(u_1\)};
\draw[block] (-2, -2.5) rectangle (2,-3.5);
\node[formula] at (1.5, -3) {\(\C U^{t}\)};

\node[world] (u2) at (4, -3) {\(u_2\)};
\draw[block] (3.5, -2.5) rectangle (6,-3.5);
\node[formula] at (5.5, -3) {\(\C U^{t_1}\)};

\node[formula] (Rc) at (2, -1.5) {\color{blue}\(R_{\C C}\)};

\draw[edgei] (t) -- (u) node [midway, right] {\small\(T^{t*}_{\C B}\)};

\draw[edge] (t) -- (-2, -2.5);
\draw[edge] (t) -- (2,-2.5);

\draw[edge] (t1) -- (3.5, -2.5);
\draw[edge] (t1) -- (6,-2.5);

\draw[reflexive left] (u) to (u);
\draw[edgek] (u) to (t);
\draw[edgekst] (u) to (u1);
\draw[edgekst] (u) to (t1);
\draw[edgek] (u) to (u2);
\end{tikzpicture}
}

  \caption*{When $\C C \subseteq \C B$}
\end{minipage}
\hfill
\begin{minipage}{0.48\textwidth}
    \centering
\resizebox{0.9\linewidth}{!}{
\begin{tikzpicture}[
    block/.style={draw, thick, rounded corners, inner sep=10pt},
    dblock/.style={draw, dashed, thick, rounded corners, inner sep=10pt},
    world/.style={circle, draw, thick, minimum size=8mm},
    formula/.style={thick, minimum size=8mm},
    actual/.style={world, fill=gray!30},
    edge/.style={thick, shorten >=2pt, shorten <=2pt},
    edgei/.style={->, shorten >=2pt, shorten <=2pt},
    edgej/.style={red, dashed, thick, shorten >=2pt, shorten <=2pt, bend right=20}, 
    edgek/.style={->, blue, dashed, thick, shorten >=2pt, shorten <=2pt, bend left=30}, 
    edgejup/.style={blue, dashed, thick, shorten >=2pt, shorten <=2pt, bend left=20}, 
    edgekst/.style={->, blue, dashed, thick, shorten >=2pt, shorten <=2pt}, 
  reflexive/.style={
    ->,                     
    >=Stealth,              
    blue, dashed,
    line width=0.8pt,
    shorten >=1pt, shorten <=1pt,
    loop above,             
    distance=6mm,           
    looseness=8             
  },
  reflexive right/.style={reflexive, loop right},
  reflexive below/.style={reflexive, loop below},
  reflexive left/.style={reflexive, loop left},
    label/.style={font=\footnotesize}
]

\node[world] (t) at (1, 2) {\(t\)};
\node[world] (u) at (0, 0) {\(u\)};
\draw[edgei] (t) -- (u) node [midway, right] {\small\(T^{t*}_{\C B}\)};

\node[world] (u0) at (3, 0) {\(u_0\)};
\draw[block] (-0.5, 0.5) rectangle (5,-0.5);
\draw[edge] (t) -- (-0.5, 0.5);
\draw[edge] (t) -- (5, 0.5);
\node[formula] at (4.5, 0) {\(\C U^{t}\)};

\node[world] (v) at (0, -3) {\(v\)};
\draw[block] (-2, -2.5) rectangle (2,-3.5);
\node[formula] at (1.5, -3) {\(\C U^{u}\)};

\node[world] (v1) at (4, -3) {\(v_1\)};
\draw[block] (3.5, -2.5) rectangle (6,-3.5);
\node[formula] at (5.5, -3) {\(\C U^{u_0}\)};

\node[formula] (Rc) at (2, -1) {\color{blue}\(R_{\C C}\)};

\draw[edge] (u) -- (-2, -2.5);
\draw[edge] (u) -- (2,-2.5);

\draw[edge] (u0) -- (3.5, -2.5);
\draw[edge] (u0) -- (6,-2.5);

\draw[reflexive left] (u) to (u);
\draw[edgekst] (u) to (v);
\draw[edgekst] (u) to (u0);
\draw[edgek] (u) to (v1);
\end{tikzpicture}
}

  \caption*{When $\C C \subseteq \C B^c$}
\end{minipage}
  \caption{Any possible $\C C$-successor of $u$, when $u$ is $t$'s exact $\C B$-son.}
\label{fig: PossibleSuccessors}
\end{figure}

Suppose $\C C \subseteq \C B^c$. As illustrated in Figure \ref{fig: PossibleSuccessors}, there are four possible cases:
\begin{enumerate}
\item[\F{1}] $v$ is a $\C C$-son of $u$: $v$ is constructed by some $\eta_v \in R_{\C C}(\eta_u) \subseteq D^P_{(m - 1) + h}$. 

\item[\F{2}] $v$ is $u$: $(u, u)$ is added to $R_{\C C}$ in Step 2-1 of constructing $(M^t, \C U^t)$.
There is a sibling $u_0$ of $u$ such that $u \in R_{\C C}(u_0)$. By Step 2-1, $\eta_u^{\downarrow} \in R_{\C C}(\eta_{u_0}) = R_{\C C}(\eta_u)$.
Then, $\eta_u^{\downarrow} \in R_{\C C}(\eta_u)$. 

\item[\F{3}] $v$ is some sibling of $u$: $\eta_v \in D^P_{m + h}$. There are two possible cases by Step 2:
\begin{itemize}
\item $u$ is constructed in Step 1 and $(u, v)$ is added to $ R_{\C C}$ in Step 2-1. Obviously, $\eta_v^{\downarrow} \in R_{\C C}(\eta_u)$.

\item there is a sibling $w$ of $u$ such that $u, v \in R_{\C C}(w)$ are added in Step 2-1, and $(u, v)$ is added to $ R_{\C C}$ in Step 2-2. Then, $\eta_v^{\downarrow} \in R_{\C C}(\eta_w) = R_{\C C}(\eta_u)$.
\end{itemize}
So, $\eta_v^{\downarrow} \in R_{\C C}(\eta_u)$.

\item[\F{4}] $v$ is a $\C C$-son of some sibling $u_1$ of $u$: Similar to Case \F{1}, we have $\eta_v \in R_{\C C}(\eta_{u_1})$. By Step 2, we have $R_{\C C}(\eta_{u_1}) = R_{\C C}(\eta_u)$. So, $\eta_v \in R_{\C C}(\eta_u)$.
\end{enumerate}
In summary: 
\begin{itemize}
\item In Case \F{1} and Case \F{4}, $\eta_v \in D^P_{(m - 1) + h}$, so $(\eta_v^{\downarrow h}) \in R_{\C C}(\eta_u^{\downarrow h})$ and then $(\eta_v^{\downarrow h})^{\uparrow(l - 1)} \in R_{\C C}(\eta_u^{\downarrow h})^{\uparrow(l - 1)}$.

\item In Case \F{2} and Case \F{3}, $\eta_v \in D^P_{m + h}$, so $\eta_v^{\downarrow (h + 1)} \in R_{\C C}(\eta_u^{\downarrow h}) \subseteq D^P_{m - 1}$. Since $m - 1 \geq l - 1$, by Proposition \ref{prop: uparrow-property},
\[(\eta^{\downarrow h}_v)^{\uparrow(l - 1)} = (\eta_v^{\downarrow (h + 1)})^{\uparrow(l - 1)} \in R_{\C C}(\eta^{\downarrow h}_u)^{\uparrow(l - 1)}.\]
\end{itemize}
By the hypothesis on $l$, $M, v \vDash (\eta^{\downarrow h}_v)^{\uparrow(l - 1)}$.

Suppose $\C C \subseteq \C B$. As illustrated in Figure \ref{fig: PossibleSuccessors}, there are four possible cases:
\begin{itemize}
\item[\F{5}] $v$ is either $u$ or a sibling of $u$: $\eta_v \in D^P_{m + h}$ and $v \in T^t_{\C C}$. By the construction, $\eta_v \in R_{\C C}(\eta_t)$. By Proposition \ref{prop: dist-identical-son} (the identical successors property), $R_{\C C}(\eta_u) = R_{\C C}(\eta_t^{\downarrow})$. So, $\eta_v^{\downarrow} \in R_{\C C}(\eta_u)$. Similar to Case \F{3}, we have $(\eta^{\downarrow h}_v)^{\uparrow(l - 1)} \in R_{\C C}(\eta^{\downarrow h}_u)^{\uparrow(l - 1)}$.

\item[\F{6}] $v$ is $t$: Then, $(u, t) \in R_{\C C}$ can only added in Step 6 and $(t, t) \in R_{\C C}$. Similar to Case \F{2}, $\eta_t^{\downarrow} \in R_{\C C}(\eta_t)$ and then, $\eta_t^{\downarrow 2} \in R_{\C C}(\eta_t^{\downarrow})$. By Proposition \ref{prop: dist-identical-son}, we have that $R_{\C C}(\eta_u) = R_{\C C}(\eta_t^\downarrow)$. So, $\eta_t^{\downarrow (h + 2)} \in R_{\C C}(\eta_u^{\downarrow h})$. Note that $\eta_t^{\downarrow h} \in D^P_{m + 1}$ while $\eta_t^{\downarrow (h + 2)} \in R_{\C C}(\eta_u^{\downarrow h}) \subseteq D^P_{m - 1}$, and $m - 1 \geq l - 1 \geq 0$. By Proposition \ref{prop: uparrow-property},
\[(\eta_t^{\downarrow h})^{\uparrow(l - 1)} = (\eta_t^{\downarrow (h + 2)})^{\uparrow(l - 1)} \in R_{\C C}(\eta_u^{\downarrow h})^{\uparrow(l - 1)}.\] 

\item[\F{7}] $v$ is some sibling of $t$: Similar to Case \F{3}, we have $\eta_v^\downarrow \in R_{\C C}(\eta_t)$. So, $\eta_t^{\downarrow 2} \in R_{\C C}(\eta_t^{\downarrow})$. By Proposition \ref{prop: dist-identical-son} (the identical successors property), $R_{\C C}(\eta_u) = R_{\C C}(\eta_t^\downarrow)$. Then, $\eta_v^{\downarrow (h + 2)}  \in R_{\C C}(\eta_u^{\downarrow h})$. Similar to Case \F{6}, $(\eta_v^{\downarrow h})^{\uparrow(l - 1)} \in R_{\C C}(\eta_u^{\downarrow h})^{\uparrow(l - 1)}$.

\item[\F{8}] $v$ is a $\C C$-son of some sibling $t_1$ of $t$: Similar to Case \F{5} and Case \F{7}, we have that $R_{\C C}(\eta_v) = R_{\C C}(\eta_u) = R_{\C C}(\eta_t^\downarrow) = R_{\C C}(\eta_{t_1}^\downarrow)$. Since $\eta_v \in R_{\C C}(\eta_{t_1})$, then $\eta_v^\downarrow = R_{\C C}(\eta_u)$. So, $(\eta^{\downarrow h}_v)^{\uparrow(l - 1)} \in R_{\C C}(\eta^{\downarrow h}_u)^{\uparrow(l - 1)}$.
\end{itemize}
In summary, $(\eta^{\downarrow h}_v)^{\uparrow(l - 1)} \in R_{\C C}(\eta^{\downarrow h}_u)^{\uparrow(l - 1)}$ and by the hypothesis on $l$, $M, v \vDash (\eta_v^{\downarrow h})^{\uparrow(l - 1)}$.

Suppose $\C C_1 = \C C \cap \C B \neq \varnothing$ and $\C C_2 = \C C \cap \C B^c \neq \varnothing$. There are two possible cases:
\begin{enumerate}
\item[\F{9}] $v$ is $u$: Combining Case \F{2} and Case \F{5}, we have $(\eta^{\downarrow h}_v)^{\uparrow(l - 1)} \in R_{\C C}(\eta^{\downarrow h}_u)^{\uparrow(l - 1)}$.

\item[\F{10}] $v$ is a sibling of $u$: Combining Case \F{3} and Case \F{5}, we have $(\eta^{\downarrow h}_v)^{\uparrow(l - 1)} \in R_{\C C}(\eta^{\downarrow h}_u)^{\uparrow(l - 1)}$.
\end{enumerate}
Then, by the hypothesis on $l$, $M, v \vDash (\eta^{\downarrow h}_{v})^{\uparrow(l - 1)}$.

Thus, it is proved that $M, u \vDash \F{D}_{\C C} \bigvee R_{\C C}(\eta^{\downarrow h}_u)^{\uparrow (l - 1)}$.

Conversely, we need to prove ``$M, u \vDash \bigwedge \hat{\F{D}}_{\C C} R_{\C C}(\eta^{\downarrow h}_u)^{\uparrow (l - 1)}$''. That is to prove for every $\eta \in R_{\C C}(\eta_u)$, there is $v \in R_{\C C}(u)$ such that $M, v \vDash (\eta^{\downarrow h})^{\uparrow(l - 1)}$. Recall that $M', u' \vDash \zeta_u^p$ and there is $d_1$ such that $(\eta_u^{\uparrow(m + d_1)})^p = (\zeta_u^{\uparrow(m + d_1)})^p$. Let's discuss the inclusion relationship between $\C C$ and $\C B$.
\begin{itemize}
\item $\C C \subseteq \C B^c$: There is $\zeta \in R_{\C C}(\zeta_u)$ such that $(\eta^{\uparrow (m + d_1 - 1)})^p = (\zeta^{\uparrow (m + d_1 - 1)})^p$. Since $M', u' \vDash \zeta_u^p$, there is $v' \in R'_{\C C}(u')$ such that $M', v' \vDash \zeta$. By the inductive construction of $(M^u, \C U^u)$, $\eta$, $\zeta$, and $v'$ construct a world $v \in T^{u*}_{\C C}$. By Step 6, $(u, v) \in R_{\C C}$. By the hypothesis on $l$, $M, v \vDash (\eta^{\downarrow h})^{\uparrow(l - 1)}$.

\item $\C C \subseteq \C B$: By Proposition \ref{prop: dist-identical-son} (the identical successors property), there is $\eta_0 \in R_{\C C}(\eta_t) \subseteq D^P_{m + h}$ such that $\eta_0^{\downarrow} = \eta$. There is $\zeta \in R_{\C C}(\zeta_t)$ such that $(\eta^{\uparrow (m + d_1)})^p = (\zeta^{\uparrow (m + d_1)})^p$. By Step 1 of constructing the model $(M^t, \C U^t)$, $\eta_0$ constructs some $v \in T^{t*}_{\C C} \subseteq T^t_{\C C}$. By Step 3, $(u, v) \in R_{\C C}$. Note that $\eta_0^{\downarrow h} \in D^P_m$, $\eta^{\downarrow h} \in D^P_{m - 1}$, and $m - 1 \geq l - 1 \geq 0$. By Proposition \ref{prop: uparrow-property}, $(\eta_0^{\downarrow h})^{\uparrow (l - 1)} = (\eta^{\downarrow h})^{\uparrow (l - 1)}$. By the hypothesis on $l$, we have $M, v \vDash(\eta_0^{\downarrow h})^{\uparrow (l - 1)}$. So, $M, v \vDash (\eta^{\downarrow h})^{\uparrow(l - 1)}$. 

\item $\C C_1 = \C C \cap \C B \neq \varnothing$ and $\C C_2 = \C C \cap \C B^c \neq \varnothing$: There is $\zeta \in R_{\C C}(\zeta_u)$ such that $(\eta^{\uparrow (m + d_1 - 1)})^p = (\zeta^{\uparrow (m + d_1 - 1)})^p$. Since $M', u' \vDash \zeta_u^p$, there is $v' \in R'_{\C C}(u')$ such that $M', v' \vDash \zeta$. By Step 2-1, there is $\eta_0 \in R_{\C C_1}(\eta_t)\subseteq D^P_{m + h}$ such that $\eta_0^{\downarrow} = \eta$ and $\eta_0$ constructs $v \in T^{t*}_{\C C_1}$. Then, we have $(u, v) \in R_{\C C}$. Note that $\eta_0^{\downarrow h} \in D^P_m$, $\eta^{\downarrow h} \in D^P_{m - 1}$, and $m - 1 \geq l - 1$. By Proposition \ref{prop: uparrow-property}, $(\eta_0^{\downarrow h})^{\uparrow (l - 1)} = (\eta^{\downarrow h})^{\uparrow (l - 1)}$. By the hypothesis on $l$, we have $M, v \vDash(\eta_0^{\downarrow h})^{\uparrow (l - 1)}$. So, $M, v \vDash (\eta^{\downarrow h})^{\uparrow(l - 1)}$.
\end{itemize}
Thus, it is proved that $M, u \vDash \bigwedge \hat{\F{D}}_{\C C} R_{\C C}(\eta^{\downarrow h}_u)^{\uparrow (l - 1)}$.

In summary, $M, u \vDash \bigwedge_{\C C \in \PPA}\nabla_{\C C} R_{\C C}(\eta^{\downarrow h}_u)^{\uparrow (l - 1)}$. Hence, the fact is proved: $M, u \vDash (\eta^{\downarrow h}_u)^{\uparrow l}$ for every $l$ and every constructed $u$.

Consider $(M, \C U_{\C X})$ and any $\C B \in \PPX$:
\begin{itemize}
\item For every $t \in T_{\C B}$, $t$ is constructed by $\eta_t \in R_{\C B}(\gamma)$. By Proposition \ref{prop: uparrow-property},
\[(\eta^{\downarrow h}_t)^{\uparrow k} = \eta_t^{\downarrow h} = \eta_t^{\uparrow k}.\] By the fact we just proved, $M, t \vDash (\eta^{\downarrow h}_t)^{\uparrow k}$. Therefore, $M, t \vDash \eta_t^{\uparrow k}$.

\item For every $\eta \in R_{\C B}(\gamma) \subseteq D^P_{k + h}$, since $(\xi^{\uparrow (k + d + 1)})^p = (\gamma^{\uparrow (k + d + 1)})^p$, there is $\zeta \in R_{\C B}(\xi)$ such that $(\zeta^{\uparrow (k + d)})^p = (\eta^{\uparrow (k + d)})^p$. By Step 1, there is $t \in T^*_{\C B} \subseteq T_{\C B}$ constructed by $\eta$. By the fact we just proved, $M, t \vDash (\eta^{\downarrow h})^{\uparrow k}$. By Proposition \ref{prop: uparrow-property},
\[(\eta^{\downarrow h})^{\uparrow k} = \eta^{\downarrow h} = \eta^{\uparrow k}.\] 
Therefore, $M, t \vDash \eta^{\uparrow k}$.

\end{itemize}
Hence, $T_{\C B}$ is $R_{\C B}(\gamma)^{\uparrow k}$-complete for every $\C B \in \PPX$.
\end{proof}

\ConstructionSfive*
\begin{proof}
This construction is almost the same as that of Lemma \ref{lem: ConstructionK45}, except Step 2-2 of the inductive step:
\begin{itemize}
\item[Step 2-2] For every $\C B \in \PPX$, every $t \in T^*_{\C B}$, add $(t, t) \in R_j$ for every $j \in \C A$, and subsequently, for every $i \in \PP(\C B^c)$, and every $u_1, u_2 \in R_i(t)$, then add $(u_1, u_2)$ to $R_i$.
\end{itemize}
Obviously, $(M, \C U_{\C X})_4$ is reflexive. By Lemma \ref{lem: trans-Euc-1}, $(M, \C U_{\C X})_4$ is transitive, Euclidean, and reflexive. Then by the inductive hypothesis, every concatenated submodel in Step 5 is also reflexive. Then, we can repeat the discussion of Lemma \ref{lem: trans-Euc-1} by replacing ``quasi-$i$-equivalent'' with ``$i$-equivalent''. By Proposition \ref{prop: q-e-tran-Eu}, $(M, \C U_{\C X})$ is reflexive. By Lemma \ref{lem: trans-Euc}, $(M, \C U_{\C X})$ is also transitive and Euclidean. Therefore, $(M, \C U_{\C X})$ is an $\san{S5}_n \F D$-model, and $T_{\C B}$ is $\C B$-equivalent in $(M, \C U_{\C X})$ for every $\C B \in \PPX$.

The proof of collective $p$-bisimulation is analogous to Lemma \ref{lem: Common-p-bisim}, and we only need to add the following discussion in the Forth condition: For all $\C C \in \PPA$ and all $v \in S$, suppose $(u, v) \in R_\C{C}$:
\begin{itemize}
\item $v$ is $u$: Since $(M', s')$ is reflexive, then $(u', u') \in R'_{\C C}$
\end{itemize}
Therefore, $(M, s) \bisim^{col}_p (M', s')$.

We still need to show $T_{\C B}$ is $R_{\C B}(\gamma)^{\uparrow k}$-complete. This proof repeats the one of Lemma \ref{lem: Tcomplete}. We only need to change the following case in that proof.
\begin{enumerate}
\item[\F{2}] $v$ is $u$: By Lemma \ref{lem: reflexive}, $\eta_u^{\downarrow} \in R_{\C C}(\eta_u)$.
\end{enumerate}
$T_{\C B}$ is $R_{\C B}(\gamma)^{\uparrow k}$-complete.

The construction lemma for $\san{S5}_n \F D$ is proved.
\end{proof}

\LTwo*
\begin{proof}
Let's slightly modify the construction in Lemma \ref{lem: ConstructionK45}. Recall that the parameters $d$ and $h$ are needed in Lemma \ref{lem: ConstructionK45}. Let's fix $d = h = 0$. Let $\xi$ be any member in $D^P_{k + 1}(\san L)$ such that $(\delta^{\uparrow (k + 0 + 1)})^p = (\xi^{\uparrow (k + 0 + 1)})^p$, namely $\delta^p = \xi^p$. The Greek letter $\eta$ ranges over the d-canonical formulas that wholly occur in $\delta$, while the Greek letter $\zeta$ over those in $\xi$. 

\F{Base case} ($k = 0$): the construction is the same as the Base case of Lemma \ref{lem: ConstructionK45}. Since $k \leq d = 0$, the $\san L_2$-model $(M, \C U_{\C X})^{\delta, \xi, 0, 0}$ is constructed.

\F{Inductive step} ($k > 0$): Steps 1, 3, 4, and 6 are the same as those of Lemma \ref{lem: ConstructionK45}. We only show Steps 2 and 5. Note that Lemma \ref{lem: ConstructionK45} required the condition ``$k \leq d$'' for Steps 2 and 5, but here $k > d = 0$.

\begin{itemize}
\item[Step 2-1] 
For every $\C B \in \PPX$, $\C C \in \PPA$ such that $\C B \cap \C C \neq \varnothing$ and $\C B^c \cap \C C \neq \varnothing$, every $t \in T^*_{\C B}$ and every $u' \in R'_{\C C}(t')$, since $(\eta_t^{\uparrow k})^p = (\zeta_t^{\uparrow k})^p$, consider the following:

Since $\C A = \{1, 2\}$, $\C B$ is a singleton and $\C C = \C A$. Denote by $i$ and $j$ the only members of $\C B$ and $\C B^c$ respectively. By Proposition \ref{prop: A2-main-lemma}, there exists $u \in T^*_i$ that is constructed by $u'$ and $\eta_u \in R_i(\delta)$ in Step 1 such that
\begin{itemize}
\item $\eta_u^{\downarrow} \in R_{\C C}(\eta_t)$,

\item $R_j(\eta_u) = R_j(\eta_t)$,

\item $M', u' \vDash \eta_u^p$,
\end{itemize}
Then add $(t, u), (u, u)$ to $R_j$. (Note that this Step 2-1 does not add any new worlds to $S$.) 

\item[Step 2-2] For every $\C B \in \PPX$, every $t \in T^*_{\C B}$, (if $\san{L}_2$ is $\san{S5}_2 \F D$, add $(t, t) \in R_j$ for every $j \in \C A$, and subsequently,) for every $i \in \PP(\C B^c)$, and every $u_1, u_2 \in R_i(t)$, then add $(u_1, u_2)$ to $R_i$.

\item[Step 5] For every $\C B \in \C P^+(\C X)$ where $\C B \neq \C A$, since $\C B$ is a singleton, denote by $i$ the only member of $\C B$ and by $j$ the only member of $\C B^c$. If $R'_j(t') \neq \varnothing$, consider the multi-pointed submodel $(M', \C U^{t'}_{\{j\}})$ where $T^{t'}_j = R'_j(t')$ and $\C U^{t'}_{\{j\}} = \{T^{t'}_j\}$. By the inductive hypothesis, a multi-pointed $\san L$-model 
\[(M^t, \C U^t_{\{j\}})^{\eta_t, \zeta_t, (k - 1), 0},\]
where $\C U^{t'}_{\{j\}} = \{T^{t'}_j\}$, has been constructed such that
\begin{itemize}
\item $T^t_j$ is $j$-equivalent;

\item $T^t_j$ is $R_j(\eta_t)$-complete;

\item $(M^t, \C U^t_{\{j\}})^{\eta_t, \zeta_t, (k - 1), 0} \bisim^{col}_p (M', \C U^{t'}_{\{j\}})$.
\end{itemize}
Add $(M^t, \C U^t_{\C B^c})^{\eta_t, \zeta_t, (k - 1), 0}$ to $(S, R, V)$.
\end{itemize}
With the other steps, the construction is finished.

Let $(M, \C U_{\C X}) = (M, \C U_{\C X})^{\delta, \delta, k, 0}$. Observe that, as long as $(M, \C U_{\C X})$ is constructed, none of Lemma \ref{lem: trans-Euc}, Lemma \ref{lem: Common-p-bisim}, and Lemma \ref{lem: Tcomplete} requires the condition ``$k \leq d \leq h$''. Therefore, we can conclude that $(M, \C U_{\C X})$ is an $\san{L}_2$-model such that
\begin{itemize}
\item $T_{\C B}$ is $\C B$-equivalent;

\item $(M, \C U_{\C X}) \bisim^{col}_p (M', \C U'_{\C X})$;

\item $T_{\C B}$ is $R_{\C B}(\delta)$-complete.
\end{itemize}
\end{proof}

\Dpcdichotomy*
\begin{proof}
Let's prove by induction on the construction of $\phi$. Similar to the proof of Proposition \ref{prop: canon-dichotomy}, we only consider the case $\phi = \F C\psi$. Note that $\delta$ is in the form 
\[\delta = w(\delta) \land \bigwedge_{\C B \in \PPA}\nabla_{\C B} R_{\C B}(\delta) \land \nabla \Phi\]
There are two possible cases:
\begin{itemize}
\item There is $\eta_0 \in \TC(\delta)$ such that $\eta_0 \vDash \neg \psi$: In this case, for any $(M, s)$, if $M, s \vDash \delta$, then by the subformula
\[M, s \vDash \bigwedge \hat{\F C}\  \TC(\delta),\]
there exists $t_0 \in \TC(s)$ such that $M, t_0 \vDash \eta_0$. By $\eta_0 \vDash \neg \psi$, $M, t_0 \vDash \neg \psi$. So, $M, s \vDash \hat{\F C} \neg \psi$. Therefore, $M, s \vDash \neg \F C \psi$

\item For any $\eta \in \TC(\delta)$, $\eta \vDash \psi$: In this case, for any $(M, s)$, if $M, s \vDash \delta$, then by the subformula
\[M, s \vDash \F C \bigvee \TC(\delta),\]
for any $t \in \TC(s)$, there exists $\gamma \in \TC(\delta)$ such that $M, t \vDash \gamma$. As ``for any $\eta \in \TC(\delta)$, $\eta \vDash \psi$'', we have that $\gamma \vDash \psi$. (In the special case where $\TC(\delta) = \varnothing$, $\nabla \TC(\delta) = \F C \bot$. As $\vDash \bot \to \psi$, we have $\vDash \F C(\bot \to \psi)$. According to Axiom $\san K$, we have that $\vDash \F C\bot \to \F C\psi$.) So, $M, s \vDash \F C \psi$.
\end{itemize}
Therefore, either $\delta \vDash \phi$ or $\delta \vDash \neg \phi$.
\end{proof}

\DpcModelsUniCan*
\begin{proof}
Let's prove by induction on $k$, and we only need to refine the base case. 
 
Suppose $k = 0$. For every $t \in \TC(s) \cup \{s\}$, let 
\[\zeta_t = \bigwedge_{p \in P \text{ and } M, t \vDash p} p \land \bigwedge_{p \in P \text{ and } M, t \not \vDash p} \neg p\]
Let $w(\delta) = \zeta_s$ and
\[\Phi = \{\zeta_t \mid t \in \TC(s)\}\]
Then, $\delta = w(\delta) \land \nabla \Phi$. Obviously, this $\delta$ is unique w.r.t. $(M, s)$.

The inductive step is the same as that of Proposition \ref{prop: models-unique-canon}. Hence, the proposition is proved.
\end{proof}

\DpcUniCan*
\begin{proof} 
Since we have proved Proposition \ref{prop: dpc-models-unique-canon}, this proof is in parallel with the proof of Proposition \ref{prop: unique-canon}.
\end{proof}

\DpcDistcutinto*
\begin{proof}
The proof is analogous to Proposition \ref{prop: dist-cut-into} because pruning does not influence common knowledge.
\end{proof}

\DpcUparrowproperty*
\begin{proof}
The proof is analogous to Proposition \ref{prop: uparrow-property} because pruning does not influence common knowledge.
\end{proof}

\ReachSplit*
\begin{proof}
By the definition, $\delta^{\uparrow l} = \delta$ when $l \geq k$, while if $l < k$, then $\delta^{\uparrow l} \in C^P_l$ and $R_{\C B}(\delta^{\uparrow l}) \subseteq C^P_{l - 1}$ for every $\C B \in \PPA$. Therefore, it is sufficient to only prove \[\TC(\delta) = \bigcup_{\C B \in \PPA}\bigcup_{\eta \in R_{\C B}(\delta)} \{w(\eta)\} \cup \TC(\eta).\] Denote the right side of the equation by $\Phi$. Since $\delta$ is satisfiable, there is $(M, s)$ such that $M, s \vDash \delta$. Observe that
\[\TC(s) = \bigcup_{\C B \in \PPA}\bigcup_{t \in R_{\C B}(s)} \{t\} \cup \TC(t)\]
Denote the right side by $T$. It is easy to see that $\TC(s)$ is $\TC(\delta)$-complete while $T$ is $\Phi$-complete. Both $\TC(\delta)$ and $\Phi$ consist of minterms of $P$. Therefore, $\TC(\delta) = \Phi$.
\end{proof}

\DPCGuide*

\begin{proof}
This proof is analogous to that of Theorem \ref{thm: guideline}.
\end{proof}

\DPCReflexive*
\begin{proof}
The proof is in parallel with that of Lemma \ref{lem: reflexive}. We only need to show that $w(\delta) \in \TC(\delta)$. Let $M, s$ be any $\san T_n \F{DPC}$-model of $\delta$. Since it is reflexive, $s \in \TC(s)$. Since $M, s \vDash \nabla \TC(s)$, there is a minterm $\zeta$ in $\TC(\delta)$ such that $M, s \vDash \zeta$. This $\zeta$ is exactly $w(\delta)$.
\end{proof}

\DPCT*
\begin{proof}
The proof is similar to that of Lemma \ref{lem: T}. By Lemma \ref{lem: dpc-KandD}, there is a model $(M^*, s^*)$ such that $M^*, s^* \vDash \delta$ and $(M^*, s^*) \bisim^{col}_p (M', s')$. We construct $(M, s)$ by adding edges $R_i$ to every world of $(M^*, s^*)$ as we did in Lemma \ref{lem: T}. It is easy to see that $(M, s) \bisim^{col}_p (M', s')$. We still need to show that $M, s \vDash \delta$.

Let's prove for every $t$ in $(M, s)$, if $t$ is constructed by some dpc-canonical formula $\eta_t$, then $M, t \vDash \eta_t^{\uparrow l}$ for every $l \geq 0$.
\begin{itemize}
\item When $l = 0$, $\eta_t^{\uparrow 0} = w(\eta_t) \land \nabla \TC(\eta_t)$. We have known $M^*, t^* \vDash w(\eta_t) \land \nabla \TC(\eta_t)$. By Lemma \ref{lem: dpc-reflexive}, $w(\eta_t) \in \TC(\delta)$. So, $M, t \vDash w(\eta_t) \land \nabla \TC(\eta_t)$.

\item Suppose $l \geq 0$. We have known that $M, t \vDash w(\eta_t) \land \nabla \TC(\eta_t)$. For every $u \in R_{\C B}(t)$, if $u \neq t$, by the hypothesis, $M, u \vDash \eta_u^{\uparrow (l - 1)}$; if $u = t$, by Lemma \ref{lem: dpc-reflexive}, $\eta_t^{\uparrow (l - 1)} \in R_{\C B}(\eta_t^{\uparrow l})$ and by the hypothesis, $M, t \vDash \eta_t^{\uparrow (l - 1)}$. Conversely, for every $\eta \in R_{\C B}(\eta_t)$, there is a constructed $u \in R_{\C B}(t)$ such that $M, u \vDash \eta_u^{\uparrow (l - 1)}$ by the hypothesis. Therefore, $M, t \vDash \eta_t^{\uparrow l}$
\end{itemize}
Therefore $M, s \vDash \delta$.

Hence, the lemma is proved.
\end{proof}

\ConstructionCommonSfive*
\begin{proof}
This construction inherits the steps of Lemma \ref{lem: ConstructionK45}, and we only need to modify the following steps:

\F{Base case}:
\begin{itemize}
\item[Step 5] For every $\C B \in \C P^+(\C X)$ such that $\C B \neq \C A$ and every $t \in T^*_{\C B}$, consider the multi-pointed submodel 
$(M', \C U^{t'}_{\C B^c})$, where $\C U^{t'}_{\C B^c} = \{R'_{\C D}(t') \mid \C D \in \PP(\C B^c)\}$. Construct $(M^t, \C U^t_{\C B^c})$ as follows: Initially, let $S^t$, $V^t$, $\rho_t$, and every edge be empty.
\begin{itemize}
\item For every $u' \in \TC(t')$ and every minterm $\chi \in \TC(\eta_t)$, if $M', u' \vDash \chi^p$, then construct $u$ in $S$ such that for every $q \in P$, $q \in V(u)$ if and only if $\chi \vDash q$. Let $(u, u') \in \rho_t$.

\item For every $\C C \in \PPA$ and every $u, v$ in $S^t$, if $(u', v') \in R'_{\C C}$, add $(u, v)$ to $R_i$ for $i \in \C C$.

\item For every $\C D \in \PP(\C B^c)$, let $U^{t*}_{\C D} = \{u \mid (u, u') \in \rho_t \text{ and } u' \in R_{\C D}(t')\}$ and 
\[U^t_{\C D} = \bigcup_{\C D \subseteq \C D_0}U^{t*}_{\C D_0}\]
Let $\C U^t_{\C B^c} = \{U^t_{\C D} \mid \C D  \in \PP(\C B^c)\}$.

\item For all the worlds $u'$ in $(M', s')$ that are unreachable from $t'$, add the induced submodel $(M', u')$ to $S^t$ as copies, and this submodel inherits their relations and valuation in $M'$ and $(v, v') \in \rho_t$ for every $v'$ and its copy $v$ in the copied submodel.
\end{itemize}
\end{itemize}
For convenience, let's still call all the worlds in $(M^t, \C U^t_{\C B^c})$ the \emph{copies} of their counterparts in $M'$, and call $(M^t, \C U^t_{\C B^c})$ a copy of $(M', \C U^{t'}_{\C B^c})$. Observe that $(M^t, \C U^t_{\C B^c})$ is still an $\san{S5}_n \F{DPC}$-model and every set $U^t_{\C D}$ is $\C D$-equivalent. Recall that we denote the set $\RT(t, R_i \cup R_i^{- 1})$ within $(M, \C U_{\C X})^{\gamma, \xi, 0, 0}_4$ by $\RT(t, R_i \cup R_i^{- 1})_4$. 
\begin{itemize}
\item[Step 6] For every $\C B \in \C P^+(\C X)$, every $t \in T^*_{\C B}$, every $i \in \C B^c$, every $u \in \RT(t, R_i \cup R_i^{- 1})_4$, $U^u_i$ in $(M^u, \C U^u)$ is an $i$-equivalent set. Then, merge all these sets $U^u_i$ and $\RT(t, R_i \cup R_i^{- 1})_4$.
\end{itemize}

\F{Inductive step}:
\begin{itemize}
\item[Step 2-2] For every $\C B \in \PPX$, every $t \in T^*_{\C B}$, add $(t, t) \in R_j$ for every $j \in \C A$, and subsequently, for every $i \in \PP(\C B^c)$, and every $u_1, u_2 \in R_i(t)$, then add $(u_1, u_2)$ to $R_i$.
\end{itemize}
Thus, the model $(M, \C U_{\C X})$ is constructed.

Observe Step 5 of the base case and we have that for every $\C B \neq \C A$, every $t \in T^*_{\C B}$, and every $u'$ in $(M', \C U^{t'}_{\C B^c})$,
\begin{itemize}
\item if it is unreachable from $t'$, then there is a copy $u \in  S^t$ such that $(u, u') \in \rho_t$;

\item if $u' \in \TC(t')$, since $M', t' \vDash \zeta_t^p$ and $\TC(\eta_t)^p = \TC(\zeta_t)^p$, there is a minterm $\chi \in \TC(\eta_t)$ that constructs a $u$ in $(M^t, \C U^t_{\C B^c})$ such that $M, u \vDash \chi$ and $(u, u') \in \rho_t$.
\end{itemize}
Then, it is easy to see that $\rho_t: (M^t, \C U^t_{\C B^c}) \bisim^{col}_p (M', \C U^{t'}_{\C B^c})$. Then, repeating the proof of Lemma \ref{lem: Common-p-bisim}, we conclude that $(M, \C U_{\C X}) \bisim^{col}_p (M', \C U'_{\C X})$.

Note that Steps 1 and 3 are the same as Lemma \ref{lem: ConstructionK45}, so $T_{\C B}$ is $\C B$-equivalent in $(M, \C U_{\C X})$ for every $\C B \in \PPX$. 

If $(M', s')$ is transitive, Euclidean, and reflexive respectively, $(M^t, \C U^t_{\C B^c})$ is also transitive, Euclidean, and reflexive respectively. Then, repeating the proofs of Lemma \ref{lem: trans-Euc}, we have that $(M, \C U_{\C X})$ is an $\san{S5}_n \F{DPC}$-model.

It remains to prove that $T_{\C B}$ is $R_{\C B}(\delta)^{\uparrow k}$-complete. The proof is also analogous to that of Lemma \ref{lem: Tcomplete} and Lemma \ref{lem: ConstructionS5}. It is augmented by proving the fact that for every constructed world $u$, $M, u \vDash \nabla \TC(\eta_u)$. 

Let's first prove $M, u \vDash \bigwedge \hat{\F{C}}\ \TC(\eta_u)$. That is to prove for every mintern $\chi \in \TC(\eta_u)$, there is $v$ such that $M, v \vDash \chi$. Suppose $\eta_u \in D^P_{m + h}$ for $m \in \{0, \dots, k\}$. Let's prove by induction on $m$
\begin{itemize}
\item \F{Base case} ($m = 0$): Suppose $u$ is an exact $\C B$-son of some $t$.
\begin{itemize}
\item $\C B \neq \C A$: By Step 5 of the base case of constructing $(M^t, \C U^t)$, the submodel $(M^u, \C U^u)$ exists and for every mintern $\chi \in \TC(\eta_u)$, there is $v$ in $(M^u, \C U^u)$ such that $M, v \vDash \chi$.

\item $\C B = \C A$: For every $i \in \C A$, $\eta_u \in R_i(\eta_t)$ and there is $u_0$ in $T^{t*}_i$ that is also constructed by $\eta_u$. Observe that $\TC(u) = \TC(u_0)$. Similar to the case above, there is $v$ in $(M^{u_0}, \C U^{u_0})$ such that $M, v \vDash \chi$.
\end{itemize}

\item \F{Inductive step} ($m > 0$):  Similar to the base case, we assume $u$ is an exact $\C B$-son of some $t$ but $\C B \neq \C A$, and by Step 5, $(M^u, \C U^u)$ exists. For every mintern $\chi \in \TC(\eta_u)$,
\begin{itemize}
\item $\chi$ is $w(\eta)$ for some successor $\eta$ of $\eta_u$: By the inductive construction, there is $v$ constructed by $\eta$ as a son of $u$ in $(M^u, \C U^u)$, such that $M, v \vDash w(\eta)$. Note that $v \in \TC(u)$.

\item $\chi \in \TC(\eta)$ for some successor $\eta$ of $\eta_u$: Note that $\eta \in D^P_{(m - 1) + h}$. By the inductive construction, there is $w$ constructed by $\eta$ as a son of $u$ in the submodel $(M^w, \C U^w)$. By the inductive hypothesis on $m$, there is $v$ in $(M^w, \C U^w)$ such that $M, v \vDash \chi$. Note that $v \in \TC(w) \subseteq \TC(u)$.
\end{itemize}
\end{itemize}
Therefore, for every mintern $\chi \in \TC(\eta_u)$, there is $v \in \TC(u)$ such that $M, v \vDash \chi$. 

We still need to prove $M, u \vDash \F{C}\bigvee\TC(\eta_u)$., that is, for every $v \in \TC(u)$, there is a mintern $\chi \in \TC(\eta_u)$ such that $M, v \vDash \chi$. Let's discuss $v$.
\begin{itemize}
\item $v$ is constructed: Since every constructed world is actually constructed by a dpc-canonical formula that wholly occurs in $\gamma$, then by Proposition \ref{prop: ShareCommonKnowledge}, $\TC(\gamma) = \TC(\eta_u) = \TC(\eta_v)$. Observe that by Proposition \ref{prop: Reach-split}, $w(\eta_v) \in \TC(\gamma)$. So, $w(\eta_v) \in \TC(\eta_u)$ and $M, v \vDash w(\eta_v)$.

\item $v$ is a copy: There is a constructed world $w$ such that $\eta_w \in D^P_{0 + h}$ and $v$ is in the submodel $(M^w, \C U^w)$. By Step 5 of the base case, $v \in \TC(w)$ and there is $\chi \in \TC(\eta_w)$ such that $M, v \vDash \chi$. Since every constructed world is actually constructed by a dpc-canonical formula that wholly occurs in $\gamma$, then by Proposition \ref{prop: ShareCommonKnowledge}, $\chi \in \TC(\eta_w) = \TC(\gamma) = \TC(\eta_u)$.
\end{itemize}
Therefore, for every $v \in \TC(u)$, there is a mintern $\chi \in \TC(\eta_u)$ such that $M, v \vDash \chi$. 

Therefore, for every constructed world $u$, $M, u \vDash \nabla \TC(\eta_u)$. Thus, $T_{\C B}$ is $R_{\C B}(\delta)^{\uparrow k}$-complete. 

Hence, the Construction Lemma for $\san{S5}_n \F{DPC}$ is proved.

\end{proof}


\bibliography{KexuWang}

\end{document}